\pdfoutput=1 
\documentclass[manuscript,screen]{acmart}
\AtBeginDocument{%
  \providecommand\BibTeX{{%
    \normalfont B\kern-0.5em{\scshape i\kern-0.25em b}\kern-0.8em\TeX}}}

\setcopyright{acmcopyright}
\copyrightyear{2018}
\acmYear{2018}
\acmDOI{XXXXXXX.XXXXXXX}

\acmConference[Conference acronym 'XX]{Make sure to enter the correct
  conference title from your rights confirmation emai}{June 03--05,
  2018}{Woodstock, NY}
\acmBooktitle{Woodstock '18: ACM Symposium on Neural Gaze Detection,
 June 03--05, 2018, Woodstock, NY} 
\acmPrice{15.00}
\acmISBN{978-1-4503-XXXX-X/18/06}

\usepackage{hyperref}

\usepackage[utf8]{inputenc}

\usepackage{tikz}
\usetikzlibrary{positioning,calc,arrows,shapes,  tikzmark, decorations.pathreplacing}

\usepackage{virginialake}
\usepackage{tikz-cd}
\usepackage{cleveref}
 \usepackage[textsize=tiny, disable]{todonotes}

\usepackage{caption}
\usepackage{subfig}

 \usepackage{enumerate}

\newcommand{\dfn}{:=}
\renewcommand{\emptyset}{\varnothing}
\renewcommand{\setminus}{-}

\newcommand{\red}[1]{{\color{red}#1}}
\newcommand{\blue}[1]{{\color{blue}#1}}
\newcommand{\purple}[1]{{\color{purple}#1}}
\definecolor{mygreen}{rgb}{0, 0.5, 0}
\newcommand{\green}[1]{{\color{mygreen}#1}}
\newcommand{\orange}[1]{{\color{orange}#1}}

\newcommand{\anupam}[1]{\todo{A: #1}}
\newcommand{\gianluca}[1]{\todo{G: #1}}

\newtheorem{thm}{Theorem}
\newtheorem{lem}[thm]{Lemma}
\newtheorem{prop}[thm]{Proposition}
\newtheorem{cor}[thm]{Corollary}

\newtheorem{fact}[thm]{Fact}

\theoremstyle{definition}
\newtheorem{defn}[thm]{Definition}
\newtheorem{exmp}[thm]{Example}
\newtheorem{conv}[thm]{Convention}

\theoremstyle{remark}
\newtheorem{rem}[thm]{Remark}

\newcommand{\Nat}{\mathbb{N}}

\newcommand{\wfpo}{\vartriangleleft}
\newcommand{\wfpoeq}{\trianglelefteq}

\newcommand{\permpref}{\subset}
\newcommand{\permprefeq}{\subseteq}

\newcommand{\n}{N}

\newcommand{\sq}{\Box}
\newcommand{\sqn}{{\sq\n}}

\newcommand{\sn}{\sqn}
\newcommand{\lists}[3]{{#2},\overset{#1}{\ldots},{#3} }

\newcommand{\ptime}{\mathbf{PTIME}}
\newcommand{\pspace}{\mathbf{PSPACE}}
\newcommand{\elementary}{\mathbf{ELEMENTARY}}

\newcommand{\f}{\mathbf{F}}
\newcommand{\fptime}{\f\ptime}
\newcommand{\fpspace}{\f\pspace}
\newcommand{\felementary}{\f\elementary}
\newcommand{\alogtime}{\mathbf{ALOGTIME}}
\newcommand{\nc}{\mathbf{NC}}
\newcommand{\flinspace}{\mathbf{FLINSPACE}}

\newcommand{\num}[1]{\underline{#1}}
\renewcommand{\succ}[1]{\mathsf{s}_{#1}}
\newcommand{\pred}{\mathsf{p}}
\newcommand{\floor}[1]{\lfloor {#1} \rfloor}
\newcommand{\hlf}[1]{\floor{#1 /2}}
\newcommand{\cnd}{\mathsf{cond}}
\newcommand{\proj}[2]{\pi^{#1}_{#2}}
\newcommand{\s}[1]{\vert {#1}\vert}

\newcommand{\model}[1]{\denot{#1}}
\newcommand{\denot}[1]{f_{#1}}
\newcommand{\ex}{\mathsf{ex}}

\newcommand{\sumlen}[1]{|\hspace{-.15em}|#1|\hspace{-.15em}|}

\newcommand{\sexp}{\mathcal{E}}
\newcommand{\cconc}{\mathcal{C}}

\newcommand{\incrementation}{\mathcal{I}}
\newcommand{\successor}{\mathcal{S}}
\newcommand{\predecessor}{\mathcal{P}}
\newcommand{\completeness}{\mathcal{F}}
\newcommand{\primrec}{\mathcal{R}}

\newcommand{\unbounded}{\mathcal{U}}
\newcommand{\binarytounary}{\mathcal{N}}
\newcommand{\gfunction}{\mathcal{G}}
\newcommand{\hfunction}{\mathcal{H}}

\newcommand{\numeral}[1]{\underline{#1}}
\newcommand{\permi}[1]{\vec{#1}}

\newcommand{\saferec}{\mathsf{srec}}
\newcommand{\srec}{\saferec}
\newcommand{\saferecwfpo}[1]{\saferec_{#1}}
\newcommand{\srecwfpo}[1]{\saferecwfpo{#1}}

\newcommand{\srecpp}{\srecwfpo{\permpref}}

\newcommand{\safenesrec}{\mathsf{snrec}}
\newcommand{\snrec}{\safenesrec}
\newcommand{\safenesrecwfpo}[1]{\safenesrec_{#1}}
\newcommand{\snrecwfpo}[1]{\safenesrecwfpo{#1}}

\newcommand{\snrecpp}{\snrecwfpo{\permpref}}

\newcommand{\ssrecwfpo}[1]{\mathsf{ssrec}_{#1}}
\newcommand{\ssnrecwfpo}[1]{\mathsf{ssnrec}_{#1}}

\newcommand{\ssrecpp}{\ssrecwfpo{\permpref}}
\newcommand{\ssnrecpp}{\ssnrecwfpo{\permpref}}

\newcommand{\bc}{\mathsf{B}}
\newcommand{\bcwfpo}[1]{\bc^{#1}}
\newcommand{\bcnorec}{\bc^{-}}

\newcommand{\nested}{\mathsf{N}}
\newcommand{\nbc}{\nested\bc}
\newcommand{\nbcwfpo}[1]{\nested\bcwfpo{#1}}

\newcommand{\shallow}{\mathsf{S}}
\newcommand{\sbc}{\shallow\bc}

\newcommand{\bcpp}{\bcwfpo{\permpref}}
\newcommand{\nbcpp}{\nbcwfpo{\permpref}}

\newcommand{\seqar}{\Rightarrow}
\newcommand{\der}{\mathcal{D}}
\newcommand{\pder}[1]{\der_{#1}}
\newcommand{\rd}{{R}_\der}
\newcommand{\derrd}{\langle \der, \rd \rangle}
\newcommand{\bud}{\mathit{Bud}}

\newcommand{\close}[1]{{C}_{#1}}
\newcommand{\open}[1]{{O}_{#1}}
 
 \newcommand{\hnu}{\nu_0}

\newcommand{\id}{\mathsf{id}}
\newcommand{\wk}{\mathsf{w}}
\newcommand{\exch}{\mathsf{e}}
\newcommand{\cntr}{\mathsf{c}}
\newcommand{\cut}{\mathsf{cut}}

\newcommand{\boxlef}{\sq_l}
\newcommand{\boxrig}{\sq_r}
\newcommand{\sql}{\boxlef}
\newcommand{\sqr}{\boxrig}
\newcommand{\zero}{0}
\newcommand{\com}{\mathsf{dis}}
\newcommand{\rules}{\mathsf{r}}

\newcommand{\cyclic}{\mathsf{C}}
\newcommand{\cbc}{\cyclic\bc}
\newcommand{\ncbc}{\cyclic\nested\bc}

\newcommand{\slr}{\mathsf{SLR}}

\newcommand{\tgodel}{\mathsf{T}}
\newcommand{\ntgodel}[1]{\tgodel_{#1}}
\newcommand{\ct}{\mathsf{CT}}

\begin{document}

\title{Cyclic Implicit Complexity}

\author{Gianluca Curzi}
\affiliation{%
  \institution{University of Gothenburg}
  \country{Sweden}}
  \affiliation{\institution{University of Birmingham} \country{UK}}
\email{gianluca.curzi@gu.se}
\orcid{57216760760}

\author{Anupam Das}
\affiliation{%
  \institution{University of Birmingham}
  \country{UK}}
\email{A.Das@bham.ac.uk}
\orcid{0000-0002-0142-3676}

\renewcommand{\shortauthors}{G.~Curzi and A.~Das}

\begin{abstract}
  \emph{Circular} (or \emph{cyclic}) proofs have received increasing attention in recent years, and have been proposed as an alternative setting for studying (co)inductive reasoning. 
    In particular, now several type systems based on circular reasoning have been proposed.
    However, little is known about the complexity theoretic aspects of circular proofs, which exhibit sophisticated loop structures atypical of more common `recursion schemes'. 
    
    This paper attempts to bridge the gap between circular proofs and \emph{implicit computational complexity} (ICC). 
    Namely we introduce a circular proof system based on Bellantoni and Cook's famous safe-normal function algebra, and we identify proof theoretical constraints, inspired by ICC, to characterise the polynomial-time and elementary computable functions.
    Along the way we introduce new recursion theoretic implicit characterisations of these classes that may be of interest in their own right.
\end{abstract}




\received{20 February 2007}
\received[revised]{12 March 2009}
\received[accepted]{5 June 2009}

\maketitle

\section{Introduction}
Formal proofs are traditionally seen as finite objects modelling  logical or mathematical reasoning. 
\emph{Non-wellfounded} proofs are a generalisation of this notion to an infinitary (but finitely branching) setting, in which consistency is maintained by a standard global condition: the `progressing' criterion. 
%
Special attention is devoted to \emph{regular} (or \emph{circular} or \emph{cyclic}) proofs, i.e.~those non-wellfounded proofs having only finitely many distinct sub-proofs, and which may thus be represented by finite (possibly cyclic) directed graphs. 
For such proofs the progressing criterion may be effectively decided by reduction to the universality problem for B\"uchi automata.



Non-wellfounded proofs have been employed to  reason about the modal $\mu$-calculus and fixed-point logics~\cite{niwinski1996games,dax2006proof}, first-order inductive definitions~\cite{brotherston2011sequent},  Kleene algebra~\cite{das2017cut,DP18},  linear logic~\cite{baelde2016infinitary,DeSaurin-infinets}, arithmetic~\cite{Simpson17-cyclic-arithmetic,BerardiTatsuta,das2018logical},  system $\tgodel$~\cite{Das2021-preprint,Kuperberg-Pous21, Das2021}, and continuous cut-elimination~\cite{mints1978finite, fortier2013cuts,bouncing-threads}. 
In particular,  \cite{Kuperberg-Pous21} and \cite{Das2021,Das2021-preprint} investigate the computational expressivity of circular proofs, with respect to the proofs-as-programs paradigm, in the setting of higher-order primitive recursion.


However little is known about the complexity-theoretic aspects of circular proofs.
Usual termination arguments for circularly typed programs are nonconstructive, proceeding by contradiction and using a non-recursive `totality' oracle (cf.~\cite{Kuperberg-Pous21,Das2021,Das2021-preprint}). 
As a result, these arguments are not appropriate for delivering feasible complexity bounds (cf.~\cite{das2018logical}).
%

The present paper aims to bridge this gap by proposing a circular foundation for Implicit Computational Complexity (ICC), a branch of computational complexity studying machine-free characterisations of complexity classes. 
Our starting point is  Bellantoni and Cook's famous function algebra $\bc$ characterising the polynomial time computable  functions ($\fptime$) using \emph{safe recursion} \cite{BellantoniCook}. 
The prevailing idea behind safe recursion (and its predecessor, `ramified' recursion \cite{Leivant91}) is to organise data into strata in a way that prevents recursive calls being substituted into recursive parameters of previously defined functions.
This approach has been successfully employed to give resource-bound-free characterisations of polynomial-time \cite{BellantoniCook}, levels of the polynomial-time hierarchy \cite{Bellantoni-fph}, and levels of the Grzegorczyk hierarchy \cite{Wirz99characterizingthe}, and has been extended to higher-order settings too \cite{Hofmann97,Leivant99}.


\subsection*{Circular systems for implicit complexity}
Construing $\bc$ as a type system, we consider non-wellfounded proofs, or \emph{coderivations},   generated by its recursion-free subsystem $\bcnorec$. 
The circular  proof system  $\ncbc$ is then obtained by considering the regular and progressing coderivations of $\bcnorec$ which satisfy a further criterion, \emph{safety}, motivated by the eponymous notion from Bellantoni and Cook's work (cf.~also `ramification' in Leivant's work \cite{Leivant91}). 
On the one hand, regularity and progressiveness ensure that coderivations of $\ncbc$ define total computable functions; on the other hand, the latter criterion ensures that the corresponding equational programs are `safe': the recursive call of a function is never substituted into the recursive parameter of a step function.
    
Despite $\ncbc$ having only ground types, it is able to define equational programs that nest   recursive calls, a property that  typically  arises only in higher-order recursion (cf., e.g., \cite{Hofmann97,Leivant99}).  
In fact, we show that this system defines precisely the elementary computable functions ($\felementary$).
Let us point out that the capacity of circular proofs to simulate some higher-order behaviour reflects an emerging pattern in the literature. 
For instance in~\cite{Das2021,Das2021-preprint} it is  shown that the number-theoretic functions definable by type level $n$  proofs of a circular version of system  $\tgodel$ are exactly those definable by  type level $n+1$ proofs of $\tgodel$. 

In the setting of ICC, Hofmann~\cite{Hofmann97} and Leivant~\cite{Leivant99} already observed that higher-order safe recursion mechanisms can be used to characterise $\felementary$. 
In particular, in~\cite{Hofmann97}, Hofmann presents the type system $\slr$ (Safe Linear Recursion) as a higher-order version of $\bc$  imposing a `linearity' restriction on the higher-order safe recursion operator. 
He shows that this system defines just the polynomial-time computable functions on natural numbers ($\fptime$). 

Inspired by~\cite{Hofmann97}, we too introduce a linearity requirement for $\ncbc$ that is able to control the interplay between cycles and the cut rule, called the \emph{left-leaning criterion}. 
The resulting circular  proof system is called $\cbc$, which we show defines precisely $\fptime$. 
%


\subsection*{Function algebras for safe nested recursion and safe recursion along well-founded relations}
As well as introducing the circular systems $\cbc$ and $\ncbc$ just mentioned, we also develop novel function algebras for $\fptime$ and $\felementary$ that allow us to prove the aforementioned complexity characterisations via a `sandwich' technique, cf.~\Cref{fig:picture-main-results}.
This constitutes a novel (and more direct) approach to reducing circularity to recursion, 
that crucially takes advantage of safety.

We give a relativised formulation of $\bc$, i.e.\ with \emph{oracles}, that allows us to define a form of safe \emph{nested} recursion.
The resulting function algebra is called $\nbc$ and is comparable to the type 2 fragment of Leivant's extension of ramified recursion to finite types~\cite{Leivant99}.
The algebras $\bc$ and $\nbc$ will serve as lower bounds for $\cbc$ and $\ncbc$ respectively.

The relativised formulation of function algebras also admits a robust notion of (safe) recursion along a well-founded relation.
We identify a particular well-founded preorder $\permpref$ (`permutation of prefixes') whose corresponding safe recursor induces algebras $\bcpp$ and $\nbcpp$ that will serve as upper bounds for $\cbc$ and $\ncbc$ respectively.


\subsection*{Outline}
This paper is structured as follows. Section~\ref{sec:preliminaries} presents $\bc$ as a proof system.
In Section~\ref{sec:two-tiered-circular-systems-on-notation} we define non-wellfounded proofs and their semantics
and present the circular proof systems $\ncbc$ and $\cbc$.
In Section~\ref{sec:some-variants} we present the function algebras  $\nbc$, $\bcpp$ and $\nbcpp$. 
Section~\ref{sec:characterization-results-for-function-algebras} shows that $\bcpp$ captures $\fptime$ (\Cref{cor:bcpp-bc-fptime}) and that both $\nbc$ and $\nbcpp$ capture $\felementary$ (Corollary~\ref{cor:nb-elementary-characterization}). 
These results require a delicate Bounding Lemma (Lemma~\ref{lem:boundinglemma}) and an encoding of the elementary functions into $\nbc$ (Theorem~\ref{thm:elementary-in-nbc}). In Section~\ref{sec:completeness} we show that any function definable in $\bc$ is also definable in $\cbc$ (Theorem~\ref{thm:bc-in-cbc}), and that any function definable in $\nbc$ is also definable in $\ncbc$ (Theorem~\ref{thm:nbc-in-ncbc}). 
Finally, in Section~\ref{sec:translation}, we present a translation of $\ncbc$ into $\nbcpp$ that maps $\cbc$ coderivations into $\bcpp$ functions (Lemma~\ref{lem:short-translation-lemma}), by reducing circularity to a form of simultaneous recursion on $\permpref$. 

The main results of this paper are summarised in Figure~\ref{fig:picture-main-results}. 

\begin{figure}[t]
    \centering
    \small
\begin{tikzcd}[column sep=normal]
&&&  \felementary  \arrow[rrrd,bend left=10, leftarrow, pos=0.3, "\text{Theorem}~\ref{thm:fp-soundness}"]&&& \\ 
\nbc \arrow[rrr, "\text{Theorem}~\ref{thm:nbc-in-ncbc}"]\arrow[rrru,bend left=10, leftarrow, pos=0.7, "\text{Theorem}~\ref{thm:elementary-in-nbc}"] && &\ncbc 
  \arrow[rrr,  "\text{Lemma}~\ref{lem:short-translation-lemma}"] &  && \nbcpp   \\
  \bc \arrow[rrr, "\text{Theorem}~\ref{thm:bc-in-cbc}"]\arrow[rrrd,bend right=10, leftrightarrow, dashed, pos=0.7, swap,  "\text{\cite{BellantoniCook}}"] &&& \cbc 
      \arrow[rrr, "\text{Lemma}~\ref{lem:short-translation-lemma}"]& && \bcpp  \\ 
      &&& \fptime \arrow[rrru,bend right=10, leftarrow, pos=0.3, swap, "\text{Theorem}~\ref{thm:fp-soundness}"]&& &
\end{tikzcd}
    \caption{Summary of the main results of the paper, where $\rightarrow$ indicates an inclusion ($\subseteq$) of function classes. 
    }
    \label{fig:picture-main-results}
\end{figure}

\subsection*{Comparison with conference version}
This article is a full version of the conference paper \cite{CurziDas:lics22}, which appeared in the proceedings of \emph{Logic in Computer Science 2022}. 
The current version expands upon \cite{CurziDas:lics22} by including full proofs, as well as many further examples and remarks to aid the narrative.
\section{Preliminaries}\label{sec:preliminaries}

Bellantoni and Cook introduced in \cite{BellantoniCook} an algebra of functions based on a simple two-sorted structure. This idea was itself inspired by Leivant's characterisations, one of the founding works in Implicit Computational Complexity (ICC) \cite{Leivant91}.
The resulting `tiering' of the underlying sorts has been a recurring theme in the ICC literature since, and so it is this structure that shall form the basis of the systems we consider in this work.

We consider functions on the natural numbers with inputs distinguished into two sorts: `safe' and `normal'.
We shall write functions explicitly indicating inputs, namely writing $f(x_1, \dots, x_{m}; y_1, \dots, y_{n})$ when $f$ takes $m$ normal inputs $\vec x$ and $n$ safe inputs $\vec y$.
Both sorts vary over the natural numbers, but their roles will be distinguished by the closure operations of the algebras and rules of the systems we consider. We will sometimes call these functions \emph{safe-normal}.

Throughout this work, we write $|x|$ for the length of $x$ (in binary notation), and if $\vec x = x_1, \dots, x_m$ we write $|\vec x|$ for the list $|x_1|,\dots, |x_m|$.

\subsection{Bellantoni-Cook characterisation of $\fptime$}
We first recall Bellantoni-Cook in its original guise.

\begin{defn}
[Bellantoni-Cook algebra]\label{defn:bc-function-algebra}
$\bc$ is defined as the smallest class of (two-sorted) functions containing,
\begin{itemize}
    \item $0(;) := 0 \in \Nat$.
    \item $\proj {m;n} {j;} (x_1, \dots, x_{m}; y_1, \dots, y_{n}) := x_j$, for $1 \leq j\leq m$.
    \item $\proj {m;n} {;j} (x_1, \dots, x_{m}; y_1, \dots, y_{n}) := y_j$, for $1 \leq j\leq n$.
    \item $\succ i (;x):= 2x+i$, for $i \in \{0,1\}$.
    \item $\pred (;x) := \hlf x$.
    \item $\cnd(;w,x,y,z) := \begin{cases}
    x & w = 0 \\
    y & w = 0 \mod 2, w\neq 0 \\
    z & w = 1 \mod 2
    \end{cases}$
\end{itemize}
and closed under the following:
\begin{itemize}
    \item (Safe composition) 
    \begin{itemize}
        \item If $f(\vec x, x; \vec y)\, ,\, g(\vec x;) \in \bc$ then $f(\vec x, g(\vec x;);\vec y) \in \bc$.
        \item If $f(\vec x; \vec y, y) \, , \, g(\vec x; \vec y) \in \bc$ then $f(\vec x; \vec y, g(\vec x; \vec y)) \in \bc$.
    \end{itemize}
     \item (Safe recursion on notation)
    If $g(\vec x;\vec y) \, , \, h_i (x, \vec x; \vec y, y) \in \bc$ for $i=0,1$ then so is $f(x, \vec x;\vec y) $ given by:
    \[
    \arraycolsep=2pt
    \begin{array}{rcll}
        f(0,\vec x; \vec y) & := & g(\vec x;\vec y) \\
        f(\succ 0 x, \vec x; \vec y) & := & h_0 (x, \vec x; \vec y, f(x, \vec x; \vec y)) & \text{if $x \neq 0$}\\
        f(\succ 1 x, \vec x; \vec y) & := & h_1 (x, \vec x; \vec y, f(x, \vec x; \vec y))
    \end{array}
    \]
\end{itemize}
\end{defn}
Intuitively, in a function $f(\vec{x}; \vec y)\in \bc$ only the  normal arguments $\vec x$  can be used as recursive parameters. The idea  behind safe recursion  is that recursive calls can only appear in safe position, and hence they can never be used as recursive parameters of other previously defined functions. Safe composition preserves the distinction between normal and safe arguments by requiring that, when composing  along a normal parameter,  the pre-composing function has no safe parameter at all. As a result, we can effectively substitute a normal parameter into a safe position but not vice-versa.

Writing $\fptime$ for the class of functions computable in polynomial-time, the main result of Bellantoni and Cook is:
\begin{thm}
[\cite{BellantoniCook}] \label{thm:bellantoni}
$f(\vec x;) \in \bc $ if and only if $ f(\vec x) \in \fptime$.
\end{thm}

\subsection{Proof theoretic presentation of Bellantoni-Cook}
We shall work with a formulation of Bellantoni and Cook's algebra as a type system with modalities to distinguish the two sorts (similarly to \cite{Hofmann97}). 
In order to facilitate the definition of the circular  system that we present later, we here work with sequent-style typing derivations. 

We only consider \emph{types} (or \emph{formulas}) $\n$ (`safe') and $\sq \n$ (`normal') which intuitively vary over the natural numbers.
We write $A,B,$ etc.\ to vary over types.

A \emph{sequent} is an expression $\Gamma \seqar A$, where $\Gamma$ is a list of types (called the \emph{context} or \emph{antecedent}) and $A$ is a type (called the \emph{succedent}). For a list of types $\Gamma = \lists{k}{\n}{\n}$, we write $\sq\Gamma$ for $\lists{k}{\sn}{\sn}$.

In what follows, we shall essentially identify $\bc$ with the $S4$-style type system in \Cref{fig:bc-type-system}.
The colouring of type occurrences may be ignored for now, they will become relevant in the next section.
Derivations in this system are simply called \emph{$\bc$-derivations}, and will be denoted $\der, \sexp, \ldots$. We write $\der: \Gamma \seqar A$ if the derivation $\der$ has end-sequent  $ \Gamma \seqar A$. We may   write $\der= \rules (\der_1, \ldots, \der_n)$  
(for $ n\leq 3$)
if $\rules$ is the last inference step of $\der$ whose immediate subderivations are, respectively,  $\der_1, \ldots, \der_n$. 
Unless otherwise indicated, we assume that the $\rules$-instance is as typeset in \Cref{fig:bc-type-system}.

\begin{figure*}
\[
\hspace{1.3cm}
   \def\arraystretch{3}
   \begin{array}{c}
\vlinf{\id}{}{\n \seqar \n}{}
\quad 
\vliinf{\cut_\n}{}{\orange \Gamma \seqar B}{\orange \Gamma\seqar \n}{\orange \Gamma, \blue \n \seqar B}
\quad
\vliinf{\cut_\sq}{}{\orange \Gamma \seqar B}{\orange \Gamma\seqar \sn}{\blue{\sn},\orange \Gamma \seqar B}
\quad
\vlinf{\wk_\n}{}{\orange\Gamma, \blue \n \seqar B}{\orange \Gamma \seqar B}
\quad
\vlinf{\wk_{\sq}}{}{\blue{\sn},\orange\Gamma \seqar B}{\orange \Gamma \seqar B}
\quad
\vlinf\exch{}{\orange \Gamma, \blue B,\red A, \purple{\Gamma'} \seqar C}{\orange \Gamma, \red A, \blue B, \purple{\Gamma'} \seqar C}
\\
\vlinf{\sql}{}{\blue \sn, \orange \Gamma \seqar A}{\orange \Gamma, \blue \n \seqar A}
\quad
\vlinf{\sqr}{}{\orange {\sq \Gamma} \seqar \sn}{ \orange {\sq \Gamma} \seqar \n}
\quad
\vlinf{\zero}{}{\seqar \n}{}
\quad
\vlinf{\succ 0 }{}{\orange \Gamma \seqar A}{\orange \Gamma \seqar A}
\quad
\vlinf{\succ 1 }{}{\orange \Gamma \seqar A}{\orange \Gamma \seqar A}
\quad 
\vliiinf{\saferec}{}{\blue\sn, \orange \Gamma \seqar \n}{\orange \Gamma \seqar \n}{\blue\sn, \orange \Gamma  , \red \n \seqar \n}{\blue\sn, \orange \Gamma, \red \n \seqar \n}
\\ 
\vliiinf{\cnd_\n}{}{\orange \Gamma, \blue \n \seqar \n}{\orange \Gamma \seqar \n}{\orange \Gamma, \blue \n \seqar \n}{\orange \Gamma , \blue \n \seqar \n}
\qquad
\vliiinf{\cnd_\sq}{}{\blue{\sn},\orange \Gamma\seqar \n}{\orange \Gamma \seqar \n}{\blue{\sn},\orange \Gamma \seqar \n}{\blue{\sn}, \orange \Gamma  \seqar \n}
\end{array}
\]
    \caption{System $\bc$, as a sequent-style type system.}
    \label{fig:bc-type-system}
\end{figure*}

\begin{conv}
[Left normal, right safe]
In what follows, we  assume that sequents have shape ${\lists{}{\sn}{\sn}, \lists{}{\n}{\n} \seqar A}$, i.e.\ in the LHS all $\sn$ occurrences are placed before all $\n$ occurrences.
Note that this invariant is maintained by the typing rules of \Cref{fig:bc-type-system}, as long as we insist that $A=B$ in the exchange rule $\exch$.
This effectively means that exchange steps have one of the following two forms:
\[
\vlinf{\exch_\n}{}{\orange \Gamma, \blue \n, \red \n, \purple{\vec \n'} \seqar A }{\orange \Gamma, \red \n, \blue \n, \purple{\vec \n'}\seqar A}
\qquad
\vlinf{\exch_\sq}{}{\orange{\sq \vec \n} , \blue \sn , \red \sn, \purple{\Gamma'} \seqar A}{\orange{\sq \vec \n}, \red \sn, \blue \sn, \purple{\Gamma'} \seqar A}
\]
Let us point out that this convention does not change the class of definable functions with only normal inputs, under the semantics we are about to give.

{Henceforth, we shall often refrain from indicating explicit instances of the exchange rule when constructing derivations.}
\end{conv}

We construe the system of $\bc$-derivations as a class of safe-normal functions by identifying each rule instance as an operation on safe-normal functions.
Formally:

\begin{defn}
[Semantics]
\label{defn:semantics-bc}
Given a $\bc$-derivation $\der$ with conclusion  $\lists{m}{\sn}{\sn}, \lists{n}{\n}{\n} \seqar A$ 
we define a two-sorted function $\denot \der (x_1, \dots, x_m;y_1, \dots, y_n)$ by induction on the structure of $\der$ as follows (all rules as typeset in \Cref{fig:bc-type-system}):
\begin{itemize}
    \item If $\der = \id$ then $\denot{\der} (;y) \dfn y $.
    \item If $\der =\wk_\n(\der_0)$ then  $\denot {\der} (\vec x; \vec y, y) \dfn \denot{\der_0} (\vec x;\vec y)$.
    \item If $\der=\wk_\sq(\der_0)$ then $\denot {\der} (x, \vec x; \vec y) \dfn \denot{\der_0} (\vec x;\vec y)$.
    \item If $\der=\exch_\n(\der_0)$ then $\denot {\der} (\vec x; \vec y, y,y', \vec y') \dfn \denot{\der_0} (\vec x; \vec y, y',y, \vec y')$.
    \item If $\der=\exch_\sq(\der_0)$ then $\denot {\der} (\vec x, x, x', \vec x'; \vec y) \dfn \denot{\der_0} (\vec x, x',x , \vec x'; \vec y)$.
    \item If $\der=\sql(\der_0)$ then $\denot { \der} (x, \vec x; \vec y) \dfn \denot{\der_0}(\vec x; \vec y, x)$.
    \item If $\der=\sqr(\der_0)$ then $\denot {\der} (\vec x;) \dfn \denot{\der_0} (\vec x;)$.
    \item If $\der=0$ then $\denot {\der}(;) \dfn 0$.
    \item If $\der=\succ 0(\der_0)$ then $\denot {\der} (\vec x;\vec y) \dfn \succ 0 (; \denot{\der_0}(\vec x;\vec y))$.
    \item  If $\der=\succ 1(\der_0)$ then $\denot {\der} (\vec x;\vec y) \dfn \succ 1 (; \denot{\der_0}(\vec x;\vec y))$.
    \item If $\der = \cut_\n(\der_0, \der_1)$ then $\denot {\der} (\vec x;\vec y) \dfn \denot{\der_1}(\vec x; \vec y, \denot{\der_0} (\vec x; \vec y))$.
    \item If $\der = \cut_\sq(\der_0, \der_1)$ then $\denot {\der} (\vec x;\vec y) \dfn \denot{\der_1}(\denot{\der_0} (\vec x; \vec y), \vec x; \vec y)$.
    \item If $\der=\cnd_\n(\der_0, \der_1, \der_2)$ then:
    \[
        \arraycolsep=2pt
    \begin{array}{rcll}
         \denot {\der} (\vec x;\vec y,0) & \dfn &  \denot{\der_0} (\vec x;\vec y) \\
         \denot {\der} (\vec x; \vec y, \succ 0 y) & \dfn & \denot{\der_1} (\vec x; \vec y, y)  & \text{if $y\neq 0$} \\
         \denot {\der} (\vec x; \vec y, \succ 1 y) & \dfn & \denot{\der_2} (\vec x; \vec y, y) 
    \end{array}
    \]
    \item  If $\der=\cnd_\sq(\der_0, \der_1, \der_2)$  then:
    \[
            \arraycolsep=2pt
    \begin{array}{rcll}
         \denot \der (0, \vec x;\vec y) & \dfn &  \denot{\der_0} (\vec x;\vec y) \\
         \denot \der (\succ 0 x, \vec x; \vec y) & \dfn & \denot{\der_1} (x, \vec x; \vec y) & \text{if $x \neq 0$} \\
         \denot \der (\succ 1 x , \vec x; \vec y) & \dfn & \denot{\der_2} (x, \vec x; \vec y) 
    \end{array}
    \]
    \item If $\der =\srec(\der_0, \der_1, \der_2)$ then:
    \[
     \arraycolsep=2pt
    \begin{array}{rcll}
         \denot \der (0, \vec x;\vec y) & \dfn &  \denot{\der_0} (\vec x;\vec y) \\
         \denot \der (\succ 0 x, \vec x; \vec y) & \dfn & \denot{\der_1} (x, \vec x; \vec y, \denot\der(x, \vec x; \vec y)) & \text{if $x \neq 0$} \\
         \denot \der (\succ 1 x , \vec x; \vec y) & \dfn & \denot{\der_2} (x, \vec x; \vec y,  \denot\der(x, \vec x; \vec y)) 
    \end{array}
    \]
\end{itemize}
\end{defn}

This formal semantics exposes how $\bc$-derivations and $\bc$ functions relate. The rule $\saferec$ in~\Cref{fig:bc-type-system} corresponds to safe recursion, and safe composition along safe parameters is expressed by means of the rules $\cut_\n$. 
Note, however, that the function $f_\der$ is not quite defined according to function algebra $\bc$, due to the  interpretation of the $\cut_\sq$ rule apparently not satisfying the required constraint on safe composition along a normal parameter.
However, this admission turns out to
be harmless, as exposited in the following proposition:

\begin{prop}\label{prop:modal-succedent-derivations}
 Given a $\bc$-derivation $\der: \sq \Gamma , \vec \n \seqar \sn$, there is a $\bc$-derivation   $\der^*:\sq\Gamma \seqar \sn$ no larger than $\der$ such that, for any $\vec x$, $\vec y$:
\begin{equation*}
 \denot\der(\vec{x}; \vec{y})= \denot{\der^*}(\vec{x};).
\end{equation*}
\end{prop} 
\begin{proof}
The proof is by induction on the size of the derivation. The case where the last rule of $\der$ is an instance of $\id, \zero, \sqr, \cnd_{\n}, \cnd_\sq$, or $\srec$ holds vacuously. If  the last rule of $\der$ is an instance of $\exch_{\n},\exch_{\sq},  \sql,\succ i$, and $\wk_\sq$ then we apply the induction hypothesis. Let us now suppose that $\der$ has been obtained from a derivation $\der_0$ by applying an instance of  $\wk_\n$. By induction hypothesis, there exists a derivation  $\pder0^*:\lists{n}{\sn}{\sn} \seqar \sn$ such that $\model{\pder0}(\vec x; \vec y)= \model{\pder0^*}(\vec x; )$. Since $\model{\der}(\vec x; \vec y, y)= \model{\pder0}(\vec x; \vec y)=\model{\pder0^*}(\vec x; )$ we just set $\der^*= \pder{0}^*$.  Suppose now that $\der$ is obtained from two derivations $\der_0$ and $\der_1$ by applying an instance of  $\cut_\n$. By induction hypothesis, there exists $\pder{1}^*$ such that $\model{\pder{1}}(\vec x;\vec y,y)=\model{\pder{1}^*}(\vec x;)$. Since 
 $\model{\pder{1}}(\vec x; \vec y,  \model{\pder{0}}(\vec x; \vec y))= \model{\pder{1}^*}(\vec x;)$, we set $\der^*=\pder{1}^*$. As for the case where the last rule is $\cut_\sq$, by induction hypothesis, there exist derivations $\pder{0}^*$ and $\pder{1}^*$ such that $\model{\pder{0}}(\vec x; \vec y)= \model{\pder{0}^*}(\vec x;)$ and $\model{\pder{1}}(\vec x,x; \vec y)= \model{\pder{1}^*}(\vec x, x;)$, so that we define $\der^*$ as the derivation obtained from $\pder{0}^*$ and $\pder{1}^*$ by applying the rule $\cut_{\sq}$.
\end{proof}



Our overloading of the notation $\bc$, for both a function algebra and for a type system, is
now justified by:

\begin{prop}
\label{prop:bc-type-system-characterisation}
 $f(\vec x;\vec y)\in \bc$ iff there is a $\bc$-derivation $\der$ for which $\denot \der (\vec x; \vec y) = f(\vec x;\vec y)$.
\end{prop}
\begin{proof}
Let us first prove the left-right implication by induction on $f \in \bc$. The cases where $f$ is $\zero$ or $\succ i$ are straightforward. If $f$ is a projection then we construct $\der$ using the rules $\id$, $\wk_{\n}$, $\wk_\sq$, $\exch_\n$, and $\exch_{\sq}$. If $f=\pred(; x)$ then $\der$ is as follows: 
\[
\small
\vliiinf{\cnd_\n}{}{\n \seqar \n}{\vlinf{\zero}{}{\seqar \n}{\vlhy{\ }}}{\vlinf{\id}{}{\n \seqar \n}{\vlhy{\ }}}{\vlinf{\id}{}{\n \seqar \n}{\vlhy{\ }}}
\]
If $f$ is $\cnd(;x,y,z,w)$ then $\der$ is constructed using $\id$, $\wk_\n$, $\exch_\n$ and $\cnd_\n$.
Suppose now that $f(\vec x; \vec y)= h(\vec x, g(\vec x;); \vec y)$. Then $\der$ is as follows:
\[
\small
\vlderivation{
\vliin{\cut_\sq}{}
{\sq \vec \n, \vec \n \seqar \n}
{\vlin{\sqr}{}{\sq \vec \n\seqar \red{\sn}}{\vltr{\der_0}{\sq \vec \n\seqar \n}{\vlhy{\ }}{\vlhy{\ }}{\vlhy{\ }}}}
{\vltr{\der_1}{ \sq \vec \n, \red{\sn}, \vec \n \seqar \n}{\vlhy{\ }}{\vlhy{\ }}{\vlhy{\ }}}
}
\]
where $\der_0$ and $\der_1$ are such that $\denot {\der_0}=g$ and $\denot{\der_1}=h$. If $f(\vec x; \vec y)= h(\vec x; \vec y, g(\vec x; \vec y))$, then $\der$ is as follows:
\[
\small
\vlderivation{
\vliin{\cut_\n}{}
{\sq \vec \n, \vec \n \seqar \n}
{\vltr{\der_0}{\sq \vec \n, \vec \n \seqar \red{\n}}{\vlhy{\ }}{\vlhy{\ }}{\vlhy{\ }}}
{\vltr{\der_1}{ \sq \vec \n, \vec \n, \red{\n} \seqar \n}{\vlhy{\ }}{\vlhy{\ }}{\vlhy{\ }}}
}
\]
where $\der_0$ and $\der_1$ are such that $\denot {\der_0}=g$ and $\denot{\der_1}=h$. Last, suppose that $f(x, \vec x; \vec y)$ has been obtained by safe recursion from $g(\vec x; \vec y)$ and $h_i(x, \vec x; \vec y, y)$ with $i=0,1$. Then,  $\der$ is as follows:
\[
\small
\vlderivation{
\vliiin{\srec}{}
{\sn, \sq \vec \n, \vec \n \seqar \n}
{\vltr{\der_{0}}{\sq \vec \n, \vec \n \seqar \n}{\vlhy{\ }}{\vlhy{\ }}{\vlhy{\ }}}
 {\vltr{\der_{1}}{\sn, \sq \vec \n, \vec \n, \n   \seqar \n}{\vlhy{\ }}{\vlhy{\ }}{\vlhy{\ }}}
 {\vltr{\der_{2}}{\sn, \sq \vec \n, \vec \n, \n   \seqar \n}{\vlhy{\ }}{\vlhy{\ }}{\vlhy{\ }}}
}
\]
where $\pder0$, $\pder1$, and $\pder2$ are such that $\denot{\der_0}=g$, $\denot{\der_1}=h_0$ and $\denot{\der_2}=h_1$. 

For the right-left implication, we prove by  induction on the size of derivations that $\der: \lists{n}{\sn}{\sn},\lists{m}{\n}{\n} \seqar C$ implies $\denot \der\in \bc$. The cases where the last rule of $\der$ is $\id, \zero, \succ i, \wk_{\n}, \wk_\sq, \exch_\n,\exch_\sq, \sql, \sqr$ are all straightforward using constants, successors and projections. If $\der$ has been obtained from two derivations $\pder0$ and $\pder1$ by applying an instance of $\cut_\n$ then, by applying the induction hypothesis,  $\denot \der(\vec x; \vec y)= \denot{\pder1}(\vec x; \vec y, \denot{\pder0}(\vec x; \vec y))\in \bc$. As for the case where the last rule is $\cut_\sq$, by  Proposition~\ref{prop:modal-succedent-derivations}  there exists a derivation $\pder0^*$ with size no larger than $\der_0$ such that  $\denot{\pder0}(\vec x;\vec y)=\denot{\pder0^*}(\vec x;)$. By applying the induction hypothesis, we have $\denot\der(\vec x; \vec y)= \denot{\pder1}(\vec x, \denot{\pder0^*}(\vec x;); \vec y)\in \bc$. If $\der$ has been obtained from derivations $\pder0$, $\pder1$, $\pder2$ by applying an instance of  $\cnd_\n$ then, by using the induction hypothesis,  $\denot\der(\vec x; \vec y, y)= \cnd(; y,\denot{\pder0}(\vec x; \vec y), \denot{\pder1}(\vec x; \vec y,\pred(;y)), \denot{\pder2}(\vec x; \vec y,\pred(;y)) )\in \bc$. If last rule is $\cnd_{\sq}$, by  induction hypothesis, 
$\denot\der(x,\vec x; \vec y)= \cnd(;  x, \denot{\pder0}(\vec x; \vec y), \denot{\pder1}(\pred(x;),\vec x; \vec y), \denot{\pder2}(\pred(x;),\vec x; \vec y))$, where    $\pred(x;)= \pred(;\proj{1;0}{1;}(x; )) $, and so $\denot\der(x,\vec x; \vec y)\in \bc$. Last, if $\der$ has been obtained from derivations $\pder0$, $\pder1$, $\pder2$ by applying an instance of $\srec$, then  $\denot\der\in \bc$  using the induction hypothesis and the safe recursion scheme.
\end{proof}


\begin{conv}
Given \Cref{prop:bc-type-system-characterisation} above, 
we shall henceforth freely write $f(\vec x ; \vec y) \in \bc$ if there is a derivation $\der:\sq\Gamma, \vec \n \seqar \n$ with $\denot \der (\vec x; \vec y) = f(\vec x;\vec y) $.
\end{conv}

\section{Two-sorted circular systems on notation}\label{sec:two-tiered-circular-systems-on-notation}

In this section we  introduce a `coinductive' version of $\bc$,  and we study global  criteria that tame its computational strength. 
This proof-theoretic investigation will lead us to two relevant circular systems: $\ncbc$, which morally permits  `nested' versions of safe recursion, and $\cbc$, which will turn out to be closer to usual safe recursion.

Throughout this section we shall work with the set of typing rules $\bcnorec \dfn \bc \setminus \{\saferec\}$.

\subsection{Non-wellfounded typing derivations}
To begin with, we define the notion of `coderivation', which is the fundamental formal object of this section.

\begin{defn}
[Coderivations]
A ($\bcnorec$-)\emph{coderivation} $\der$ is a possibly infinite {rooted} tree (of height $\leq \omega$) generated by the rules of $\bcnorec$.
Formally, we identify $\der$ with a prefix-closed subset of $ \{0,1,2\}^*$ (i.e.\ a ternary tree) where each node  is labelled by an inference step from $\bcnorec$ such that, whenever $\nu\in \der$ is labelled by a step $\vliiinf{}{}{S}{S_1}{\cdots}{S_n}$, for $n\leq 3$, $\nu$ has $n$ children in $\der$ labelled by steps with conclusions $S_1, \dots, S_n$ respectively. Sub-coderivations of a coderivation $\der$ rooted at position $\nu \in \{0,1,2 \}^*$ are typically denoted $\pder\nu$, so that $\pder\epsilon=\der$. 
We write  $\nu \sqsubseteq \mu$ (or  ~$\nu \sqsubset \mu$) if $\nu$ is a prefix (respectively, a strict prefix) of $\mu$, and in this case we say that $\mu$ is \emph{above} (respectively, \emph{strictly above}) $\nu$ or that $\nu$ is \emph{below} (respectively, \emph{strictly below}) $\mu$. 
\anupam{this is unclear to me: check if this is ever used, if not delete. }

We say that a coderivation is \emph{regular} (or \emph{circular}) if it has only finitely many distinct sub-coderivations.
\end{defn}

Note that, while usual derivations may be naturally written as finite trees or dags, regular coderivations may be naturally written as finite directed (possibly cyclic) graphs.
Some examples of regular coderivations can be found in~\Cref{fig:examples-regular-coderivations}, employing the following writing conventions:

\begin{conv}[Representing coderivations]\label{conv:convention}
Henceforth, we may mark steps by $\bullet$ (or similar) in a regular coderivation to indicate roots of identical sub-coderivations. Moreover, to avoid ambiguities and  to ease parsing of (co)derivations, we shall often underline principal formulas of a  rule instance in a given coderivation and omit instances of $\wk_{\sq}$ and $\wk_\n$ as well as certain structural steps, e.g.\ during a cut step. 

Finally, when the sub-coderivations $\der_0$ and $\der_1$ above the second and the third premise of the conditional rule (from left) are similar (or identical), we may compress them into a single parametrised sub-coderivation $\der_i$ (for $i=0,1$).
\end{conv}

\begin{figure*}
    \centering
\[
\small
\vlderivation{
\vliin{\cut_{\sq}}{\bullet}{\blue{\sn} \seqar \n}{
    \vlin{ \sqr}{}{\blue{\sn} \seqar \sn}{\vlin{\sql}{}{\blue{\sn} \seqar \n}{\vlin{ \succ 1}{}{\n \seqar \n}{\vlin{\id}{}{\n \seqar \n}{\vlhy{}}}}}
}{
   \vlin{\cut_{\sq}}{\bullet}{\sn \seqar \n}{\vlhy{\vdots}}
}
}
\qquad \qquad 
\vlderivation{
\vliin{\cnd_\sq}{\bullet}{\blue{\underline{\sn}}, \sq \vec \n \seqar \n}
{
\vltr{\gfunction}{\sq \vec \n \seqar \n}{\vlhy{\ }}{\vlhy{\ }}{\vlhy{\ }}
}
{
  \vliin{\cut_\sq}{{\scriptstyle{i=0,1}}}{\blue{\sn}, \sq \vec \n \seqar \n}
  {
  \vlin{\sqr}{}{\blue{\sn}, \sq \vec \n  \seqar \red{\sn}}{
    \vlin{\cnd_{\sq}}{\bullet}{\blue{\sn}, \sq \vec \n  \seqar {\n}}{\vlhy{\vdots}}
  }
  }{
  \vltr{\hfunction_i}{\blue{\sn}, \sq \vec \n, \red{\sn}  \seqar \n}{\vlhy{\ }}{\vlhy{\ }}{\vlhy{\ }}
  }
}
{
}
}
\]
\[
\small
\vlderivation{
\vliin{\cnd_\sq}{\bullet}{\blue{\underline{\sn}}, \red{\sn} , \n \seqar \n }
 {
   \vliin{\cnd_\sq}{\circ}{\red{\underline{\sn}}, \n \seqar \n }
   {
    \vlin{\id}{}{\n \seqar \n}{\vlhy{}}
   }
   {
   \vlin{\succ i}{{\scriptstyle{i=0,1}}}{\red{\sn},  \n \seqar \n}{\vlin{\cnd_\sq}{\circ}{ \red{\sn}, \n \seqar \n}{\vlhy{\vdots}}}
   }
 }
 {
  \vlin{\succ i}{{\scriptstyle{i=0,1}}}{\blue{\sn}, \red{\sn}, \n \seqar \n}{\vlin{\cnd_\sq}{\bullet}{\blue{\sn}, \red{\sn}, \n \seqar \n}{\vlhy{\vdots}}}
 }
}
\]
    \caption{Examples of regular coderivations $\incrementation$, $\primrec$ and $ \cconc$, from left (assuming $\gfunction$, $\hfunction_0$ and $\hfunction_1$ regular).}
    \label{fig:examples-regular-coderivations}
\end{figure*}

As discussed in~\cite{Das2021,Das2021-preprint,Kuperberg-Pous21}, coderivations can be identified with Kleene-Herbrand-G\"odel style equational programs, in general computing partial recursive functionals  (see, e.g., \cite[\S 63]{Kleene71:intro-to-metamath} for further details).
We shall specialise this idea to our two-sorted setting.

\begin{defn}
[Semantics of coderivations]
\label{defn:semantics-coderivations}
To each $\bcnorec$-coderivation $\der$ we associate a two-sorted Kleene-Herbrand-G\"odel partial function $\denot \der$ obtained by construing the semantics of \Cref{defn:semantics-bc} as a (possibly infinite) equational program.
Given a two-sorted function $f(\vec x; \vec y)$, we say that $f$ is \emph{defined} by a $\bcnorec$-coderivation $\der$ if $\denot\der(\vec x; \vec y)=f(\vec x; \vec y)$. 
\end{defn}

\begin{rem}\label{rem:regular-finite-equation}
Note, in particular, that from a regular coderivation $\der$ we obtain a \emph{finite} equational program determining $\denot \der$.
Of course, our overloading of the notation $\denot \der$ is suggestive since it is consistent with that of \Cref{defn:semantics-bc}.
\end{rem}

\begin{exmp}[Regular coderivations, revisited] Let us consider the semantics of coderivations $\incrementation$, $\primrec$ and $\cconc$ from~\Cref{fig:examples-regular-coderivations}.  
\begin{itemize}
    \item The  partial functions $\denot{\incrementation \nu}$  are given by the following equational program:
\[
 \arraycolsep=2pt
\begin{array}{rcll}
    \denot{\incrementation_{\epsilon}}  (x;) & = & \denot{\incrementation_{1}}(\denot{\incrementation_{0}}(x;);) \\
   \denot{\incrementation_{0}}(x;) & = & \denot{\incrementation_{00}}(x;)\\
   \denot{\incrementation_{00}}(x;) & = & \denot{\incrementation_{000}}(;x)\\
   \denot{\incrementation_{000}}(;x) & = & \succ 1(;x)\\
   \denot{\incrementation_{1}}(x;) & = & \denot{\incrementation_{\epsilon}}(x;)
\end{array}
\]
By purely equational reasoning, we can simplify this program to obtain $\denot{\incrementation_{\epsilon}}(x;)  =  \denot{\incrementation_{\epsilon}}(\succ 1x;)$. 
 Since the above equational program keeps increasing the input, the function $\denot\incrementation=\denot{\incrementation_{\epsilon}}$ is always undefined. 
\item Let $\denot{\gfunction}(\vec x; \vec y)$ and $\denot{\hfunction_i}( x, \vec x, z; \vec y)$ ($i=0,1$) be the functions defined by the  regular  $\bcnorec$-coderivations $\gfunction$ and $\hfunction_i$, respectively. Then the equational program for $\primrec$  can be rewritten as follows:
\begin{equation}
    \label{eq:norm-prim-rec-scheme}
     \arraycolsep=2pt
    \begin{array}{rcll}
 \denot{\primrec}(0, \vec x;) & = & \denot{\gfunction}(\vec x; )\\
  \denot{\primrec}(\succ ix,\vec x; ) & = &\denot{\hfunction_i}( x, \vec x, \denot{\primrec}(x, \vec x; ); )  
\end{array}
\end{equation}
which is an instance of a \emph{non-safe}  recursion scheme (on notation).
    \item The equational program of  $\cconc$ can be written as:
\[
 \arraycolsep=2pt
\begin{array}{rcll}
    \denot{\cconc_{\epsilon}}  (0, 0;z) & = & z\\
     \denot{\cconc_{\epsilon}}  ( 0, \succ i y;z) & = &   \succ i \denot{\cconc_{\epsilon}} ( x, y;z) & \neq 0 \\ 
     \denot{\cconc_{\epsilon}}  (\succ i x, y;z) & = &   \succ i\denot{\cconc_{\epsilon}}  ( x, y;z) & \neq 0 
\end{array}
\]
which computes concatenation of the binary representation of three natural numbers.
\end{itemize}

\end{exmp}
The above examples illustrate two undesirable features of regular $\bcnorec$-coderivations, from the point of view of implicit complexity:
\begin{enumerate}[I.]
    \item \label{enum:problem1} on the one hand, despite being finitely presentable, they can define \emph{partial functions};
    \item \label{enum:problem2} on the other hand, despite the presence of modalities implementing the normal/safe distinction of function arguments, they  can define \emph{non-safe recursion schemes}.
\end{enumerate}

\subsection{The progressing criterion}

To address Problem~\ref{enum:problem1} we shall adapt to our setting a well-known `termination criterion' from non-wellfounded proof theory.
First, let us recall some standard  proof theoretic concepts about (co)derivations, similar to those in \cite{Das2021,Das2021-preprint,Kuperberg-Pous21}.
\begin{defn}
[Ancestry]
\label{defn:ancestry}
Fix a coderivation $\der$. We say that a type occurrence $A$ is an \emph{immediate ancestor} of a type occurrence $B$ in $\der$ if they are types in a premiss and conclusion (respectively) of an inference step and, as typeset in \Cref{fig:bc-type-system}, have the same colour.
If $A$ and $B$ are in some $\orange \Gamma$ or $\purple{ \Gamma'}$, then furthermore they must be in the same position in the list.
\end{defn}

Being a binary relation, immediate ancestry forms a directed graph upon which our correctness criterion is built. 

\gianluca{I guess that the figure with the examples must be expanded for readability.}
\begin{defn}
[Progressing coderivations]
\label{defn:progressing}
A \emph{thread} is a maximal path in the graph of immediate ancestry.
We say that a (infinite) thread is \emph{progressing} if it is eventually constant $\sn$ and infinitely often principal for a $\cnd_{\sq}$ rule.

A coderivation is \emph{progressing} if each of its infinite branches has a progressing thread.
\end{defn}

\begin{exmp}[Regular coderivations, re-revisited] In~\Cref{fig:examples-regular-coderivations},  $\incrementation$ has precisely one infinite branch (that loops on $\bullet$) which contains no instances of  $\cnd_{\sq}$ at all. Therefore, $\incrementation$ is not progressing. 
On the other hand, $\cconc$ has two simple loops, one on $\bullet$ and the other one on $\circ$. For any infinite branch $B$ we have two cases:
\begin{itemize}
    \item if $B$ crosses the bottommost conditional infinitely many times, it contains a progressing \blue{blue} thread;
    \item otherwise, $B$ crosses the topmost conditional infinitely many times, so that it  contains  a progressing \red{red} thread.
\end{itemize}
Therefore, $\cconc$ is progressing. By similar reasoning, we can conclude that $\primrec$ is progressing whenever $\gfunction$ and $\hfunction_i$ are.
\end{exmp}

Like in \cite{Kuperberg-Pous21,Das2021-preprint,Das2021}, the progressing criterion is sufficient to guarantee that the partial function computed is, in fact, a well-defined total function:

\begin{prop}
\label{prop:prog-total}
If $\der$ is progressing then $\denot \der$ is total.
\end{prop}
\begin{proof}
[Proof sketch]
We proceed by contradiction.
If $\denot \der$ is non-total then, since each rule preserves totality top-down, we must have that $\denot{\der'}$ is non-total for one of $\der$'s immediate sub-coderivations $\der'$.
Continuing this reasoning we can build an infinite leftmost `non-total' branch $B=(\der^{i})_{i<\omega}$.
Let $(\sn^i)_{i\geq k}$ be a progressing thread along $B$, and assign to each $\sn^i$ the least natural number $n_i \in \Nat$ such that $\denot{\der^{i}}$ is non-total when $n_i$ is assigned to the type occurrence $\sn^i$.

Now, notice that:
\begin{itemize}
    \item $(n_i)_{i\geq k}$ is monotone non-increasing, by inspection of the rules and their interpretations from \Cref{defn:semantics-bc}.
    \item $(n_i)_{i\geq k}$ does not converge, since $(\sn^i)_{i\geq k}$ is progressing and so is infinitely often principal for $\cnd_\sq$, where the value of $n_i$ must strictly decrease (cf., again, \Cref{defn:semantics-bc}).
\end{itemize}
This contradicts the well-ordering property of the natural numbers.
\end{proof}

One of the most appealing features of the progressing criterion is that, while being rather expressive and admitting many natural programs, e.g.~as we will see in the next subsections, it remains effective (for regular coderivations) thanks to well known arguments in automaton theory:

\begin{fact}[Folklore]
 It is decidable whether a regular coderivation is progressing.
\end{fact}
This well-known result (see, e.g.,~\cite{DaxHL06} for an exposition for a similar circular system) follows from the fact that the progressing criterion is equivalent to the universality of a  B\"{u}chi automaton of size determined by the (finite) representation of the input coderivation. This problem is decidable in polynomial space, though the correctness of this algorithm requires nontrivial infinitary combinatorics, as formally demonstrated in~\cite{KolodziejczykMP19}.

Let us finally observe that the progressing condition turns out to be sufficient to restate~\Cref{prop:modal-succedent-derivations} in the setting of  non-wellfounded coderivations:

\begin{prop} \label{prop:cutbox} Given a progressing  $\bcnorec$-coderivation  $\der: \sq \Gamma , \vec \n \seqar \sn$, there is a progressing  $\bcnorec$-coderivation $\der^*:\sq\Gamma \seqar \sn$   such that, for any $\vec x$, $\vec y$:
\begin{equation*}
    \model{\der}(\vec{x}; \vec{y})= \model{\der^*}(\vec{x};).
\end{equation*}
\end{prop}
\begin{proof} By progressiveness, any infinite branch contains a $\cnd_\sq$-step, which has non-modal succedent. 
Thus there is a set of $(\cnd_{\sq}, 0, \id)$-occurrences  that forms a bar across $\der$.
By K\"{o}nig Lemma, the set of all nodes of $\der$ below this bar, say $X_\der$, is finite. The proof now follows by induction on the size of $X_\der$, and proceeds analogously to~\Cref{prop:modal-succedent-derivations}. If the size is $0$ then the last rule of $\der$ must be an instance of $\cnd_{\sq}$, $0$, or $\id$, in which case the statement holds vacuously. Otherwise, we do case analysis. If the last rule of $\der$ is an instance of $\sqr$ or $\cnd_{\n}$ then the statement holds vacuously. If  the last rule of $\der$ is an instance of $\exch_{\n},\exch_{\sq},  \sql,\succ i$, and $\wk_\sq$ then we apply the induction hypothesis. Let us now suppose that $\der$ has been obtained from a coderivation $\der_0$ by applying an instance of  $\wk_\n$. By induction hypothesis, there exists a coderivation  $\pder0^*:\lists{n}{\sn}{\sn} \seqar \sn$ such that $\model{\pder0}(\vec x; \vec y)= \model{\pder0^*}(\vec x; )$. Since $\model{\der}(\vec x; \vec y, y)= \model{\pder0}(\vec x; \vec y)=\model{\pder0^*}(\vec x; )$ we just set $\der^*= \pder{0}^*$.  Suppose now that $\der$ is obtained from two coderivations $\der_0$ and $\der_1$ by applying an instance of  $\cut_\n$. By induction hypothesis, there exists $\pder{1}^*$ such that $\model{\pder{1}}(\vec x;\vec y,y)=\model{\pder{1}^*}(\vec x;)$. Since 
 $\model{\pder{1}}(\vec x; \vec y,  \model{\pder{0}}(\vec x; \vec y))= \model{\pder{1}^*}(\vec x;)$, we set $\der^*=\pder{1}^*$. As for the case where the last rule is $\cut_\sq$, by induction hypothesis, there exist coderivations $\pder{0}^*$ and $\pder{1}^*$ such that $\model{\pder{0}}(\vec x; \vec y)= \model{\pder{0}^*}(\vec x;)$ and $\model{\pder{1}}(\vec x,x; \vec y)= \model{\pder{1}^*}(\vec x, x;)$, so that we define $\der^*$ as the coderivation obtained from $\pder{0}^*$ and $\pder{1}^*$ by applying the rule $\cut_{\sq}$.
\end{proof}

Note that the above proof only depends on progressiveness, and so it holds for non-regular progressing $\bcnorec$-coderivations as well.

\subsection{Computational expressivity of coderivations} \label{subsec:definability in typing derivations.}

Problem~\ref{enum:problem2} indicates that the modal/non-modal distinction for (progressing) $\bcnorec$-coderivations is somewhat redundant and does not suffice, by itself, to control complexity.  
Indeed, by inspection of~\Cref{fig:bc-type-system} it is not hard to see that we can stepwise replace each $\n$ by $\sn$ in the coderivation while preserving progressiveness, so that all definable functions with normal and safe inputs can be defined by using only normal inputs. 

\begin{prop}\label{prop:modal-nonmodal-distinction}
Let $\der: \lists{n}{\sn}{\sn},  \lists{m}{\n}{\n}\seqar \n$ be a  $\bcnorec$-coderivation. Then,    there exists   a  $\bcnorec$-coderivation $\der^\sq: \lists{n+m}{\sn}{\sn}\seqar \n$    s.t.:
\begin{itemize}
    \item $\denot\der(\vec x; \vec y)= \denot{\der^\sq}(\vec x ,\vec y;)$;
    \item $\der^\sq$ does not contain instances of $\wk_\n$, $\exch_\n$, $\cut_\n$, $\cnd_\n$.
\end{itemize}
Moreover, $\der^\sq$ is regular (resp.~progressing) if $\der$ is.
\end{prop}
\begin{proof}
We construct $\der^\sq$ coinductively. We only consider two  interesting cases,  when $\der$ is $\id$ or it is  obtained from $\pder0:\Gamma \seqar \n$ and $\pder1:\Gamma, \n\seqar A$ by applying a $\cut_\n$-step. Then, $\der^\sq$ is constructed,  respectively, as follows:
\[
\small
\vlderivation{
\vlin{\sql}{}{\sn \seqar \n}{\vlin{\id}{}{\n \seqar \n}{\vlhy{\ }}}}\qquad  \quad 
  \vlderivation{
   \vliin{\cut_\sq}{}{\sq \Gamma \seqar \n}{
        \vlin{\sqr}{}{\sq \Gamma \seqar \sn}{ \vltr{\der_0^\sq}{\sq \Gamma \seqar \n}{\vlhy{\ }}{\vlhy{\ }}{\vlhy{\ }}}
    }{
        \vlin{\sql}{}{\sq \Gamma, \sn\seqar \n}{\vltr{\der_1^\sq}{\sq \Gamma, \n\seqar \n}{\vlhy{\ }}{\vlhy{\ }}{\vlhy{\ }}}
    }
    }
\]
the remaining interesting cases are treated similarly.
\end{proof}

Consequently, we may view (regular, progressing) $\bcnorec$-coderivations as the type 0 (regular, progressing) fragment of the system $\ct$ from \cite{Das2021-preprint,Das2021,Kuperberg-Pous21}.

As a result we inherit the following characterisations:

\begin{prop}
\label{prop:computational-props-coderivations}
We have the following:
\begin{enumerate}
    \item\label{item:type1-complete} Any function $f(\vec x;)$ is defined by a progressing $\bcnorec$-coderivation.
    \item\label{item:turing-complete} The class of regular $\bcnorec$-coderivations is Turing-complete, i.e.\ they define every partial recursive function.
    \item\label{item:type1-prim-rec} $f(\vec x;)$ is defined by a regular progressing $\bcnorec$-coderivation
if and only if 
$f(\vec x)$ is type-1-primitive-recursive, i.e.\ it is in the level 1 fragment $\ntgodel 1$ of G\"odel's $\tgodel$.
\end{enumerate}
\end{prop}

The  proof of the above proposition can be found in~\Cref{sec:computational-strength}.

Given the computationally equivalent system $\ct_0$ with contraction from \cite{Das2021}, we can view the above result as a sort of `contraction admissibility' for regular progressing $\bcnorec$-coderivations. 
Call $\bcnorec+ \{\cntr_\n, \cntr_\sn \}$ the extension of $\bcnorec$ with the rules $\cntr_\n$ and $\cntr_\sn$ below:
\begin{equation} \label{eqn:explicit-contraction}
  \vlinf{\cntr_\n}{}{\orange \Gamma, \blue \n\seqar B}{\orange \Gamma, \blue \n , \blue \n \seqar B}
\qquad 
\vlinf{\cntr_\sq}{}{\blue \sn , \orange \Gamma\seqar B}{\blue \sn,\blue \sn , \orange \Gamma \seqar B}
\end{equation}
 where  the semantics for the new system extends the one for $\bcnorec$ in the obvious way, and the notion of (progressing) thread is induced by the given colouring.\footnote{Note that the totality argument of \Cref{prop:prog-total} still applies in the presence of these rules, cf.~also \cite{Das2021}.} 
We have:
\begin{cor} $f(\vec x; )$ is definable by a regular progressing  $\bcnorec+ \{\cntr_\n, \cntr_\sn \}$-coderivation  if and only if it is definable by a regular progressing $\bcnorec$-coderivation.
\end{cor}

\subsection{Proof-level conditions motivated by implicit complexity}\label{sec:Proof-level conditions motivated by implicit complexity}
$\bc$-derivations  locally introduce safe recursion  by means of the rule $\saferec$, and Proposition~\ref{prop:modal-succedent-derivations} ensures that the composition schemes defined by the cut rules are safe. As suggested by Problem~\ref{enum:problem2}, a different situation arises when we move to   $\bcnorec$-coderivations, where the lack of further constraints 
means that we can define `non-safe' equational programs.
We may recover safety by a natural proof-level condition:


\begin{defn}
[Safety]
\label{defn:safety}
A $\bcnorec$-coderivation is \emph{safe} if each infinite branch crosses only finitely many $\cut_\sq$-steps.
\end{defn}

The corresponding equational programs of safe coderivations indeed only have safe inputs in hereditarily safe positions, as we shall soon see.
Let us illustrate this by means of examples.

\begin{exmp}
The coderivation $\incrementation$ in~\Cref{fig:examples-regular-coderivations} is not safe, as there is an instance of $\cut_{\sq}$ in the loop on $\bullet$, which means that there is an infinite branch crossing infinitely many $\cut_{\sq}$-steps. By contrast,  the coderivation  $\cconc$ is safe because it does not contain instances of the rule $\cut_\sq$.  Finally, by inspecting the coderivation $\primrec$  of~\Cref{fig:examples-regular-coderivations}, we notice that the infinite branch that loops on $\bullet$  contains infinitely many $\cut_{\sq}$ steps, 
so it  is not safe.
\end{exmp}

Perhaps surprisingly, however, the safety condition is not enough to restrict the set of $\bcnorec$-definable functions to $\fptime$, as the following example shows.

\begin{exmp}
[Safe exponentiation]
\label{exmp:safe-exponentiation}
Consider the following coderivation $\sexp$,
\[
\vlderivation{
    \vliin{\cnd_{\sq}}{\bullet}{\blue{\sn}, \orange \n \seqar \n}{
        \vlin{\succ{0}}{}{\orange \n \seqar \n}{
        \vlin{\id}{}{\orange \n \seqar \n}{\vlhy{}}
        }
    }{
        \vliin{\cut_\n}{\scriptstyle{i=0,1}}{\blue{\sn}, \orange \n \seqar \n}{
            \vlin{\cnd_{\sq}}{\bullet}{\blue{\sn}, \orange \n \seqar \red \n}{\vlhy{\vdots}}
        }{
            \vlin{\cnd_{\sq}}{\bullet}{\blue{\sn} , \red \n \seqar \n}{\vlhy{\vdots}}
        }
    }
}
\]
where we identify the sub-coderivations above the second and the third premises of the conditional. The coderivation is clearly progressing. Moreover it is safe, as $\sexp$ has no instances of $\cut_{\sq}$. Its  associated equational program can be written as follows:
\begin{equation}
 \arraycolsep=2pt
    \label{eq:sexp-eq-prog}
    \begin{array}{rcll}
    \denot \sexp (0;y) & = & \succ 0 (;y) \\
    \denot \sexp (\succ 0 x;y) & = & \denot \sexp (x; \denot \sexp(x;y))  &\quad  x \neq 0\\
    \denot \sexp (\succ 1 x ; y) & = & \denot \sexp (x; \denot \sexp(x;y))
\end{array}
\end{equation}
The above equational program has already appeared in \cite{Hofmann97,Leivant99}. It is not hard to show, by induction on $x$, that  $\denot\sexp(x;y)=2^{2^{\s{x}}} \cdot  y$.
Thus $f_\sexp$ has exponential growth rate (as long as $y\neq 0$), despite being defined by a `safe' recursion scheme. 
\end{exmp}

The above coderivation exemplifies a  safe recursion scheme that is able to \emph{nest} one recursive call inside another in order to obtain exponential growth rate. This is in fact a peculiar feature of circular proofs, and it is worth discussing.

\begin{rem}
[On nesting and higher-order recursion]
As we
have seen, namely in~\Cref{prop:computational-props-coderivations}.\ref{item:type1-prim-rec},
(progressing) $\bcnorec$-coderivations are able to simulate some sort of higher-order recursion, namely at type 1 (cf.~also \cite{Das2021}).
In this way it is arguably not so surprising that the sort of `nested recursion' in \Cref{eq:sexp-eq-prog} is definable since type 1 recursion, in particular, allows such nesting of the recursive calls. 
To make this point more apparent, consider the following higher-order `safe' recursion operator:
\begin{equation*}
  \mathsf{rec}_A: \sq{\n}\rightarrow (\sq{\n}\rightarrow A  \rightarrow     A)\rightarrow A \rightarrow A   
\end{equation*}
with $A= \n \rightarrow \n$, and      $f(x)=\mathsf{rec}_A(x,  h, g)$ is defined as $  f(0)=g$ and    $f(\succ i x)=h(x,f (x) )$  for $x>0$. By setting $g\dfn \lambda y:\n. \succ 0y$ and $h\dfn  \lambda x: \sn . \lambda u: \n \rightarrow \n . \lambda y: \n . u(u\, y)$ we can easily check that $\denot\sexp(x; y)= \mathsf{rec}_A(x,  h, g)(y)$, where $\sexp$ is as in Example~\ref{exmp:safe-exponentiation}. Hence, the function $\denot\sexp(x; y)$ can be defined by means of a higher-order version of safe recursion. 

As noticed by Hofmann~\cite{Hofmann97}, and formally proved by Leivant~\cite{Leivant99}, $\felementary$ can be characterised using higher-order safe recursion, thanks to this capacity to nest recursive calls.
Moreover, Hofmann showed in~\cite{Hofmann97} that by introducing a `linearity' restriction on the operator $ \mathsf{rec}_A$, which prevents duplication of recursive calls, it is possible to recover the class $\fptime$. The resulting type system, called $\slr$ (`Safe Linear Recursion'), can thus be regarded as a higher-order formulation of $\bc$.
\end{rem}

Following~\cite{Hofmann97}, we shall impose a linearity criterion to rule out those coderivations that nest recursive calls. This is achieved by observing that the duplication of the recursive calls of $\denot\sexp$ in Example~\ref{exmp:safe-exponentiation} is due to the presence in $\sexp$ of loops on $\bullet$ crossing both premises of a $\cut_\n$ step. Hence, our circular-proof-theoretic counterpart of Hofmann's linearity restriction  can be  obtained  by demanding that only the left premiss of each $\cut_\n$ step is crossed by such loops. 
Again, we rather give a more natural proof-level criterion which does not depend on our intuitive notion of loop.

\begin{defn}
[Left-leaning]
\label{defn:left-leaning}
A $\bcnorec$-coderivation is said to be \emph{left-leaning} if each infinite branch goes right at a $\cut_\n$-step only finitely often.
\end{defn}
\begin{exmp}
In~\Cref{fig:examples-regular-coderivations}, $\incrementation$ is trivially left-leaning, as it contains no instances of $\cut_\n$ at all. The  coderivations $\cconc$  and $\primrec$ are also left-leaning, since no  infinite branch can go right at the $\cut_\n$ steps. By contrast, the coderivation $\sexp$ in Example~\ref{exmp:safe-exponentiation} is not left-leaning, as there is an infinite branch looping at $\bullet$ and crossing infinitely many times the rightmost premise of the $\cut_\n$-step.
\end{exmp}

We are now ready to present our circular systems:
\begin{defn}
[Circular Implicit Systems]
\label{defn:cyclic-systems}
$\ncbc$ is the class of safe regular progressing $\bcnorec$-coderivations. 
$\cbc$ is the restriction of $\ncbc$ to only left-leaning coderivations. 
A two-sorted function $f(\vec x;\vec y)$ is \emph{$\ncbc$-definable} (or \emph{$\cbc$-definable}) if there is a coderivation $\der\in \ncbc$ (resp., $\der \in \cbc$) such that $\denot \der (\vec x;\vec y) = f(\vec x;\vec y)$.
\end{defn}

Let us point out that Proposition~\ref{prop:cutbox} can be strengthened to preserve safety and left-leaningness:

\begin{prop} \label{prop:cutbox-for-left-leaning} Let    $\der: \sq \Gamma , \vec \n \seqar \sn$ be a  coderivation in $\ncbc$ (or $\cbc$). There exists a $\ncbc$-coderivation (resp., $\cbc$-coderivation) $\der^*:\sq\Gamma \seqar \sn$  such that
\(
    \model{\der}(\vec{x}; \vec{y})= \model{\der^*}(\vec{x};)
\).
\end{prop}

\subsection{On the complexity of proof-checking}

Note that both the safety and the left-leaning conditions above are defined at the level of arbitrary coderivations, not just regular and/or progressing ones. Moreover, since these conditions are defined at the proof-level rather than the thread-level, they are easy to check on regular coderivations:
\begin{prop} \label{prop:safety-left-leaning-NL} The safety and the left-leaning condition are $\mathbf{NL}$-decidable for regular coderivations.
\end{prop}
\begin{proof} We can represent a regular coderivation $\der$ as a  finite directed (possibly cyclic) graph $G_\der$ labelled with inference rules. Then,  the problem of deciding whether $\der$ is not safe (resp.~left leaning) comes down to the problem of deciding whether  a cycle $\pi$ of $G_\der$  exists crossing a node labelled $\cut_\sq$ (resp.~crossing the rightmost child of a node labelled $\cut_\n$).  W.l.o.g.~we can assume that $\pi$ is a simple cycle. Given (an encoding of)  $G_\der$ as read-only input and (an encoding of) $\pi$ as a read-only  certificate, we can easily construct a deterministic Turing machine $M$ verifying that the cycle $\pi$ crosses $\cut_\sq$ (resp. the rightmost child of a node $\cut_\n$) in $G_\der$. More specifically:
\begin{itemize}
    \item the size of the certificate is smaller than the size of the description of $G_\der$;
    \item  $M$ reads once from left to right the addresses of the certificate which provide the information about where to move the pointers;
    \item $M$ works in logspace as it only stores addresses in memory. \qedhere
\end{itemize}
\end{proof}

The idea here is that, for regular coderivations, checking that no branch has infinitely many occurrences of a particular rule can be reduced to checking acyclicity of a certain subgraph, which is well-known to be in $\mathbf{coNL}=\mathbf{NL}$.

Recall that progressiveness of regular coderivations is decidable by reduction to universality of B\"uchi automata, a $\pspace$-complete problem.
Indeed progressiveness itself is $\pspace$-complete in many settings, cf.~\cite{NolletST19}.
It is perhaps surprising, therefore, that the safety of a regular coderivation also allows us to decide progressiveness efficiently too, thanks to the following reduction:
\begin{prop}\label{prop:regular-safe-progressing-iff-infinite-branch-infinite-cond}
A  safe  $\bcnorec$-coderivation is progressing iff every infinite branch has infinitely many $\cnd_\sq$-steps.
\end{prop}
\begin{proof}
The left-right implication is trivial. For the right-left implication, let us consider an infinite branch $B$ of a  safe $\bcnorec$-coderivation $\der$. By safety, there exists  a node $\nu$ of $B$ such that any sequent  above $\nu$ in $B$ is not the conclusion of a $\cut_\sq$-step. Now, by inspecting the rules of $\bcnorec \setminus \{ \cut_\sq \}$ we  observe that:
\begin{itemize}
    \item every modal formula occurrence in $B$ has a unique thread along $B$;
    \item infinite threads along $B$ cannot start strictly above $\nu$. 
\end{itemize}
Hence, setting $k$ to be the number of $\sn$ occurrences in the antecedent of $\nu$, $B$ has (at most) $ k$ infinite threads. Moreover, since $B$ contains infinitely many $\cnd_\sq$-steps, by the Infinite Pigeonhole Principle we conclude that  one of these  threads is infinitely often principal for the $\cnd_\sq$ rule. 
\end{proof}

Thus, using similar reasoning to that of~\Cref{prop:safety-left-leaning-NL}  we may conclude from~\Cref{prop:regular-safe-progressing-iff-infinite-branch-infinite-cond} the following:
\begin{cor} \label{rem:NL} Given a regular  $\bcnorec$-coderivation $\der$, the problem of deciding if $\der$ is in $\ncbc$ (resp.~$\cbc$)  is in  $\mathbf{NL}$.
\end{cor}

Let us point out that the reduction above is similar to (and indeed generalises) an analogous one for cut-free extensions of Kleene algebra, cf.~\cite[Proposition 8]{DP18}.


\section{Some variants of safe recursion}\label{sec:some-variants}
In this section we shall introduce various extensions of $\bc$ to ultimately classify the expressivity of the circular systems $\cbc$ and  $\ncbc$. 
First, starting from the analysis of   Example~\ref{exmp:safe-exponentiation} and  the subsequent system $\ncbc$, we shall define a version of $\bc$ with safe \emph{nested} recursion, called $\nbc$. 
Second, motivated by the more liberal way of defining functions in both $\cbc$ and $\ncbc$, we shall endow the function algebras $\bc$ and $\nbc$ with forms of safe recursion over a well-founded relation $\permpref$ on lists of normal parameters.  Figure~\ref{fig:function-algebras-table} summarises the function algebras considered and their relations.

\begin{figure}[t]
    \[
    \small
   \def\arraystretch{1.3}
    \begin{array}{|c|c|c|}
    \hline
     \textit{\textbf{safe recursion}} & \textit{ on notation} &\textit{ on $\permpref$ }\\       \hline
        \textit{ unnested}& \qquad \bc\qquad  & \qquad \bcpp \qquad \\
        \textit{ nested} &  \nbc  & \nbcpp \\ \hline 
    \end{array}
    \]
        \caption{The function algebras considered in \Cref{sec:some-variants}. Any algebra is included in one below it and to the right of it.}
    \label{fig:function-algebras-table}
\end{figure}

\subsection{Relativised algebras and nested recursion}
One of the key features of the Bellantoni-Cook algebra $\bc$ is that `nesting' of recursive calls is not permitted.
For instance, let us recall the equational program from \Cref{exmp:safe-exponentiation}:
\begin{equation}
 \arraycolsep=2pt
    \label{eq:safe-exp-by-snrec}
    \begin{array}{rcl}
    \ex(0;y) & = & \succ 0 y \\
    \ex(\succ i x ; y) & = & \ex(x;\ex(x;y))
\end{array}
\end{equation}
Recall that $\ex(x;y)=f_\sexp (x;y) = 2^{2^{\s{x}}}\cdot y$.
The `recursion step' on the second line is compatible with safe composition, in that safe inputs only occur in hereditarily safe positions, but one of the recursive calls takes another recursive call among its safe inputs.
In Example~\ref{exmp:safe-exponentiation}  we showed how  the above function  $\ex(x;y)$ can be   $\ncbc$-defined. 
We thus seek a suitable extension of $\bc$ able to formalise such nested recursion to serve as a function algebraic counterpart to $\ncbc$. 

It will be convenient for us to work with generalisations of  $\bc$ including \emph{oracles}.  

\begin{defn}\label{defn:oracles-functions}
For all sets of oracles $\vec a = a_1, \dots, a_k$, we define the algebra of $\bcnorec$ functions \emph{over} $\vec a$ to include all the initial functions of $\bcnorec$ and,
\begin{itemize}
    \item (oracles). $a_i(\vec x;\vec y)$ is a function over $\vec a$, for $1\leq i\leq k$, (where $\vec x$,$\vec y$ have appropriate length for $a_i$).
\end{itemize}
and closed under:
\begin{itemize}
    \item (Safe Composition). 
    \begin{enumerate}[(1)]
        \item \label{enum:safe-composition-condition-1} from $g(\vec x;\vec y), h(\vec x;\vec y,y)$ over $\vec a$ define $f(\vec x;\vec y)$ over $\vec a$  by $f(\vec x;\vec y) = h(\vec x; \vec y, g(\vec x;\vec y))$.
        \item \label{enum:safe-composition-condition} from $g(\vec x;)$ over $\emptyset$ and $h(\vec x, x;\vec y) $ over $\emptyset$ define $f(\vec x; \vec y)$ over $\vec a$ by $f(\vec x; \vec y) = h(\vec x, g(\vec x;); \vec y)$.
    \end{enumerate}
\end{itemize}
We write $\bcnorec(\vec a)$ for the class of functions over $\vec a$ generated in this way.
\end{defn}

Henceforth, we may write $f(\vec a)(\vec x; \vec y)$ to stress that the function $f(\vec x; \vec y)$ is over oracles $\vec a$.

\begin{rem}\label{rem:variables-vs-functions}
Formally, a function $f(\vec a)(\vec x; \vec y)$ can be understood as a type 2 functional that takes type 1 inputs $\vec a$ (safe-normal functions), and type 0 inputs $\vec x$ (normal) and $\vec y$ (safe). With little abuse of notation, we may represent this functional as $\lambda \vec a. f(\vec a)(\vec x; \vec y)$, where $\vec a$ are now \emph{variables} ranging over first-order functions (with appropriate sorts). Notice that this justifies replacing an oracle with another first-order function (with matching sorts). However, to avoid any reference to higher-order functions, in this paper we prefer to adopt the parametric notation $f(\vec a)(\vec x; \vec y)$ where oracles are arbitrary but fixed \emph{functions}. 
\end{rem}

Note that Safe Composition along normal parameters (\Cref{enum:safe-composition-condition} above) comes with the condition that $g(\vec x; )$ is oracle-free. This restriction prevents oracles appearing in normal position. Since our recursion schemes replace oracles with recursive calls, the latter will also be required to appear in safe position. 
The same condition on $h(\vec x, x; \vec y)$ being oracle-free is not strictly necessary for the complexity bounds we are after, as we shall see in the next section when we define more expressive algebras, but is convenient in order to facilitate the `grand tour' strategy of this paper (cf.~\Cref{fig:picture-main-results}).

We shall write, say, $\lambda \vec v .f(\vec x; \vec v)$ for the function taking only safe arguments $\vec v$ with $(\lambda \vec v .f(\vec x; \vec v))(;\vec y) = f(\vec x; \vec y)$ (here $\vec x$ may be seen as parameters).
Nested recursion can be formalised in the setting of algebras-with-oracles as follows:

\begin{defn}
[Safe Nested Recursion]
\label{defn:safe-nested-recursion}
We write $\snrec$ for the scheme:
\begin{itemize}
    \item from $g(\vec x; \vec y)$ over $\emptyset$ and $h_i(a)(x,\vec x;\vec y)$ over $a,\vec a$, define $f(x, \vec x;\vec y) $ over $\vec a$ by:
    \[
     \arraycolsep=2pt
    \begin{array}{rcll}
     f(0, \vec x; \vec y) &=&g(\vec x;\vec y)
    \\
        f(\succ 0 x, \vec x; \vec y) &=&h_0(\lambda \vec v . f(x,\vec x;\vec v))(x,\vec x;\vec y) &\quad  x \neq 0 \\
        f(\succ 1 x, \vec x; \vec y) &=&h_1(\lambda \vec v . f(x,\vec x;\vec v))(x,\vec x;\vec y)
    \end{array}
    \]
\end{itemize}
We write $\nbc(\vec a)$ for the class of functions over $\vec a$ generated from $\bcnorec$ under $\snrec$ and Safe Composition (from \Cref{defn:oracles-functions}), and write simply $\nbc$ for $\nbc(\emptyset)$.
\end{defn}

\begin{rem}[Safe Composition During Safe Recursion]
\label{rem:errata}
    Note that Safe Nested Recursion also admits variants that are not morally `nested' but rather use a form of `composition during recursion':
\begin{itemize}
    \item from $g(\vec x; \vec y), \left(\vec g_j(x, \vec x; \vec y)\right)_{1 \leq j \leq k}$ and  $h_i\left(x, \vec x; \vec y, (z_j)_{1 \leq j \leq k}\right)$  define:
    \begin{equation}
        \label{eq:safe-rec-w-comp-dur-rec-no-oracles}
         \arraycolsep=2pt
    \begin{array}{rcll}
         f(0, \vec x; \vec y) &=&g(\vec x;\vec y)\\
         f(\succ i x , \vec x; \vec y) &=& 
    h_i \left(x, \vec x; \vec y, \left(f(x, \vec x; \vec g_j(x, \vec x; \vec y))\right)_{1 \leq j \leq k}\right)
        \end{array}
    \end{equation}
\end{itemize}
In this case, note that we have allowed the safe inputs of $f$ to take arbitrary values given by previously defined functions, but at the same time $f$ never calls itself in a safe position, as in~\eqref{eq:safe-exp-by-snrec}. 
In the conference version of this paper \cite[p.8]{CurziDas:lics22} we made an unsubstantiated claim that a function algebra extending $\bcnorec$ by \eqref{eq:safe-rec-w-comp-dur-rec-no-oracles} above is equivalent to a restriction $\sbc$ of $\nbc$ requiring that oracles are \emph{unnested} in Safe Composition: in \cref{enum:safe-composition-condition-1}, one of $g(\vec x;\vec y)$ and $h(\vec x;\vec y,y)$ must be oracle-free, i.e.\ over $\emptyset$. 
Such an equivalence is in fact not immediate, as our Safe Composition scheme allows substitution only along a single parameter, so there does not seem to be a direct way of expressing~\Cref{eq:safe-rec-w-comp-dur-rec-no-oracles} in $\nbc$ without nesting oracles, in the sense just described. 
While we did not rely on this equivalence for any of the main results of \cite{CurziDas:lics22}, in this work we simply omit reference to $\sbc$ and its variants to avoid any confusion.
\end{rem}

\subsection{Safe recursion on well-founded relations}
Relativised function algebras may be readily extended by recursion on arbitrary well-founded relations.
For instance, given a well-founded preorder $\wfpoeq$, and writing $\wfpo$ for its strict variant,\footnote{To be precise, for a preorder $\wfpoeq$ we write $x \wfpo y  $ if $x\wfpoeq y$ and $ y \not\wfpoeq x$. As abuse of terminology, we say that $\wfpoeq$ is well-founded just when $\wfpo$ is.}
`safe recursion on $\wfpo$' is given by the scheme:
\begin{itemize}
    \item from $h(a)(x,\vec x;\vec y)$ over $a,\vec a$, define $f(x, \vec x;\vec y)$ over $\vec a$ by: 
    \[ f(x,\vec x; \vec y) = h(\lambda v \wfpo x . f(v,\vec x;\vec y ))(x,\vec x; \vec y )\] 
\end{itemize}
Note here that we employ the notation $\lambda v \wfpo x$ for a `guarded abstraction'. Formally:
\[
(\lambda v \wfpo x . f(v,\vec x; \vec y))(z)
\dfn
\begin{cases}
f(z, \vec x; \vec y) & z \wfpo x \\
0 & \text{otherwise}
\end{cases}
\]
It is now not hard to see that total functions (with oracles) are closed under the recursion scheme above, by reduction to induction on the well-founded relation $\wfpo$.

Note that such schemes can be naturally extended to preorders on \emph{tuples} of numbers too, by abstracting several inputs.
We shall specialise this idea to a particular well-founded preorder that will be helpful later to bound the complexity of definable functions in our systems $\cbc$ and $\ncbc$.

Recall that we say that $x$ is a \emph{prefix} of $y$ if $y$ has the form $xz$ in binary notation, i.e.\ $y $ can be written $x2^n + z$ for some $n\geq 0$ and some $z<2^n$.
We say that $x$ is a \emph{strict prefix} of $y$ if $x$ is a prefix of $y$ but $x \neq y$.

\anupam{Some example earlier should highlight why the permuation of prefixes order is important for our considerations}
\gianluca{Recall example about "max" function, which merge two standard recursions into 1 on two parameters. Perhaps not the best example but it is nstructive.}

\begin{defn}
[Permutations of prefixes]
\label{defn:perm-pref}
Let $[n]$ denote $\{0,\dots, n-1\}$.
We write $(x_0, \dots, x_{n-1}) \permprefeq (y_0, \dots, y_{n-1})$ if, for some permutation $\pi:[n] \to [n] $, we have that $x_i$ is a prefix of $y_{\pi i}$, for all $i<n$.
We write $\vec x \permpref \vec y$ if $\vec x \permprefeq \vec{y} $ but $ \vec y \not\permprefeq \vec{x}$, i.e.\ there is a permutation $\pi : [n]\to [n]$ with $x_i$ a prefix of $y_i$ for each $i<n$ and, for some $i<n$, $x_i$ is a strict prefix of $y_i$.
\end{defn}

It is not hard to see that $\permprefeq$ is a well-founded preorder, by reduction to the fact that the prefix relation is a well-founded partial order.
As a result, we may duly devise a version of safe (nested) recursion on $\permpref$:

\begin{defn}
[Safe (nested) recursion on permutations of prefixes]
\label{defn:saferec-pp}
We write $\nbcpp(\vec a)$ for the class of functions over $\vec a$ generated from $\bcnorec$ under Safe Composition, the scheme $\snrecpp$,
\begin{itemize}
    \item from $h(a)(\vec x; \vec y)$ over $a, \vec a$ define $f(\vec x;\vec y)$ over $\vec a$ by:
    \[
    f(\vec x; \vec y) = 
    h(\lambda \vec u \permpref \vec x, \lambda \vec v. f(\vec u; \vec v))(\vec x; \vec y)
    \]
\end{itemize}
and
the following generalisation of Safe Composition along a Normal Parameter:
\begin{itemize}
 \item[(2)$'$] \label{item:safe-composition-normal-param-with-nonempty-oracles} from $g(\vec x;)$ over $\emptyset$ and $h(\vec x, x;\vec y) $ over $\vec a$ define $f(\vec x; \vec y)$ over $\vec a$ by $f(\vec x; \vec y) = h(\vec x, g(\vec x;); \vec y)$.
\end{itemize}
We define $\bcpp(\vec a)$ to be the restriction of $\nbcpp(\vec a)$
where every instance of $\snrecpp$ has the form:
\begin{itemize}
    \item from $h(a)(\vec x; \vec y)$ over $a, \vec a$, define $f(\vec x;\vec y) $ over $\vec a$:
    \[
    f(\vec x; \vec y) = 
    h(\lambda \vec u \permpref \vec x , \lambda \vec v \permprefeq \vec y . f(\vec u; \vec v))(\vec x ; \vec y)
    \]
\end{itemize}
We call this latter recursion scheme $\srecpp$, e.g.\ if we need to distinguish it from $\snrecpp$.
\end{defn}
Note that the version of safe composition along a normal parameter above differs from the previous one, \Cref{enum:safe-composition-condition} from \Cref{defn:oracles-functions}, since the function $h$ is allowed to use oracles.
Again, this difference is inessential in terms of computational complexity, as we shall see. However, as we have mentioned, the greater expressivity of $\bcpp$ and $\nbcpp$ will facilitate our overall strategy for characterising $\cbc$ and $\ncbc$, cf.~\Cref{fig:picture-main-results}.

%

Let us take a moment to point out that $\nbcpp(\vec a)\supseteq  \bcpp(\vec a)$ indeed contain only well-defined total functions over the oracles $\vec a$, by reduction to induction on $\permpref$.

\subsection{Simultaneous recursion schemes}
For our main results, we will ultimately need a typical property of function algebras, that $\bcpp$ and $\nbcpp$ are closed under simultaneous versions of their recursion schemes.

\begin{defn}[Simultaneous schemes]
\label{defn:simultaneous-schemes}
We define schemes $\ssrecpp$ and $\ssnrecpp$, respectively, as follows, for arbitrary $\vec a = a_1, \dots, a_{k}$:
\begin{itemize}
    \item from $h_i(\vec a)(\vec x; \vec y)$ over $\vec a, \vec b$, for $1 \leq i\leq k$, define $f_i(\vec x;\vec y)$ over $\vec b$, for $1 \leq i \leq k$, by:
    \[
    f_i(\vec x; \vec y) = 
    h_i((\lambda \vec u \permpref \vec x, \lambda \vec v\subseteq \vec y. f_j(\vec u; \vec v))_{1 \leq j\leq  k})(\vec x; \vec y)
    \]
    \item from $h_i(\vec a)(\vec x; \vec y)$ over $\vec a, \vec b$, for $1\leq i \leq k$, define $f_i(\vec x;\vec y)$ over $\vec b$, for $1 \leq i \leq  k$, by:
    \[
    f_i(\vec x; \vec y) = 
    h_i((\lambda \vec u \permpref \vec x, \lambda \vec v. f_j(\vec u; \vec v))_{1 \leq j \leq k})(\vec x; \vec y)
    \]
\end{itemize}
\end{defn}
\begin{prop} \label{prop:simultaneous-recursion-admissible}
We have the following:
\begin{enumerate}
    \item\label{item:sim-rec-admiss-bcpp} If $\vec f(\vec x;\vec y)$ over $\vec b$ are obtained by applying  $\ssrecpp$ to $\vec h(\vec a)(\vec x;\vec y) \in \bcpp(\vec a,\vec b)$, then also $\vec f(\vec x;\vec y) \in \bcpp (\vec b)$.
\item\label{item:sim-rec-admiss-nbcpp-sbcpp}
If $\vec f(\vec x;\vec y)$ over $\vec b$ are obtained by applying  $\ssnrecpp$ to $\vec h(\vec a)(\vec x;\vec y) \in \nbcpp(\vec a,\vec b)$, then also $\vec f(\vec x;\vec y) \in \nbcpp (\vec b)$.
\end{enumerate}
\end{prop}

\begin{proof}
We only prove \Cref{item:sim-rec-admiss-bcpp}, i.e. that $\bcpp$ is closed under $\ssrecpp$, but the same argument works for \Cref{item:sim-rec-admiss-nbcpp-sbcpp}: just ignore all guards on safe inputs in recursive calls.

Let $f_i(\vec x;\vec y)$ and $h_i(a_1, \dots, a_k)(\vec x; \vec y)$ be as given in \Cref{defn:simultaneous-schemes}, and temporarily write $f_j^{\vec x;\vec y}$ for $\lambda \vec u \permpref \vec x, \lambda \vec v \permprefeq \vec y . f_j (\vec u;\vec v)$, so we have:
\[
f_i(\vec x;\vec y) = h_i (f_1^{\vec x; \vec y}, \dots, f_k^{\vec x;\vec y})(\vec x; \vec y)
\]
For $i\in \Nat$, let us temporarily write $\numeral i$ for $i$ in binary notation\footnote{In fact, any notation will do, but we pick one for concreteness.}, and $\permi i$ for the list $\numeral i, \numeral{i+1}, \dots, \numeral k, \numeral 1, \numeral 2, \dots, \numeral{i-1}$.
Note that, for all $i=1, \dots, k$, $\permi i$ is a permutation (in fact a rotation) of $\numeral 1, \dots, \numeral k$.

Now, let $f(\vec x; \vec y, \vec z)$ over oracles $\vec b$ be given as follows:
\begin{equation}
    \label{eq:sim-rec-reduction}
    f(\vec x; \vec y, \vec z) \dfn \begin{cases}
h_1 (f_1^{\vec x;\vec y}, \dots, f_k^{\vec x; \vec y})(\vec x;\vec y) & \vec z = \permi 1 \\
\vdots \\
h_k (f_1^{\vec x;\vec y}, \dots, f_k^{\vec x; \vec y})(\vec x;\vec y) & \vec z = \permi k \\
0 & \text{otherwise}
\end{cases}
\end{equation}
Note that this really is a finite case distinction since each of the boundedly many $\permi i$ has bounded size, both bounds depending only on $k$, and so is computable in $\bcnorec$ over $\vec h$. 

By definition, then, we have that $f(\vec x;\vec y, \permi i) = f_i (\vec x; \vec y)$.
Moreover note that, for each $j=1,\dots,k$, we have,
\[
\arraycolsep=2pt
\begin{array}{rcl}
   f_j^{\vec x;\vec y} (\vec u'; \vec v')  &=& (\lambda \vec u \permpref \vec x, \lambda \vec v \permprefeq \vec y. f_j(\vec u;\vec v)) (\vec u';\vec v') \\
    & = & (\lambda \vec u \permpref \vec x, \lambda \vec v \permprefeq \vec y. f(\vec u;\vec v,\permi j ))(\vec u';\vec v') \\
    & = & (\lambda \vec u \permpref \vec x, \lambda \vec v \permprefeq \vec y, \lambda \vec w \permprefeq \vec z . f(\vec u;\vec v, \vec w)) (\vec u';\vec v',\permi j)
\end{array}
\]
as long as $\vec z$ is some $\vec i$,
so indeed~\eqref{eq:sim-rec-reduction} has the form,
\[
f(\vec x; \vec y, \vec z) = h (\lambda \vec u \permpref \vec x, \lambda \vec v \permprefeq \vec y, \lambda \vec w \permprefeq \vec z . f(\vec u; \vec v, \vec w)) (\vec x;\vec y)
\]
and $f(\vec x;\vec y,\vec z) \in \bcpp(\vec b)$ by $\srecpp$.
Finally, since  $f_i(\vec x;\vec y) = f(\vec x;\vec y, \permi i)$, we indeed have that each $f_i(\vec x;\vec y) \in \bcpp(\vec b)$.
\end{proof}

\section{Characterisations for function algebras}\label{sec:characterization-results-for-function-algebras}
In this section we characterise the complexities  of the function algebras we introduced in the previous section. 
Namely, despite apparently extending $\bc$, $\bcpp$ still contains just the polynomial-time functions, whereas  both  $\nbc$ and $\nbcpp$ are shown to contain just the elementary functions. 
All such results rely on a `bounding lemma' inspired by~\cite{BellantoniCook}.

\subsection{A relativised Bounding Lemma}
\label{sec:bounding-lemma}
Bellantoni and Cook showed in \cite{BellantoniCook} that any function $f(\vec x;\vec y) \in \bc$ satisfies the `poly-max bounding lemma': there is a polynomial $p_f(\vec n)$ such that:\footnote{Recall that, for $\vec x = x_1, \dots, x_n$, we write $|\vec x|$ for $|x_1|, \dots, |x_n|$.}
\begin{equation}
    \label{eq:bc-polymax-eqn}
    |f(\vec x;\vec y)| \leq p_f (|\vec x|) + \max |\vec y| 
\end{equation}
This provided a suitable invariant for eventually proving that all $\bc$-functions were polynomial-time computable.

In this work, inspired by that result, we generalise the bounding lemma to a form suitable to the relativised algebras from the previous section. 
To this end we establish in the next result a sort of `elementary-max' bounding lemma that accounts for the usual poly-max bounding as a special case, by appealing to the notion of (un)nested safe composition.
Both the statement and the proof are quite delicate due to our algebras' formulation using oracles; we must assume an appropriate bound for the oracles themselves, and the various (mutual) dependencies in the statement are subtle.

To state and prove the Bounding Lemma, let us employ the notation $\sumlen{\vec x} \dfn \sum |\vec x|$. 

\begin{lem}[Bounding Lemma]
\label{lem:boundinglemma}
Let $f(\vec a)(\vec x;\vec y) \in \nbcpp(\vec a)$, with $\vec a = a_1, \dots, a_k$, and suppose $g_1, \ldots, g_k$ are safe-normal functions. Then there is an elementary function $e_f(n)$ and a constant $d_f \geq 1$ such that whenever there are constants $\vec c=c_1, \dots, c_k$ satisfying, for $1\leq i \leq k$,
\begin{equation}
    \label{eq:oracle-const-max-bound}
    |g_i (\vec u;\vec v)| \leq c_i + \left(d_f \sum_{j\neq i} c_j \right) + \max|\vec v| \qquad \qquad \text{for all inputs } \vec u,\vec v
\end{equation}
then we have the following, for all inputs $\vec u, \vec v$\footnote{To be clear,  we write $|\vec y| \leq m_f(\vec c,\vec u,\vec v)  $ here as an abbreviation for $\{|y_{j}| \leq m_f(\vec c,\vec u,\vec v)\}_j$.}:
\begin{equation}
    \label{eq:elem-const-max-bound}
    |f(\vec g)(\vec u;\vec v)| \leq m_f(\vec c, \vec u,\vec v)
\end{equation}
\begin{equation}
    \label{eq:input-bounding-eqn}
    f(\vec g)(\vec u; \vec v) \ = \ f\left(\lambda \vec x.\lambda |\vec y| \leq m_f(\vec c, \vec u,\vec v) . g_i(\vec x;\vec y)\right)_{1 \leq i \leq k}(\vec u;\vec v)
\end{equation}
where 
\[
m_f(\vec c,\vec x,\vec y) = 
e_f(\sumlen{\vec x}) + d_f\sum\vec c + \max |\vec y|
\]
Moreover, if in fact $f(\vec x;\vec y) \in \bcpp(\vec a) $, then $d_f=1$ and $e_f(n)$ is a polynomial. 
\end{lem}

Notice that in the above statement we employ the notation $\lambda \vec x, \lambda|\vec y|\leq n$, where $\vec y=y_1, \ldots, y_k$, for the guarded abstraction defined by:
\[
(\lambda \vec x, \lambda|\vec y|\leq n. h(\vec x; \vec y))(\vec w; \vec z):= \begin{cases}
    h(\vec w; \vec z) &\text{if } |z_j|\leq n \text{ for all }1 \leq j \leq k\\
    0 &\text{otherwise.}
\end{cases}
\]

Unwinding the statement above, note that $e_f$ and $d_f$ depend \emph{only} on the function $f$ itself, not on the constants $\vec c$ given for the (mutual) oracle bounds in \Cref{eq:oracle-const-max-bound}.
This is crucial for the proof, namely in the case when $f$ is defined by recursion, substituting different values for $\vec c$ during an inductive argument.

While the role of the elementary bounding function $e_f$ is a natural counterpart of $p_f$ in Bellantoni and Cook's bounding lemma, cf.~\Cref{eq:bc-polymax-eqn}, the role of $d_f$ is perhaps slightly less clear.
Intuitively, $d_f$ represents the amount of `nesting' in the definition of $f$, increasing whenever oracle calls are substituted into arguments for other oracles. 
Hence, if $f$ uses only unnested Safe Composition, then $d_f=1$ as required. 
In fact, it is only important to distinguish whether $d_f=1$ or not, since $d_f$ forms the base of an exponent for defining $e_f$ when $f$ is defined by safe recursion.

Finally let us note that~\Cref{eq:elem-const-max-bound,eq:input-bounding-eqn} are somewhat dual:  while the former bounds the \emph{output} of a function (serving as a modulus of \emph{growth}), the latter bounds the \emph{inputs} (serving as a modulus of \emph{continuity}).

\begin{rem}
Let us also point out that we may relativise the statement of the lemma to any set of oracles including those which $f(\vec x;\vec y)$ is over.
In particular, if $f(\vec x;\vec y)$ is over no oracles, then we may realise \Cref{eq:oracle-const-max-bound} vacuously by choosing $\vec a = \emptyset$ and we would obtain that $|f(\vec x;\vec y)| \leq e_f(\sumlen{\vec x}) + \max |\vec y|$.
More interestingly, in the case when $f(\vec x;\vec y)$ is just, say, $a_i(\vec x;\vec y)$, we may choose to set $\vec a = a_i$ or $\vec a = a_1,\dots, a_n$ in the above lemma, yielding different bounds in each case.
We shall exploit this in inductive hypotheses  in the proof that follows (typically when we write `WLoG').
\end{rem}

\begin{proof}[Proof of~\Cref{lem:boundinglemma}.]
\label{prf:boundinglemma}
We prove~\Cref{eq:elem-const-max-bound} and~\Cref{eq:input-bounding-eqn} by induction on the definition of $f(\vec x;\vec y)$, always assuming that we have $\vec c$ satisfying \Cref{eq:oracle-const-max-bound}.

Throughout the argument we shall actually construct $e_f( n)$ that are \emph{monotone} elementary functions (without loss of generality) and exploit this invariant.
In fact $e_f( n)$ will always be generated by composition from $0, \succ{}, +, \times,  x^y$ and projections.
If $f(\vec x;\vec y) \in \bcpp(\vec a)$ then $e_f(n)$ will be in the same algebra without exponentiation, $ x^y$, i.e.\ it will be a polynomial with only non-negative coefficients.
This property will be made clear by the given explicit definitions of $e_f( n)$ throughout the argument.

Let us start with~\Cref{eq:elem-const-max-bound}. If $f(\vec x;\vec y)$ is an initial function then it suffices to set $e_f(n) \dfn 1 +  n$ and $d_f\dfn 1$.

If $f(\vec x;\vec y) = a_i(\vec x;\vec y)$ then it suffices to set $e_f(n) \dfn 0 $ and $d_f \dfn 1$.

If $f(\vec x;\vec y) = h(\vec x;\vec y, g(\vec x;\vec y))$, let $e_h, e_g, d_h, d_g$ be obtained from the inductive hypothesis.
We have,
\[
\arraycolsep=2pt
\begin{array}{rcl}
    |f(\vec x;\vec y)| & = & |h(\vec x; \vec y, g(\vec x; \vec y))| \\
    & \leq & e_h(\sumlen{\vec x}) + d_h\sum\vec c + \max(|\vec y|, |g(\vec x;\vec y)|) \\
    & \leq & e_h(\sumlen{\vec x}) + d_h\sum\vec c + e_g(\sumlen{\vec x}) + d_g\sum \vec c + \max|\vec y|
\end{array}
\]
so we may set $e_f(n) \dfn e_h(n) + e_g(n)$ and $d_f\dfn d_h+d_g$.

For the `moreover' clause, note that if $f(\vec x;\vec y)$ uses only unnested Safe Composition, then one of $g$ or $h$ does not have oracles, and so we can assume WLoG that either the term $d_h\sum \vec c$ or the term $d_g\sum \vec c$ above does not occur, and we set $d_f \dfn d_g$ or $d_f\dfn d_h$, respectively. In either case we obtain $d_f=1$ by the inductive hypothesis, as required.

If $f(\vec x;\vec y) = h(\vec x, g(\vec x;);\vec y)$, let $e_h,e_g,d_h, d_g$ be obtained from the inductive hypothesis. Note that, by definition of Safe Composition along a normal parameter, we must have that $g$ has no oracles, and so in fact $|g(\vec x;)| \leq e_g(\sumlen{\vec x})$.
We thus have,
\[
\arraycolsep=2pt
\begin{array}{rcl}
    |f(\vec x;\vec y)| & = & |h(\vec x, g(\vec x;);\vec y)| \\
     & \leq & e_h(\sumlen{\vec x}, |g(\vec x;)|) + d_h\sum \vec c + \max|\vec y| \\
     & \leq & e_h(\sumlen{\vec x}, e_g(\sumlen{\vec x})) + d_h\sum \vec c + \max|\vec y| 
\end{array}
\]
so we may set $e_f(n)\dfn e_h(n, e_g(n))$ and $d_f \dfn d_h$.
For the `moreover' clause, note that by the inductive hypothesis $e_h,e_g$ are polynomials and $d_h=1$, so indeed $e_f$ is a polynomial and $d_f=1$.

Finally, if $f(\vec x;\vec y) = h(\lambda \vec u \permpref \vec x, \lambda \vec v. f(\vec u;\vec v) ) (\vec x;\vec y)$, let $e_h,d_h$ be obtained from the inductive hypothesis.
We claim that it suffices to set $d_f \dfn d_h$ and $e_f(n) \dfn nd_h^n e_h(n)$.
Note that, for the `moreover' clause, if $d_h=1$ then also $d_h^n = 1$ and so indeed $e_f(n)$ is a polynomial if $d_h=1$ and $e_h(n)$ is a polynomial.

First, let us calculate the following invariant, for $n>0$:
\begin{equation*}
    \arraycolsep=2pt
    \begin{array}{rcll}
    e_f(n) & = & nd_h^n e_h(n) &  \\
        & = & d_h^n e_h(n) + (n-1)d_h^n e_h(n) & \\
        & = & d_h^n e_h(n) + d_h (n-1) d_h^{n-1}e_h(n) \\
        & \geq & e_h(n) + d_h (n-1) d_h^{n-1}e_h(n) & \text{$d_h\geq 1$} \\
        & \geq & e_h(n) + d_h (n-1) d_h^{n-1}e_h(n-1) \  &  \text{$e_h$~monotone} \\
        & \geq & e_h(n) + d_h e_f(n-1) & \text{def.~of $e_f$}
\end{array}
\end{equation*}
Hence:
\begin{equation}
 \label{eq:bnd-lem-rec-fn-inv}
     e_f(n) \geq  e_h(n) + d_h e_f(n-1) 
\end{equation}

Now, to show that \Cref{eq:elem-const-max-bound} bounds the $f,e_f,d_f$ at hand, we proceed by a sub-induction on $\sumlen{\vec x}$.
For the base case, when $\sumlen{\vec x} = 0 $ (and so, indeed, $\vec x = \vec 0$), note simply that $\lambda \vec u\permpref \vec x, \lambda \vec v . f(\vec u;\vec v)$ is the constant function $0$, and so we may appeal to the main inductive hypothesis for $h(a)$ setting the corresponding constant $c$ for $a$ to be $0$ to obtain,
\[
\arraycolsep=2pt
\begin{array}{rcl}
    |f(\vec 0;\vec y)| & = & |h(0)(\vec 0;\vec y)| \\
        & \leq & e_h(0) + d_h\sum\vec c + \max |\vec y| \\
        & \leq & e_f(0) + d_f\sum\vec c + \max |\vec y|
\end{array}
\]
as required.
For the sub-inductive step,  let $\sumlen{\vec x} > 0$. 
Note that, whenever $\vec u \permpref \vec x$ we have $\sumlen{\vec u} < \sumlen{\vec x} $ and so, by the sub-inductive hypothesis and monotonicity of $e_f$ we have:
\[
|f(\vec u;\vec v)| \leq e_f(\sumlen {\vec x} -1) + d_f \sum \vec c + \max |\vec v|
\]
Now we may again appeal to the main inductive hypothesis for $h(a)$ by setting $c= e_f(\sumlen{\vec x}  -1)$ to be the corresponding constant for $a = \lambda \vec u \permpref \vec x, \lambda \vec v. f(\vec u;\vec v)$. 
We thus obtain:
\[
\arraycolsep=2pt
\begin{array}{rclll}
     |f(\vec x;\vec y)| & = & |h(\lambda \vec u \permpref \vec x, \lambda \vec v . f(\vec u;\vec v)) (\vec x;\vec y)| \\
         & \leq & e_h (\sumlen{\vec x}) + d_h c + d_h\sum\vec c + \max |\vec y| & & \text{main IH} \\
         & \leq & ( e_h (\sumlen{\vec x}) + d_h e_f(\sumlen{\vec x} -1) ) +\\
         & &  + d_h\sum\vec c + \max |\vec y| & & \text{def.~of $c$}\\
         & \leq & e_f(\sumlen{\vec x}) + d_h \sum \vec c + \max|\vec y| && \eqref{eq:bnd-lem-rec-fn-inv} \\ & \leq & e_f(\sumlen{\vec x}) + d_f \sum \vec c + \max|\vec y| && \text{$d_f = d_h$}
\end{array}
\]
Let us now prove~\Cref{eq:input-bounding-eqn}, and let $e_f$ and $d_f$ be constructed by induction on $f$ as above. We proceed again by induction on the definition of $f(\vec a)(\vec x;\vec y)$, always making explicit the oracles of a function.
\gianluca{Is it better to define $e_f$ and $d_f$ in advance before proving the two inequations?}

The initial functions and oracle calls are immediate, due to the `$\max |\vec y|$' term in \Cref{eq:input-bounding-eqn}. 

If $f(\vec a)(\vec x; \vec y) = h(\vec a)(\vec x; \vec y, g(\vec a)(\vec x; \vec y))$ then, by the inductive hypothesis for $h(\vec a)$, any oracle call from $h(\vec a)$ only takes safe inputs of lengths:
    \[
    \begin{array}{rll}
        \leq & e_h(\sumlen{\vec x}) + d_h \sum \vec c + \max (|\vec y|, |g(\vec a)(\vec x;\vec y)|) \\
        \leq & e_h(\sumlen{\vec x})+ d_h \sum \vec c + e_g(\sumlen{\vec x}) + d_g\sum \vec c + \max |\vec y| & \text{\eqref{eq:oracle-const-max-bound}} \\
        \leq &  (e_h(\sumlen{\vec x} + e_g(\sumlen{\vec x})) + (d_h +d_g)\sum\vec c + \max |\vec y|\\
         \leq & e_f(\sumlen{\vec x}) + d_f\sum\vec c + \max|\vec y|
    \end{array}
    \]
    Note that any oracle call from $g(\vec a)$ will still only take safe inputs of lengths $\leq e_g(\sumlen{\vec x}) + d_g\sum \vec c + \max |\vec y|$, by the inductive hypothesis, and $e_g$ and $d_g$ are bounded above by $e_f$ and $d_f $ respectively.
    
If $f(\vec a)(\vec x;\vec y) = h(\vec a)(\vec x, g(\emptyset)(\vec x;);\vec y)$ then, by the inductive hypothesis, any oracle call will only take safe inputs of lengths:
\[
\arraycolsep=2pt
\begin{array}{rlll}
    \leq & e_h(\sumlen{\vec x} + |g(\emptyset)(\vec x;)|) + d_h\sum\vec c + \max |\vec y| \\
    \leq & e_h(\sumlen{\vec x} + e_g(\sumlen{\vec x})) + d_h\sum\vec c + \max |\vec y| & & \text{\Cref{lem:boundinglemma}} \\
    \leq & e_f(\sumlen{\vec x}) + d_f\sum \vec c + \max |\vec y|
\end{array}
\]

Last, suppose $f(\vec a)(\vec x;\vec y)= h(\vec a, \lambda \vec u\permpref \vec x, \lambda \vec v . f(\vec a)(\vec u;\vec v))(\vec x;\vec y)$. 
We proceed by a sub-induction on $\sumlen{\vec x}$.
Note that, since $\vec u \permpref \vec x \implies \sumlen{\vec u} <\sumlen{\vec x}$, we immediately inherit from the inductive hypothesis the appropriate bound on safe inputs for oracle calls from $\lambda \vec u\permpref \vec x, \lambda \vec v. f(\vec a)(\vec u;\vec v)$.

Now, recall from the Bounding \Cref{lem:boundinglemma}, whenever $\vec u \permpref \vec x$ (and so $\sumlen{\vec u} < \sumlen{\vec x}$), we have $|f(\vec u;\vec v)| \leq e_f(\sumlen{\vec x}-1) + d_f\sum \vec c + \max |\vec v|$.
So by setting $c = e_f(\sumlen{\vec x}-1)$ in the inductive hypothesis for $h(\vec a,a)$, with $a = \lambda \vec u\permpref \vec x, \lambda \vec v. f(\vec a)(\vec u;\vec v)$, any oracle call from $h(\vec a,a)$ will only take safe inputs of lengths:
\[
\arraycolsep=2pt
\begin{array}{rlll}
    \leq & e_h(\sumlen{\vec x}) + d_h e_f(\sumlen{\vec x}-1) + d_h\sum \vec c + \max |\vec y| \\
    \leq & e_f(\sumlen{\vec x}) + d_f \sum \vec c + \max |\vec y| & & \text{\eqref{eq:bnd-lem-rec-fn-inv}}
\end{array}
\]
This completes the proof.  
\end{proof}

\subsection{Soundness results}
In this subsection we  show that the function algebras $\bcpp$ and $\nbcpp$ (as well as $\nbc$)  capture precisely the classes  $\fptime$ and $\felementary$, respectively. We start with $\bcpp$:

\begin{thm}
\label{thm:fp-soundness}
Suppose $\vec a$ satisfies \Cref{eq:oracle-const-max-bound} for some constants $\vec c$.
We have the following:
\begin{enumerate}
    \item\label{item:soundness-fp} If $f(\vec x;\vec y) \in \bcpp (\vec a)$ then $f(\vec x,\vec y) \in \fptime(\vec a)$.
    \item\label{item:soundness-felementary} If $f(\vec x;\vec y) \in \nbcpp (\vec a)$ then $f(\vec x,\vec y) \in \felementary(\vec a)$.
\end{enumerate}
\end{thm}

Note in particular that, for $f(\vec x;\vec y)$ in $\bcpp$ or $\nbcpp$, i.e.\ not using any oracles, we immediately obtain membership in $\fptime$ or $\felementary$, respectively.
However, the reliance on intermediate oracles during a function definition causes some difficulties that we must take into account.
At a high level, the idea is to use the Bounding Lemma   (namely \Cref{eq:input-bounding-eqn}) to replace certain oracle calls with explicit appropriately bounded functions computing their graphs. 
From here we compute $f(\vec x;\vec y)$ by a sort of `course-of-values' recursion on $\permpref$, storing previous values in a lookup table.
In the case of $\bcpp$, it is important that this table has polynomial-size, since there are only $m! \prod |\vec x|$ permutations of prefixes of a list $\vec x= x_1, \dots, x_m$ (which is a polynomial of degree $m$).

\begin{proof}[Proof of~\Cref{thm:fp-soundness}.]
We proceed by induction on the definition of $f(\vec x;\vec y)$.

Each initial function is polynomial-time computable, and each (relativised) complexity class considered is under composition, so it suffices to only consider the respective recursion schemes.
We shall focus first on the case of $\bcpp (\vec a)$, \Cref{item:soundness-fp} above, so that $e_f$ is a polynomial and $d_f=1$.

Suppose we have $h(a)(\vec x;\vec y) \in \bcpp(a, \vec a)$ and let:
\[
f(\vec x;\vec y) = h(\lambda \vec u \subset \vec x, \lambda \vec v \subseteq \vec y. f(\vec u; \vec v))(\vec x;\vec y)
\]
We start by making some observations:
\begin{enumerate}
    \item\label{item:poly-growth-rate} First, note that $|f(\vec x;\vec y)| \leq e_f(|\vec x|) +d_f\sum \vec c + \max |\vec y| $, by the Bounding \Cref{lem:boundinglemma}, and so $|f(\vec x;\vec y)|$ is polynomial in $|\vec x,\vec y|$.
\item\label{item:poly-many-pp} Second, note that the set $[\vec x;\vec y] \dfn \{(\vec u,\vec v) \ | \ \vec u \permpref \vec x , \vec v \permprefeq \vec y\} $ has size polynomial in $|\vec x,\vec y|$:
\begin{itemize}
    \item write $\vec x = x_1 , \dots, x_m$ and $\vec y = y_1, \dots, y_n$.
    \item Each $x_i$ and $y_j$ have only linearly many prefixes, and so there are at most $|x_1|\cdot \cdots \cdot |x_m||y_1| \cdot \cdots \cdot |y_n| \leq \sumlen{\vec x,\vec y}^{m+n}$ many choices of prefixes for all the arguments $\vec x, \vec y$. 
    (This is a polynomial since $m$ and $n$ are global constants).
    \item Additionally, there are $m!$ permutations of the arguments $\vec x$ and $n!$ permutations of the arguments $\vec y$.
    Again, since $m$ and $n$ are global constants, we indeed have
    $|[\vec x;\vec y]| = O(\sumlen{\vec x, \vec y}^{m+n})$, which is polynomial in $|\vec x,\vec y|$.
\end{itemize}
\end{enumerate}

We describe a polynomial-time algorithm for computing $f(\vec x;\vec y)$ (over oracles $\vec a$) by a sort of `course-of-values' recursion on the order $\permpref \times \permprefeq $ on $[\vec x;\vec y]$.

First, for convenience, temporarily extend $\permpref \times \permprefeq$ to a total well-order on $[\vec x;\vec y]$, and write $S$ for the associated successor function.
    Note that $S$ can be computed in polynomial-time from $[\vec x;\vec y]$.

Define $F(\vec x,\vec y) \dfn \langle f(\vec u;\vec v)\rangle_{\vec u \permpref \vec x, \vec v \permprefeq \vec y}$, i.e.\ it is the graph of $\lambda \vec u \permpref \vec x, \lambda \vec v \permprefeq \vec y . f(\vec u;\vec v)$ that we shall use as a `lookup table'.
    Note that $|F(\vec x,\vec y)| $ is polynomial in $|\vec x,\vec y|$ by \Cref{item:poly-growth-rate} and \Cref{item:poly-many-pp} above.
    Now, we can write:\footnote{Here, as abuse of notation, we are now simply identifying $F(\vec x;\vec y)$ with $\lambda \vec u \permpref \vec x, \lambda \vec v \permprefeq \vec y . f(\vec x;\vec y)$.}
    \[
    \arraycolsep=2pt
    \begin{array}{rcl}
        F(S(\vec x,\vec y)) & =&  \langle f(S(\vec x,\vec y)), F(\vec x,\vec y) \rangle  \\
        & = & \langle h(F(\vec x ,\vec y))(\vec x;\vec y),F(\vec x,\vec y)\rangle
    \end{array}
    \]
    Again by \Cref{item:poly-many-pp} (and since $F$ is polynomially bounded), this recursion terminates in polynomial-time.
    We may now simply calculate $f(\vec x;\vec y)$ as $h(F(\vec x,\vec y))(\vec x;\vec y)$.
    
The argument for $\nbcpp$ is similar, though we need not be as careful about computing the size of the lookup tables ($F$ above) for recursive calls.
The key idea is to use the Bounding Lemma (\Cref{eq:elem-const-max-bound}) to bound the safe inputs of recursive calls so that we can adequately store the lookup table for previous values.
\end{proof}

\subsection{Completeness and characterisations}
We are now ready to give our main function algebraic characterisation results for polynomial-time:

\begin{cor}\label{cor:bcpp-bc-fptime}
The following are equivalent:
\begin{enumerate}
    \item\label{item:complfp-bc} $f(\vec x;) \in \bc$.
    \item\label{item:complfp-bcpp} $f(\vec x;) \in \bcpp$.
    \item\label{item:complfp-fp} $f(\vec x) \in \fptime$.
\end{enumerate}
\end{cor}
\begin{proof}
$\eqref{item:complfp-bc}\implies \eqref{item:complfp-bcpp}$ is trivial, and $\eqref{item:complfp-bcpp} \implies \eqref{item:complfp-fp}$ is given by \Cref{thm:fp-soundness}.\eqref{item:soundness-fp}.
Finally, $\eqref{item:complfp-fp} \implies \eqref{item:complfp-bc}$ is from \cite{BellantoniCook}, stated in \Cref{thm:bellantoni} earlier.
\end{proof}

The remainder of this subsection is devoted to establishing a similar characterisation for $\nbc$, $ \nbcpp$ and $\felementary$.

 To begin with, we recall the definition of the class $\felementary$: 

\begin{defn} \label{defn:felementary}  $\felementary$ is the smallest set of functions containing:
\begin{itemize}
    \item $0()\dfn 0 \in \Nat$,  
    \item $\pi^n_i(x_1, \ldots, x_{n})\dfn x_j$, whenever $1 \leq j \leq n$;
    \item $\mathsf{s}(x)\dfn x+1$;
    \item the function $E_2$ defined as follows:
    \begin{equation*}
    \begin{aligned}
            E_1(x)&=x^2+2\\
             E_2(0)&=2\\
             E_2(x+1)&=E_1(E_2(x))
    \end{aligned}
    \end{equation*}
    \end{itemize}
    and closed under the following:
    \begin{itemize}
    \item (Composition) If $f(\vec x, x), g(\vec x)\in \felementary$ then so is $f(\vec x, g(\vec x))$; 
    \item (Bounded recursion) If $g(\vec x), h (x,\vec x,y),  j(x, \vec x)$ are functions in $ \felementary$ then so is $f(x, \vec x)$ given by:
    \begin{equation*}
    \arraycolsep=2pt
        \begin{array}{rcl}
              f(0, \vec{x})&\dfn& g(\vec{x})\\
              f(x+1, \vec{x})&\dfn& h(x, \vec{x}, f(x, \vec{x}))
        \end{array}
    \end{equation*}
    provided that $f(x, \vec{x})\leq j(x, \vec{x})$.
\end{itemize}
\end{defn}

\begin{prop}[\cite{rose1984subrecursion}]\label{prop:rose} Let $f\in \felementary$ be a $k$-ary function. Then, there exists an integer $m$ such that:
\begin{equation*}
    f(\vec{x})\leq E^m_{2}(\max_k(\vec{x}))
\end{equation*}
where $E^0_{2}(x)=x$ and $E^{m+1}_{2}(x)=E_2(E^{m}_{2}(x))$.
\end{prop}

For our purposes we shall consider a  formulation of this class in binary notation, that we call  $\felementary_{0,1}$. 

\begin{defn} \label{defn:E3}  $\felementary_{0,1}$ is the smallest set of functions containing:
\begin{itemize}
    \item $0()\dfn 0 \in \Nat$,  
    \item $\pi^n_i(x_1, \ldots, x_{n})\dfn x_j$, whenever $1 \leq j \leq n$;
    \item $\succ i(x)\dfn 2x+i$, for $i \in \{ 0,1\}$
    \item the function $\varepsilon(x,y)$ defined as follows:
    \begin{equation*}
    \arraycolsep=2pt
    \begin{array}{rcl}
             \varepsilon(0,y)&\dfn& \succ 0 (y)\\
             \varepsilon(\succ ix,y)&\dfn& \varepsilon(x,  \varepsilon(x,y))
    \end{array}
    \end{equation*}
    \end{itemize}
    and closed under the following:
    \begin{itemize}
    \item (Composition) If $f(\vec x, x), g(\vec x)\in \felementary_{0,1}$ then so is $f(\vec x, g(\vec x))$; 
    \item (Bounded recursion on notation) If $g(\vec x), h_i (x,\vec x,y)$, $j(x, \vec x)$ are functions in  $\felementary_{0,1}$  then so is $f(x, \vec x)$ given by:
    \begin{equation*}
    \arraycolsep=2pt
        \begin{array}{rcl}
              f(0, \vec{x})&\dfn &g(\vec{x})\\
              f(\succ ix, \vec{x})&\dfn & h_i(x, \vec{x}, f(x, \vec{x}))
        \end{array}
    \end{equation*}
    provided that $f(x, \vec{x})\leq j(x, \vec{x})$.
\end{itemize}
\end{defn}

It is easy to show that the unary and the binary definition of the class of elementary time computable functions coincide. To see this, we first define $\varepsilon^n(x)$ as
\begin{equation}\label{eqn:vaerpsilon}
\arraycolsep=2pt
    \begin{array}{rcl}
   \varepsilon^1(x)  &\dfn&  \varepsilon(x,1)  \\
   \varepsilon^{m+1}(x)& \dfn & \varepsilon^1(\varepsilon^{m}(x))   
\end{array}
\end{equation}
which allows us to prove that $\varepsilon^n(x)$ plays the role of rate growth function as the function $E_2^n(x)$ (see Proposition~\ref{prop:rose}). 

\gianluca{Stress the equivalence between  binary notation and unary notation for elementary time functions is folklore?}

\begin{prop}\label{prop:EsubseteqEbin}{\ }
\begin{enumerate}
    \item \label{enum:EsubseteqEbin1} $\felementary = \felementary_{0,1}$;
    \item \label{enum:EsubseteqEbin2} for any $f \in \felementary_{0,1}$ $k$-ary function there is an integer $m$ such that:
\begin{equation*}
    f(\vec{x})\leq \varepsilon^m(\max_k(\vec{x}))
\end{equation*}
\end{enumerate}
\end{prop}
\begin{proof}
 Let us first prove point~\ref{enum:EsubseteqEbin1}. For the $\supseteq$ direction we show that for any $f  \in \felementary_{0,1}$ there exists 
 $n$ such that:
\begin{equation}\label{eqn:feleementary01-bound}
    \s{f(\vec x)}\leq 2_n(\sum\s{\vec x}) 
\end{equation}
where $2_0(x)=x$ and  $2_{n+1}(x)=2^{2_n(x)}$. 
 Since $\s{x}=\lceil \log_2(x+1)\rceil$, from the above inequation we would have that, for some $m$:
\begin{equation*}
    f(\vec x)\leq 2_{n+m}(\sum \vec x)
\end{equation*}
 which allows us to conclude $f  \in \felementary$, as  the elementary time computable functions are exactly the elementary space  ones. The inequation~\eqref{eqn:feleementary01-bound} can be proved by induction on $f$,  noticing that $\s{\varepsilon(x, y)}= 2^{\s{x}}+ \s{y}$. Concerning the $\subseteq$ direction, we prove  by induction on  $f\in \felementary$ that there exists a function $\hat{f}\in \felementary_{0,1}$ such that, for all $\vec x= x_1, \ldots, x_n$:
\begin{equation*}
   f(\vec x)= \s{\hat{f}(\num{x_1}, \ldots, \num{x_n})}
\end{equation*}
where $\num{m}\dfn \succ 1 ^m(0)=\succ 1(\overset{m}{\ldots}  \succ 1(0))$.  Since the functions $\s{\cdot}$ and $x \mapsto \num{x}$ are both in $\felementary_{0,1}$, we are able to conclude $f \in \felementary_{0,1}$. The case $f=0$ is trivial. As for the cases $f=\mathsf{s}$ and $f= \pi^n_i$, we first notice that $\s{\num{x}}=x$. Then, we have: 
\begin{equation*}
\def\arraystretch{1.5}
\begin{array}{c}
\mathsf{s}(x)=\s{\num{s(x)}}
\\
  \pi^n_i(\vec x)=x_i=\s{\num{x_i}}=  \s{\pi^n_i(\num{x_1}, \ldots, \num{x_n})}
\end{array}
\end{equation*}
Concerning the case of $E_2(x)$, we first notice that the following property holds for any $m$ and some $k$:
 \begin{equation}\label{eqn:boundgrowthrate}
     E^m_2(x)\leq \varepsilon^{m+k}(x)
 \end{equation}
where $\varepsilon^{n}(x)$ is as in~\eqref{eqn:vaerpsilon}. Moreover, the  function $E_2(x)$ can be defined by two applications of bounded recursion proceeding from the successor and the projection functions, where each recursion can be bounded by $E_2(x)$, and hence by $\varepsilon^{k}(x)$ for some $k$. This means that the case of  $E_2(x)$ can be reduced to the case of bounded recursion. Suppose now that $f(\vec x)=h(\vec x, g(\vec x))$.  We define  $\hat{f}(\vec x)= \hat{h}(\vec x, \hat{g}(\vec x))$ so that, by induction hypothesis:
\begin{equation*}
\arraycolsep=2pt
    \begin{array}{rcl}
     f(\vec x)&=& h(\vec x, g(\vec x))\\
     &=&\s{\hat{h}(\num{x_1}, \ldots, \num{x_n}, \num{g(\vec x)})}\\
     &=&\s{\hat{h}(\num{x_1}, \ldots, \num{x_n}, \num{\s{\hat{g}(\num{x_1}, \ldots, \num{x_n})}})}\\
      &=&\s{\hat{h}(\num{x_1}, \ldots, \num{x_n}, \hat{g}(\num{x_1}, \ldots, \num{x_n}))}\\
      &=& \s{\hat{f}(\vec x)}= \s{\hat{h}(\vec x, \hat{g}(\vec x))}
    \end{array}
\end{equation*}
Last, suppose that $f$ has been obtained by  bounded recursion from $h, g, j$, i.e.
    \begin{equation*}
    \arraycolsep=2pt
        \begin{array}{rcl}
              f(0, \vec{x})&=&g(\vec{x})\\
              f(y+1, \vec{x})&=& h(y, \vec{x}, f(y, \vec{x}))
        \end{array}
    \end{equation*}
    provided that $f(y, \vec{x})\leq j(y, \vec{x})$. We define $\hat{f}$ as follows:
    \begin{equation*}
    \arraycolsep=2pt
        \begin{array}{rcl}
              \hat{f}(0, \vec{x})&=&\hat{g}(\vec{x})\\
              \hat{f}(\succ iy, \vec{x})&=& \hat{h}(y, \vec{x}, \hat{f}(y, \vec{x}))
        \end{array}
    \end{equation*}
We show by induction on $y$ that:     
    \begin{equation*}
   f(y, \vec x)= \s{\hat{f}(\num{y}, \num{x_1}, \ldots, \num{x_n})}
\end{equation*}
We have:
\begin{equation*}
\arraycolsep=2pt
    \begin{array}{rcl}
  f(0, \vec x)&=& g(\vec x)\\
  &=&\s{\hat{g}(\num{x_1}, \ldots, \num{x_n})}\\
  &=&
    \s{\hat{f}(\num{0}, \num{x_1}, \ldots, \num{x_n})}\\ \\
    f(y+1, \vec x)&=& h(y,\vec x, f(y, \vec x))\\
    &=& \s{\hat{h}(\num{y}, \num{x_1}, \ldots, \num{x_n}, \num{f(y, \vec x)})}\\
    &=& \s{\hat{h}(\num{y}, \num{x_1}, \ldots, \num{x_n}, \num{\s{\hat{f}(\num{y}, \num{x_1}, \ldots, \num{x_n})}})}\\
     &=& \s{\hat{h}(\num{y}, \num{x_1}, \ldots, \num{x_n}, \hat{f}(\num{y}, \num{x_1}, \ldots, \num{x_n})}\\
    &=& \s{\hat{f}(\num{y+1}, \num{x_1}, \ldots, \num{x_n})}
   \end{array}
\end{equation*}
since    $f(y, \vec{x})\leq j(y, \vec{x})$ and $ j(y, \vec x)= \s{\hat{j}(\num{y}, \num{x_1}, \ldots, \num{x_n})}$  by induction hypothesis, we are done.

     Point~\ref{enum:EsubseteqEbin2} follows by point~\ref{enum:EsubseteqEbin1},  Proposition~\ref{prop:rose} and~\eqref{eqn:boundgrowthrate}.
\end{proof}

Completeness for $\nbc$ is based on a standard technique (see~\cite{BellantoniCook}),  adapted to the case of $\felementary$  in~\cite{Wirz99characterizingthe}.

\begin{lem}\label{lem:padding}
For any $f(\vec{x}) \in \felementary_{0,1}$  there are a function $f^*(x; \vec{x})\in \nbc$ and a monotone function $t_f\in \felementary_{0,1}$ such that for all integers $\vec{x}$ and all $w \geq t_f(\vec{x})$ we have $f^*(w; \vec{x})=f(\vec{x})$.
\end{lem}
\begin{proof}
The proof is by induction on the definition of $f$.  If $f$ is the zero, successor or projection  function then $f^*\in \nbc$.  In this case we choose $t_f=0$. The function $\varepsilon$ has a definition by one application of bounded recursion on notation, proceeding from the successors and the projection functions, where each recursion is bounded by $\varepsilon$. Since the treatment of bounded recursion  does not make use of the induction hypothesis for the bounding function, we can use this method to get functions $\varepsilon^*\in \nbc$ and $t_{\varepsilon}\in \felementary_{0,1}$ with the required properties.  If $f(\vec x)= h(\vec x, g(\vec{x}))$ then we set  $f^*(w;\vec{x})=h^*(w; \vec x, g^*(w;\vec{x}))$, which is in $\nbc$.  Since the function $g^*$ is clearly bounded by a monotone function $b \in \felementary_{0,1}$, we set $ t_f(\vec{x})= t_h(\vec x, b(\vec{x}))+  t_{g}(\vec{x})$,  which is monotone. By applying the induction hypothesis,  if   $w \geq t_f(\vec{x})$ then:
\begin{equation*}
         f^*(w; \vec{x})=  h^*(w; \vec{x}, g^*(w; \vec{x})) 
         =h^*(w;\vec x, g(\vec{x})) 
         = h(\vec x, g(\vec{x}))
\end{equation*}
Let us finally suppose that $f(x, \vec{y})$ is defined by bounded recursion on notation from $g(\vec y)$, $h_i(x, \vec y, f(x, \vec y))$ and $j(x, \vec y)$. By applying the induction hypothesis we set:
\begin{equation*}
\arraycolsep=2pt
    \begin{array}{rcl}
     \hat{f}(0,w; x, \vec{y})&=& g^*(w; \vec{y})\\
      \hat{f}(\succ i(x),w; x, \vec{y})&=& \cnd(; W(\succ i(v),w;x), g^*(w; \vec{y}),  h^*_i(w; W(v,w;x), \vec{y}, f^*(v,w;x,\vec{y})) )\\
     f^*(w;x, \vec{y})&= & \hat{f}(w,w; x, \vec{y}) 
    \end{array}
\end{equation*}
where $ W(v,w;x)=\dot-( \dot-(v;w) ; x)$ and $\dot-(x;y)$ is the truncated subtraction, which is in $\bc$, and hence in  $\nbc$. We can easily show that $f^*(w;x, \vec y)\in \nbc$. We define $t_f(x, \vec{y})=t_{g}(\vec{y})+\sum_i t_{h_i}(x, \vec{y}, j(x, \vec{y}))$, where $j$ is the bounding function. Assuming $j$ to be monotone, $t_f$ is monotone too.  We now show by induction on $u$ that, whenever $w \geq t_f(x,\vec{y})$ and $w-x\leq u \leq w$:
\begin{equation}
\label{eqn:toshowpadding}
    \hat{f}(u,w;x,\vec{y})=f(x-(w-u), \vec{y})
\end{equation}
If $u=w-x$ then we have two cases:
\begin{itemize}
    \item if $u=0$ then $\hat{f}(0, w; x, \vec{y})=g^*(w; \vec{y})$;
    \item if $u=\succ i(v)$ then, since $W(\succ i(v),w;x)=0$, we have $\hat{f}(\succ i(v), w; x, \vec{y})=g^*(w; \vec{y})$.
\end{itemize}
Hence, in any case:
\begin{equation*}
\hat{f}(u, w; x, \vec{y})=g^*(w; \vec{y})=g(\vec y)=f(0; \vec{y})= f(x-(w-u), \vec{y})    
\end{equation*}
Let us now suppose that $w-x<u\leq w$. This means that $u=\succ i(v)$ and $W(\succ i(v),w;x)>0$. Moreover,  by monotonicity of $t_f$ and definition of $j$:
\begin{equation*}
\arraycolsep=2pt
\begin{array}{rcl}
     w &\geq&     t_f(x, \vec{y})\\ 
     &\geq& t_f(x-(w-\succ i(v)), \vec{y}) \\
     &\geq& t_f(x-(w-v), \vec{y}) \\
     &\geq& t_{h_i}(x-(w-v), \vec{y}, j(x-(w-v), \vec{y}))\\
      &\geq& t_{h_i}(x-(w-v), \vec{y}, f(x-(w-v), \vec{y}))
\end{array}
\end{equation*}
By applying the induction hypothesis:
\begin{equation*}
\arraycolsep=2pt
\begin{array}{rcl}
      \hat{f}(\succ i(v), w; x, \vec{y})&=&
        h^*_i(w; W(v, w;x), \vec{y}, \hat{f}(v,w;x, \vec{y}))\\
         &=&  h^*_i(w; W(v, w;x), \vec{y}, f(x-(w-v), \vec{y}))\\
         &=& h^*_i(w; x-(w-v), \vec{y}, f(x-(w-v), \vec{y}))\\
         &=& h_i(x-(w-v), \vec{y}, f(x-(w-v), \vec{y}))\\
           &=&f(\succ i((x-(w-v))), \vec{y})\\
           &=& f((x-(w-\succ i(v))), \vec{y})
\end{array}
\end{equation*}
Now, by~\eqref{eqn:toshowpadding}, for all $w \geq t_f(x, \vec{y})$ we have:
\begin{equation*}
  f^*(w; x, \vec{y})=  \hat{f}(w,w;x, \vec{y})= f(x, \vec{y})
\end{equation*}
and this concludes the proof.
\end{proof}

\begin{thm}
\label{thm:elementary-in-nbc}
If $f(\vec x) \in \felementary$ then $f(\vec x;) \in \nbc$.
\end{thm}
\begin{proof}
First, given the function $\ex(x; y)$ in~\eqref{eq:safe-exp-by-snrec}, we construct the function $\ex^m(x;)$ by induction on $m$:
\begin{equation*}
\arraycolsep=2pt
    \begin{array}{rcl}
     \ex^1(x;)&=& \ex(x;1)\\
       \ex^{m+1}(x;)&=&   \ex^{m}(  \ex^{1}(x;);)
    \end{array}
\end{equation*}
Hence  $\varepsilon^m(x)=\ex^m(x;)$, for all $m\geq 1$. Now, let $f(\vec x) \in \felementary$. By Proposition~\ref{prop:EsubseteqEbin}.\ref{enum:EsubseteqEbin1}, we have that $f(\vec x) \in \felementary_{0,1}$. By Lemma~\ref{lem:padding},   there exist $f^*(w; \vec x)\in\nbc$ and a monotone function $t_{f}\in \felementary_{0,1}$  such that, for all $w, \vec{x}$ with  $w \geq t_f(\vec{x})$, it holds that $f^*(w; \vec{x})= f(\vec{x})$. By Proposition~\ref{prop:EsubseteqEbin}.\ref{enum:EsubseteqEbin2}  there exists $m \geq 1$ such that:
\begin{equation*}
    t_f(\vec{x})\leq  \ex^{m}(\max_{\sharp \vec{x}}(\vec{x};);) 
\end{equation*}
where $\max_{\sharp \vec{x}}(\vec{x};)$ is the $k$-ary maximum function, which is in $\bc$ by Theorem~\ref{thm:bellantoni}, and hence in $\nbc$. Therefore:
\begin{equation*}
    f(\vec{x};)=f^*(\ex^{m}(\max_{\sharp \vec{x}}(\vec{x};););\vec{x})\in \nbc
\end{equation*}
\end{proof}

Now, by the same argument as for \Cref{cor:bcpp-bc-fptime}, only using \Cref{lem:padding} above instead of appealing to \cite{BellantoniCook}, we can give our main characterisation result for algebras for elementary computation:

\begin{cor}\label{cor:nb-elementary-characterization} The following are equivalent:
\begin{enumerate}
    \item\label{enum:elementary1} $f(\vec x;)\in \nbc$.
    \item \label{enum:elementary2}  $f(\vec x;)\in \nbcpp$. 
    \item \label{enum:elementary3}  $f(\vec x)\in \felementary$.
\end{enumerate}
\end{cor}
  \section{Characterisations for circular systems}
\label{sec:completeness}
We now return our attention to the circular systems $\cbc$ and $\ncbc$ that we introduced in \Cref{sec:two-tiered-circular-systems-on-notation}. 
We will address the complexity of their definable functions by `sandwiching' them between function algebras of \Cref{sec:some-variants}, given their characterisations that we have just established.


\subsection{Completeness}
To show that $\cbc$ contains all polynomial-time functions, we may simply simulate Bellantoni and Cook's algebra:

\begin{thm}
\label{thm:bc-in-cbc}
If $f(\vec x;\vec y) \in \bc$ then $f(\vec x; \vec y)\in \cbc$.
\end{thm}
\begin{proof}
By Proposition~\ref{prop:bc-type-system-characterisation} it suffices to show that for any $\bc$-derivation $\der$  there is a $\cbc$-coderivation $\der^*$  such that $\denot\der(\vec x; \vec y)=\denot{\der^*}(\vec x,  \vec y)$. The proof is by induction on $\der$. The only non-trivial case is when   $\der$ is the following derivation:
\[
\small
 \vlderivation{
    \vliiin{\srec}{}{\sn, \Gamma \seqar \n}{
        \vltr{\der_0}{\Gamma \seqar \n}{\vlhy{\ }}{\vlhy{\ }}{\vlhy{\ }}
    }{
        \vltr{\der_1}{\sn, \Gamma, \n \seqar \n}{\vlhy{\ }}{\vlhy{\ }}{\vlhy{\ }}
    }{
        \vltr{\der_2}{\sn, \Gamma, \n \seqar \n}{\vlhy{\ }}{\vlhy{\ }}{\vlhy{\ }}
    }
    }
\]
We define $\der^*$ as follows:
\[
\small
\vlderivation{
\vliin{\cnd_\sq}{\bullet}{\blue{\underline{\sn}},\Gamma \seqar \n}{
\vltr{\der_0^*}{\Gamma \seqar \n}{\vlhy{\ }}{\vlhy{\ }}{\vlhy{\ }}
}{
\vliin{\cut_{\n}}{{\scriptstyle i=1,2}}{\blue{\sn}, \Gamma  \seqar \n}{
   \vlin{\cnd_{\sq}}{\bullet}{
   \blue{\sn}, \Gamma  \seqar \red{\n}
   }{
   \vlhy{\vdots}
   }
}{
   \vltr{\der^*_i}{\blue{\sn}, \Gamma , \red{\n} \seqar \n}{\vlhy{\ }}{\vlhy{\ }}{\vlhy{\ }}
}
}
}
\]
where $\Gamma= \sq \vec \n, \vec \n$ and we identify the coderivations corresponding to the second and the third premise of the conditional rule, as they only differ on the sub-coderivation $\der^*_{i}$ ($i=1$ for the former and $i=2$ for the latter). 

The above coderivation is clearly safe and left-leaning, by the inductive hypotheses for $\der_0,\der_1,\der_2$. 
{To see that it is progressing, note that any infinite branch is either eventually entirely in $\der_0^*$, $\der_1^*$ or $\der_2^*$, in which case it is progressing by the inductive hypotheses, or it simply loops on $\bullet$ forever, in which case there is a progressing thread along the blue $\blue{\sn}$.} 

Moreover, the equational program associated with $\der^*$  is equivalent to:
 \[
 \arraycolsep=2pt
    \begin{array}{rcll}
         \denot{\der^*_\epsilon} (0, \vec x;\vec y) & \dfn &  \denot{\der^*_0} (\vec x;\vec y) \\
         \denot{\der^*_\epsilon} (\succ 0 x, \vec x; \vec y) & \dfn & \denot{\der^*_1} (x, \vec x; \vec y, \denot{\der^*_\epsilon}(x, \vec x; \vec y)) & \text{if $x \neq 0$} \\
         \denot{\der^*_\epsilon} (\succ 1 x , \vec x; \vec y) & \dfn & \denot{\der^*_2} (x, \vec x; \vec y,  \denot{\der_\epsilon}(x, \vec x; \vec y)) 
    \end{array}
    \]
so that $\denot\der(\vec x; \vec y)=\denot{\der^*}(\vec x,  \vec y)$.
\end{proof}

We can also show that $\ncbc$ is complete for elementary functions by  simulating  our nested algebra $\nbc$. 
First, we need to introduce the notion of \emph{oracle} for coderivations.\footnote{Notice that a similar notational convention discussed in~\Cref{rem:variables-vs-functions} applies \emph{mutatis mutandis} to~\Cref{defn:derivation-oracles}.}

\begin{defn}
[Oracles for coderivations]
\label{defn:derivation-oracles}
Let $\vec{a}=a_1,\ldots, a_n$ be a set of safe-normal functions. A $\bcnorec(\vec a)$-coderivation is just a usual $\bcnorec$-coderivation that may use initial sequents of the form
$\vlinf{a_i}{}{\sn^{n_i},\n^{m_i}\seqar \n }{}$, when $a_i $ takes $n_i$ normal and $m_i$ safe inputs. 
We write:
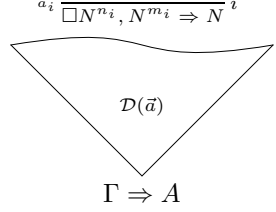
\begin{equation*}
    \toks0={0.5}
    \vlderivation{
    \small
\vltrf{\mathcal{D}(\vec a)}{\Gamma\seqar A}{\vlhy{}}{\scriptsize \vlin{a_i}{i}{\sn^{n_i},\n^{m_i}\seqar \n }{\vlhy{\ }}}{\vlhy{\ }}{\the\toks0}
}
\end{equation*}
for a coderivation $\der$ whose initial sequents are among the initial sequents  $\vlinf{a_i}{}{\sn^{n_i},\n^{m_i}\seqar \n }{}$, with $i=1,\dots, n$.  We write $ \ncbc(\vec{a})$ for the set of $\ncbc$-coderivations  with initial functions $\vec{a}$. 
We may sometimes omit indicating some oracles $\vec a$ if it is clear from context.

The semantics of such coderivations  and the notion of $\ncbc(\vec a)$-definability are as expected, with coderivations representing functions over the oracles $\vec a$, and   $\denot{\der(\vec a)}\in \ncbc(\vec{a})$ denoting the induced interpretation of $\der(\vec a)$. 
\end{defn}

Before giving our main completeness result for $\ncbc$, we need the following lemma allowing us to `pass' parameters to oracle calls.
It is similar to the notion $\der^{\vec \rho}$ from \cite[Lemma 42]{Das2021}, only we must give a more refined argument due to the unavailability of contraction in our system.

\begin{lem}\label{lem:to-prove-ncbc-complete-for-nbc}
Let $\der(a)$ be a regular coderivation over initial sequents $\vec a, a$ of the form:
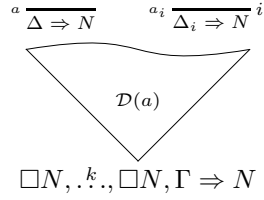
\[
\toks0={0.5}
\vlderivation{
\small
\vltrf{\der(a)}{\sn, \overset{k}{\ldots},\sn, \Gamma  \seqar \n}{\scriptsize \vlin{a}{}{    \Delta \seqar \n}{\vlhy{\ }}}{\vlhy{}}{\scriptsize \vlin{a_i}{i}{ \Delta_i \seqar \n}{\vlhy{\ }}}{\the\toks0}
}
\]
where $\Gamma$ and $\Delta$ are lists of non-modal formulas, and   the path from the conclusion to each initial sequent $a$ does not contain $\cut_\sq$-steps, $\sql$-steps and the leftmost premise of a $\cnd_\sq$-step. 
Then, there exists an $a^*$ and a regular coderivation   $\der^*( a^*)$  with shape:
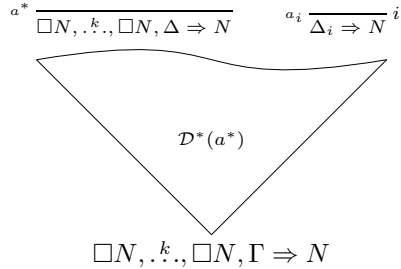
\[
\toks0={0.5}
\vlderivation{
\small
\vltrf{\der^*( a^*)}{\sn, \overset{k}{\ldots},\sn, \Gamma \seqar \n}{\scriptsize \vlin{a^*}{}{ \sn, \overset{k}{\ldots},\sn,\Delta  \seqar \n}{\vlhy{\ }}}{\vlhy{}}{\scriptsize \vlin{a_i}{i}{  \Delta_i \seqar \n}{\vlhy{\ }}}{\the\toks0}
}
\]
such that:
\begin{itemize}
\item  $\denot{\der^*( a^*)}( \vec x; \vec y)=\denot 
     {\der( a( \vec x))}(\vec x; \vec y) $;
    \item there exists a $\sn$-thread from the $j$\textsuperscript{th} $\sn$ in the LHS of the end-sequent to the $j$\textsuperscript{th} $\sn$ in the context of any occurrence of the initial sequent $a^*$ in $\der^*( a^*)$, for $1\leq j\leq k$.
    \end{itemize}
    Moreover, if $\der(a)$ is progressing, safe or left-leaning, then $\der^*(a^*)$ is also progressing, safe or left-leaning, respectively.
\end{lem}
\begin{proof}[Proof sketch]
Let us consider a path $B$ from the root of $\der( a)$ to an occurrence of the initial sequent $a$. 
Since the conclusion of   $\der( a)$ contains $k$ modal formulas, and $B$ cannot cross $\cut_\sq$-steps, $B$ contains exactly $k$ $\sn$-threads, and all such threads  start from the root. 
Moreover, since $a$ has only non-modal formulas and $B$ cannot cross $\sql$-steps or the leftmost premise of a $\cnd_\sq$-step, we conclude that each of the $k$  $\sn$-threads must end in the principal formula of a $\wk_\sq$-step. 
For each $j$ we remove the corresponding  $\wk_\sq$-step in $B$  we add an extra $\sn$ to the antecedent of all higher sequents in $B$ (this operation may require us to introduce weakening steps for other branches of the proof).  
By repeatedly applying the above procedure for each possible path  from the root of $\der(a)$ to an initial sequent $a$ we obtain a coderivation with the desired properties.
\end{proof}

Finally we can give our main simulation result for $\ncbc$:

\begin{thm}
\label{thm:nbc-in-ncbc}
If $f(\vec x;\vec y) \in \nbc$ then $f(\vec x; \vec y)\in \ncbc$.
\end{thm}
\begin{proof}
We show by induction on $f(\vec x;\vec y)\in \nbc(\vec a)$ that there is a $ \ncbc(\vec{a})$-coderivation $\der_f$ such that:
\begin{enumerate}
    \item $\denot{\der_f} (\vec x;\vec y) = f(\vec x;\vec y)$;
    \item \label{enum:construction} the path from the conclusion of $\der_{f}$ to each initial sequent $a_i$ does not  contain $\cut_\sq$-steps, $\sql$-steps and the leftmost premise of a $\cnd_\sq$-step.
\end{enumerate}
When $f(\vec x;\vec y)$ is an initial function the definition of $\der_f$ is  straightforward, as $\vec a= \emptyset$ and $\der_f$ can be constructed as in~\Cref{prop:bc-type-system-characterisation} . 
If $f(\vec x;\vec y)=a_i(\vec x;\vec y)$ then $\der_f$ is the  initial sequent  $
 \vlinf{a_i}{}{\sq \vec \n , \vec \n\seqar \n }{}
$. 

Suppose that
 $f(\vec x; \vec y)= h(\vec x, g(\vec x;); \vec y)$ with $g(\vec x;)\in \nbc(\emptyset)$
 and
 $h(\vec x, z; \vec y)\in \nbc(\emptyset)$.  Then $f$ can be  $\ncbc(\emptyset)$-defined by:
\[
\vlderivation{
\small
\vliin{\cut_\sq}{}
{\sq \vec \n, \vec \n \seqar \n}
{\vlin{\sqr}{}{\sq \vec \n\seqar \red{\sn}}{\vltr{\der_{g}}{\sq \vec \n\seqar \n}{\vlhy{\ }}{\scriptsize  \vlhy{\ \ \  \emptyset\ \ \ }  }{\vlhy{\ }}}}
{\vltr{\der_{h}}{\red{\sn}, \sq \vec \n, \vec \n \seqar \n}{\vlhy{\ }}{\scriptsize  \vlhy{\ \ \  \emptyset\ \ \ } }{\vlhy{\ }}}
}
\]
Note that, while we introduce a $\cut_\sq$ here, there crucially remains no $\cut_\sq$ between the conclusion and an oracle sequent, thanks to the condition that $g$ and $h$ are over no oracles.

Suppose $f(\vec x; \vec y)= h(\vec x; \vec y, g(\vec x; \vec y))$. 
Then $f$ is $\ncbc(\vec a)$-defined by:
\[
\vlderivation{
\small
\vliin{\cut_\n}{}
{\sq \vec \n, \vec \n \seqar \n}
{\vltr{\der_g}{\sq \vec \n, \vec \n \seqar \red{\n}}{\vlhy{\ }}{\scriptsize \vlin{a_i}{i}{\sq \vec \n, \vec \n \seqar \n}{\vlhy{\ }}}{\vlhy{\ }}}
{\vltr{\der_h}{ \sq \vec \n, \vec \n, \red{\n} \seqar \n}{\vlhy{\ }}{\scriptsize  \vlin{a_i}{i}{\sq \vec \n,  \vec \n \seqar \n}{\vlhy{\ }}}{\vlhy{\ }}}
}
\]

Finally, suppose that $f(x, \vec x; \vec y)$ is obtained from $g(\vec x; \vec y)$ {over $\emptyset$}, and  $h_j( a)(x,\vec x;\vec y)$ ($j=0,1$) over $a,\vec a$ by $\snrec$. Then,  $f(0, \vec x; \vec y) = g(\vec x;\vec y)$ and $f(\succ j x, \vec x; \vec y) = h_j(\lambda \vec v . f(x,\vec x;\vec v))(x,\vec x;\vec y)$. 
Note that $a$ has same type as $\lambda \vec v . f(x,\vec x;\vec v)$, so it is a function taking safe arguments only. 
Thus by induction hypothesis, $h_j(a)(x,\vec x;\vec y)$ is $\ncbc(\vec a, a)$-defined by:
\[
\vlderivation{
\small
\vltr{\der_{h_j}(a)}{ \sn, \sq \vec \n, \vec \n \seqar \n}{\scriptsize \vlin{a}{}{ \vec \n \seqar \n}{\vlhy{\ }}}{\vlhy{}}{\scriptsize \vlin{a_i}{i}{   \sq \vec \n , \vec \n \seqar \n}{\vlhy{\ }}}
}
\]
where the path from the conclusion of $\der_{h_j}( a)$ to each initial sequent $a_i$ does not  contain $\cut_\sq$-steps, $\sql$-steps and the leftmost premise of a $\cnd_\sq$-step. By Lemma~\ref{lem:to-prove-ncbc-complete-for-nbc} we obtain the following coderivation:
\[
\vlderivation{
\small
\vltr{\der_{h_j}^*( a^*)}{\sn, \sq \vec \n, \vec \n \seqar \n}{\scriptsize \vlin{a^*}{}{\sn, \sq \vec \n , \vec \n \seqar \n}{\vlhy{\ }}}{\vlhy{}}{\scriptsize \vlin{a_i}{i}{\sn, \sq \vec \n, \vec \n \seqar \n}{\vlhy{\ }}}
}
\]
where $\denot{\der_{h_j}^*( a^*)}(x, \vec x; \vec y)=\denot {\der_{h_j}( a(x, \vec x))}(x, \vec x; \vec y) $ and there exists a $\sn$-thread from the $k$-th modal formula in the antecedent of the conclusion to the $k$-th modal formula in the antecedent of any $a^*$ initial sequent in $\der^*(a^*)$.  
We define $\der_f$ as follows: 
\[
\vlderivation{
\small
\vliin{\cnd_{\sq}}{\bullet}{\underline{\blue{\sn}},  \sq \vec \n, \vec \n \seqar \n}
{
\vltr{\der_g}{ \sq \vec \n, \vec \n \seqar \n}{\vlhy{\ }}{\vlhy{\scriptsize  \ \emptyset\ \ \ }}{\vlhy{}}
}
{\vltr{\der_{h_j}^*{(\der_f)}}{\blue{\sn},  \sq \vec \n, \vec \n \seqar \n}{\scriptsize \vlin{\cnd_{\sq}}{\bullet }{\blue{\sn},  \sq \vec \n, \vec \n \seqar \n}{\vlhy{\vdots }}}{\vlhy{}}{\scriptsize \vlin{a_i}{i}{\sn,  \sq \vec \n, \vec \n \seqar \n}{\vlhy{\ }}}}
}
\]
where we identify the sub-coderivations corresponding to the second and the third premises of the conditional rule. 
By construction and induction hypothesis, the above coderivation is regular and safe. 
{To see that it is progressing, note that any infinite branch $B$ either hits $\bullet$ infinitely often, in which case there is a progressing thread along the blue $\blue{\sn}$, by the properties of $\der_{h_j}^*$ inherited from \Cref{lem:to-prove-ncbc-complete-for-nbc}, or $B$ shares a tail with an infinite branch of $\der_g$ or $\der_{h_j}^*(a^*)$, which are progressing by the inductive hypotheses.}

We show that $\der_f$ $\ncbc(\vec a)$-defines $f$  by induction on $x$. For the base case, $ \denot{\der^*_f}(0, \vec x; \vec y)=   \denot{\der_g}(\vec x; \vec y)= g(\vec x; \vec y)= f(0, \vec x; \vec y)$. For the inductive step:
\[
\renewcommand{\arraystretch}{1.5}
\arraycolsep=2pt
\begin{array}{rcl}
        \denot{\der^*_f}(\succ jx, \vec x; \vec y)&=&  \denot{\der^*_{h_j}(\lambda \vec v.  \denot{\der^*_f}(x, \vec x; \vec v) )}(x, \vec x; \vec y)  \\ 
     & =& \denot{\der^*_{h_j}( \lambda \vec v.  f(x, \vec x; \vec v) )}(x, \vec x; \vec y) \\
       & =& \denot{\der_{h_j}( \lambda \vec v.  f(x, \vec x; \vec v) )}(x, \vec x; \vec y) \\
       & =& h_j( \lambda \vec v.  f(x, \vec x; \vec v) )(x, \vec x; \vec y)\\
       &=& f(\succ j x, \vec x; \vec y) .
\end{array}
\]
This completes the proof.
\end{proof}

\subsection{The Translation Lemma} \label{sec:translation}

 
The goal of this section is to prove the following theorem, which shows  a translation of  $\ncbc$-coderivations into   functions of $\nbcpp$   mapping, in particular,    $\cbc$-derivations into functions of $\bcpp$:

\begin{thm}[Translation Lemma]\label{lem:short-translation-lemma}
Let $\der$ be a  $\ncbc$-coderivation. Then, there exists a set of $n$ functions $(f_i)_{1 \leq i \leq n}$ such that $f_{1}= \model{\der}$ and, for all $i$:
\begin{equation}\label{eqn:naive1}
\arraycolsep=2pt
\begin{array}{rcl}
 f_i(\vec{x}; \vec{y})   & = & h_i\left( \lambda \vec{u}\subset \vec{x},\lambda \vec v. f_j(\vec{u};\vec v)\right)_{1 \leq j \leq n}(\vec{x}; \vec{y})
\end{array}
\end{equation}
where  $h_i \in \nbcpp (a_i)_{1 \leq i \leq n}$. Moreover, if    $\der$ is a $\cbc$-coderivation  then, for all $i$:
\begin{equation}\label{eqn:naive2}
\arraycolsep=2pt
\begin{array}{rcl}
 f_i(\vec{x}; \vec{y})   & = & h_i\left( \lambda \vec{u}\subset \vec{x},\lambda \vec v\subseteq \vec y. f_j(\vec{u};\vec v)\right)_{1 \leq j \leq n}(\vec{x}; \vec{y})
\end{array}
\end{equation}
and  $h_i \in \bcpp  (a_i)_{1 \leq i \leq n}$.
\end{thm}

This would conclude  our characterisation of $\cbc $ and $\ncbc$ in terms of computational complexity. 

A formal proof of the translation lemma requires a more general statement (i.e.,~\Cref{lem:translation}) as well as a few technical notions and intermediate results. 
First, we observe that a regular coderivation can be naturally seen as a finite tree with `backpointers', a representation known as \emph{cycle normal form}, cf.~\cite{Brotherston05,brotherston2011sequent}.

\begin{defn}[Cycle normal form]
Let $\der$ be a regular $\bcnorec$-coderivation. 
The  \emph{cycle normal form} (or simply \emph{cycle nf}) of $\der$ is a pair $\derrd$, where  $\rd$ is a partial self-mapping on the nodes of $\der$ whose domain of definition is denoted $\bud(\der)$ and:
\begin{enumerate}[(i)]
    \item every infinite branch of $\der$ contains some (unique) $\nu \in \bud(\der)$;
    \item if $\nu \in \bud(\der)$ then both $\rd(\nu) \sqsubset \nu$ and  $\der_{\rd(\nu)} = \der_\nu$;
    \item for any two distinct nodes $\mu \sqsubset \nu$ strictly below $\bud(\der)$, $\der_\mu \neq \der_\nu$
\end{enumerate}

We call any $\nu \in \bud(\der)$ a \emph{bud}, and $\rd(\nu)$ its \emph{companion}. A \emph{terminal} node is either a leaf of $\der$ or a bud. The set of nodes of $\der$ bounded above by a terminal node is denoted $T_\der$.  
Given a node $\nu\in T_\der$, we  define $\bud_\nu(\der)$ as the restriction of buds to those above $\nu$.
\end{defn}

\begin{rem}\label{rem:canonicalnodesfinite}
The cycle normal form of a regular coderivation $\der$ always exists, as  by definition any infinite branch contains a node $\nu$ such that $\pder \nu=\pder \mu$ for some node $\mu$ below $\nu$. 
$\bud(\der)$ is designed to consist of just the \emph{least} such nodes, so that by construction the  cycle normal form is unique. 
Note that $\bud(\der)$ must form an $\sqsubset$-antichain: if $\mu, \nu \in \bud(\der)$ with $\mu\sqsubset \nu$, then $\rd(\mu)\sqsubset \mu $ are below $\bud(\der)$ but we have $\der_{\rd(\mu)} = \der_\mu$ by point (ii) above, contradicting point (iii). 

Also, notice that any branch of $\der$ contains a leaf of $T_\der$. Moreover, since $\bud(\der)$ is an antichain, the leaves of $T_\der$ defines a ‘bar’ across  $\der$,
and so $T_\der$ is a finite tree. 
\end{rem}


The following proposition allows us to reformulate   progressiveness, safety and left-leaning conditions for cycle normal forms. 
\begin{prop}\label{prop:structureofcycles}
Let $\der$ be a regular   $\bcnorec$-coderivation with cycle nf $\derrd$. 
For any  $\nu \in \bud(\der)$, the (finite) path $\pi$ from $\rd(\nu)$ to $\nu$ satisfies:
\begin{enumerate}
\item \label{enum:budcompanion0} if  $\der$ is progressing, $\pi$ must contain the conclusion of an instance of $\cnd_{\sn}$;
    \item \label{enum:budcompanion1} if  $\der$ is a $ \ncbc$-coderivation, $\pi$ cannot  contain  the conclusion of  $\cut_{\sq}$, $\sql$,  $\wk_{\sq}$, and  the   leftmost premise of  $\cnd_{\sq}$;
    \item \label{enum:budcompanion3} if $\der$ is a $\cbc$-coderivation,  $\pi$  cannot contain the conclusion of $\wk_{\n}$, the leftmost premise of  $\cnd_{\n}$, and the rightmost premise of $\cut_{\n}$.
\end{enumerate}
\end{prop}
\begin{proof}
By definition of cycle nf, each path from $\rd(\nu)$ to $\nu$ in $\derrd$ is contained in a branch of $\der$ such that each rule instance in the former appears  infinitely many times in the latter. Hence:
\begin{enumerate}
\item[(i)]  if $\der$ is progressing,  the path contains the conclusion of an instance of $\cnd_{\sq}$;
\item[(ii)]  if $\der$ is safe, the path   cannot contain the conclusion of a $\cut_{\sq}$ rule;
\item[(iii)] if $\der$ is left-leaning, the path    cannot contain the rightmost premise of a $\cut_{\n}$ rule.
\end{enumerate}
This shows point~\ref{enum:budcompanion0}. Let us consider  point~\ref{enum:budcompanion1}.  By point~(ii), if $\der$ is safe then,  going from a node $\mu$ of the path to each of its children $\mu'$,  the number of modal formulas in the context of the corresponding sequents cannot increase.  Moreover, the only cases where this number strictly decreases is when $\mu$ is the conclusion of $\sql$,  $\wk_{\sn}$, or  when $\mu'$ is  the leftmost premise of $\cnd_{\sq}$. Since $\rd(\nu)$ and $\nu$ must be labelled with the same sequent, all such cases are impossible. As for  point~\ref{enum:budcompanion3} we notice that,   by point~(iii) and the above reasoning, if $\der$ is safe and left-leaning then,  going from a node $\mu$ of the path to each of its children $\mu'$,  the number of non-modal formulas in the context of the corresponding sequents cannot increase.  Moreover, the only cases where this number strictly decreases is when $\mu$ is the conclusion of $\wk_{\n}$, or  when $\mu'$ is  the leftmost premise of $\cnd_{\n}$. Since $\rd(\nu)$ and $\nu$ must be labelled with the same sequent, all such cases are impossible. 
\end{proof} 

In what follows we shall indicate circularities in cycle nfs explicitly by extending both $\ncbc$ and  $\cbc$ with a new inference rule called $\com$: 
\[
\vlinf{\com}{X}{\Gamma \Rightarrow A}{\Gamma \Rightarrow A}
\]
where $X$ is a finite set of nodes. 
In this presentation, we insist that each companion $\nu$ of the cycle nf $\derrd$ is always the conclusion of an instance of $\com$, where $X$ denotes the set of buds $\nu'$ such that $\rd(\nu')=\nu$.  
{This expedient will allow us to treat the case of companions separately, avoiding repeating the same argument for each rule.}


To facilitate the translation, we shall define two disjoint  sets  $\close{\nu}$ and $\open{\nu}$. Intuitively, given a cycle nf $\derrd$ 
 and  $\nu\in T_\der$,  $\close{\nu}$ is the set of  companions above  $\nu$, while   $\open{\nu}$  is the set of  buds whose companion is strictly below   $\nu$.

\begin{defn} 
Let $\derrd$ be the  cycle normal form of a $\bcnorec$-coderivation $\der$. We define the following two sets for any $\nu\in T_\der$:
\begin{equation*}
\begin{aligned}
\close{\nu}&\dfn \{ \mu \in \rd(\bud_\nu(\der))  \ \vert \ \nu \sqsubseteq \mu \}\\
\open{\nu}&\dfn \{ \mu \in \bud_\nu(\der) \ \vert \ \rd(\mu) \sqsubset \nu \}
\end{aligned}
\end{equation*}
\end{defn}

\newcommand{\amodel}[1]{a_{#1}}

We now state a generalised version of the translation lemma and prove it in  detail. The idea, here, is to associate with each node $\hnu\in T_\der$  an instance of the scheme  $\ssnrecpp$ such that:
\begin{itemize}
    \item the scheme  simultaneously defines the  functions in  $\mathcal{C}_{\hnu}=\{ \model{\nu} \ \vert \ \nu\in \close {\hnu} \cup \{\hnu \}\}$,   with the help of an additional set of oracles $\mathcal{O}_{\hnu}=\{ \amodel{\mu} \ \vert \ \mu\in \open{\hnu}\}$
    \item if  $\mathcal{O}_{\hnu}=\emptyset$ then,    for every $f_\nu\in \mathcal{C}_{\hnu}$   it holds that $\model{\nu}= \model{\pder \nu}$, i.e.,    $f_\nu$ defines  the interpretation of $\der_{\nu}$ via a set of equations (see~\Cref{defn:semantics-coderivations}). 
\end{itemize}
In particular, when  $\hnu$ is the root of $\der$ then $\open{\hnu} = \emptyset$, and so $\model{\hnu}=\model{\pder \hnu}=\model{\der}$. Moreover, since  $\model{\hnu}$ is defined by an instance of   $\ssnrecpp$ then $\model{\der}\in \nbcpp$ by Proposition~\ref{prop:simultaneous-recursion-admissible}. As a special case,  if $\der$ is a $\cbc$-coderivation then $\model{\hnu}$ can be defined by an instance of   $\ssrecpp$, and so  $\model{\der}\in \bcpp$.

\begin{lem}[Translation Lemma, general version]\label{lem:translation}
If $\derrd$ is the cycle nf of a $\ncbc$-coderivation $\der$ and  $\hnu\in T_\der$, 
    then $\forall \nu \in \close  \hnu \cup \{ \hnu\}$:
    \begin{equation*}
     \label{eqn:translationlemmaequation1}
    \begin{aligned}
     \model{ \nu}(\vec{x}; \vec{y})&=h_{\nu}\left(( \lambda \vec{u}\subset \vec{x},\lambda \vec v. \model{ \mu}(\vec{u};\vec v))_{\mu \in \close \hnu },( \lambda \vec{u}\subseteq \vec{x},\lambda \vec v. \amodel{ \mu}(\vec{u};\vec v))_{\mu \in  \open \hnu}\right)(\vec{x}; \vec{y})
    \end{aligned}
    \end{equation*}
    where:
    \begin{enumerate}
    \item \label{enum:a} $h_\nu \in \nbcpp( (\model{\mu})_{\mu \in \close \hnu}, (\amodel \mu)_{\mu \in \open \hnu})$, and so    $\model{ \nu}\in  \nbcpp( (\amodel{ \mu})_{\mu \in \open \hnu})$;
    \item \label{enum:b}  for all $\mu \in \open \hnu$, the order $\vec u \subseteq \vec x$ in $\lambda \vec{u}\subseteq \vec{x},\lambda \vec v. \amodel{ \mu}(\vec{u};\vec v)$ is strict  if the path from $\nu$ to $\mu$ in $T_\der$ contains the conclusion of an instance of  $\cnd_{\sq}$;
    \item \label{enum:d}  if moreover   $\der$ is a $\cbc$-coderivation  then 
    \begin{equation*}
     \begin{aligned}
     \model{ \nu}(\vec{x}; \vec{y})&=h_{\nu}\left(( \lambda \vec{u}\subset \vec{x},\lambda \vec v\subseteq \vec y. \model{ \mu}(\vec{u};\vec v))_{\mu \in \close \hnu },( \lambda \vec{u}\subseteq \vec{x},\lambda \vec v\subseteq \vec y. \amodel{ \mu}(\vec{u};\vec v))_{\mu \in  \open \hnu}\right)(\vec{x}; \vec{y})
    \end{aligned}
    \end{equation*}
    with $h_\nu \in \bcpp( (\model{\mu})_{\mu \in \close \hnu}, (\amodel \mu)_{\mu \in \open \hnu})$, and so    $\model{ \nu}\in  \bcpp( (\amodel{ \mu})_{\mu \in \open \hnu})$.
    \end{enumerate}
\end{lem}

\begin{proof}
By induction on the longest distance of $\nu_0$ from a leaf of  $T_\der$. Concerning the cases where   $\nu_0$ is the conclusion of an instance of $\id$ or $\zero$ we have $\close{\hnu}= \open{\hnu}= \emptyset$, and we simply set $h_{\hnu}\dfn \id$ and $h_{\hnu}\dfn 0$ respectively. If $\hnu$ is the conclusion of a bud then $\close{\hnu}= \emptyset$ and $\open {\hnu}= \{\hnu\}$, and we simply set $h_{\hnu}\dfn \amodel{\hnu}$. The cases where   $\nu$ is the conclusion of a rule in $\{\wk_{\n},\wk_{\sq}, \sql , \exch_\n, \exch_\sq, \sqr,  \succ 0, \succ 1\}$ are straightforward. Notice that for the cases $\wk_{\n}, \wk_{\sq},\sql $ we have $\open{\hnu}=\emptyset$ by Proposition~\ref{prop:structureofcycles} points~\ref{enum:budcompanion1} and \ref{enum:budcompanion3}.

Let us now consider the case where $\hnu$ is the conclusion of an instance of   $\cut_{\sq}$ with premises $\nu_1$ and $\nu_2$. By Proposition~\ref{prop:structureofcycles}.\ref{enum:budcompanion1} we have $\open \hnu= \emptyset$. By induction hypothesis we have  $\model{ {\nu_1}}(\vec x; \vec y), \model{{\nu_2}}(\vec x, x; \vec y)\in \nbcpp$. Since the conclusion of $\pder{\nu_1}$  has modal succedent, by Proposition~\ref{prop:cutbox}  there must be a coderivation $\der^*$ such that $\model{\der^*}(\vec x;)= \model{{\nu_1}}(\vec x; \vec y)\in \nbcpp$. Moreover, by Proposition~\ref{prop:cutbox-for-left-leaning},  if $\pder{\nu_1}$ is a $\cbc$-coderivation then so too is $\der^*$.  Hence, we define $h_{ \hnu}(\vec x; \vec y)\dfn \model{{\nu_2}}(\vec x, \model{\der^*}(\vec x;); \vec y)$. It is easy to see that points~\ref{enum:a}-\ref{enum:d} hold.

Suppose now that $\hnu$ is the conclusion of an instance of $\cut_{\n}$ with premises $\nu_1$ and $\nu_2$.  Then, $\open{\hnu}=\open{\nu_1}\cup \open{\nu_2}$ and  $\close{\hnu}=\close{\nu_1}\cup \close{\nu_2}$.
By induction hypothesis on $\nu_1$ and $\nu_2$ we have:
\[
\arraycolsep=2pt
\begin{array}{rcl}
     \model{{\nu_1}}(\vec{x}; \vec{y})
    & =&
     h_{\nu_1}\left(( \lambda \vec{u}\subset \vec{x}, \lambda \vec v. \model{{\mu}}(\vec{u};\vec v))_{\mu \in \close{\nu_1}}, ( \lambda \vec{u}\subseteq \vec{x}, \lambda \vec v. \amodel{{\mu}}(\vec{u};\vec v))_{\mu \in  \open{\nu_1}}\right)(\vec{x}; \vec{y})
\\
       \model{{\nu_2}}(\vec{x};y, \vec{y})  &=& 
       h_{\nu_2}\left(( \lambda \vec{u}\subset \vec{x}, \lambda v,\vec v. \model{{\mu}}(\vec{u};v,\vec v))_{\mu \in \close{\nu_2}}, ( \lambda \vec{u}\subseteq \vec{x}, \lambda v,\vec v. \amodel{{\mu}}(\vec{u};v,\vec v))_{\mu \in \open{\nu_2}}\right)(\vec{x}; y,\vec{y})
\end{array}
\]
So that we set
\begin{equation*}
    h_{{\hnu}}(\vec{x}; \vec{y})\dfn \model{{\nu_2}}(\vec{x}; \model{{\nu_1}}(\vec{x}; \vec{y}), \vec{y}) 
\end{equation*}
Points~\ref{enum:a} and~\ref{enum:b}  hold by applying the induction hypothesis. Concerning point~\ref{enum:d}, notice that   if $\der$ is a $\cbc$-coderivation then   $\open{\nu_2}= \emptyset$ by Proposition~\ref{prop:structureofcycles}.\ref{enum:budcompanion3}.  By applying the induction hypothesis on $\nu_1$ and $\nu_2$, we have $ \model{{\nu_2}}(\vec{x};y, \vec{y})\in \bcpp$ and   $ \model{{\nu_1}}(\vec{x}; \vec{y})=h_{\nu_1}\left(( \lambda \vec{u}\subset \vec{x}, \lambda \vec v\subseteq \vec y. \model{{\mu}}(\vec{u};\vec v))_{\mu \in \close{\nu_1}},( \lambda \vec{u}\subseteq \vec{x}, \lambda \vec v\subseteq \vec y. \amodel{{\mu}}(\vec{u};\vec v))_{\mu \in  \open{\nu_1}}\right)(\vec{x}; \vec{y})$, so that:
\begin{equation*}
     \model{{\hnu}}(\vec{x}; \vec{y})=h_{\hnu}\left(( \lambda \vec{u}\subset \vec{x}, \lambda \vec v\subseteq \vec y. \model{{\mu}}(\vec{u};\vec v))_{\mu \in \close{\hnu}},( \lambda \vec{u}\subseteq \vec{x}, \lambda \vec v\subseteq \vec y. \amodel{{\mu}}(\vec{u};\vec v))_{\mu \in \open{\hnu}}\right)(\vec{x}; \vec{y})
\end{equation*}

Suppose that $\hnu$ is the conclusion of a $\cnd_\sq$ step with premises $\nu'$, $\nu_1$, and $\nu_2$. By Proposition~\ref{prop:structureofcycles}.\ref{enum:budcompanion1} we have $\open{\nu'}=\emptyset$, so that $\model{{\nu'}}\in \nbcpp$ (resp., $\model{{\nu'}}\in \bcpp$) by induction hypothesis.  By definition,  $\open \hnu=\open{\nu_1}\cup \open{\nu_2}$ and  $\close \hnu=\close{\nu_1}\cup \close{\nu_2}$. Then, we set:
\[
\arraycolsep=2pt
    \begin{array}{c}
    h_{ \hnu}(x, \vec{x}; \vec{y})
    = \cnd(;x,  \model{{\nu'}}(\vec{x}; \vec{y}), \model{{\nu_1}}(\pred(x;),\vec{x}; \vec{y}) , \model{{\nu_2}}(\pred(x;), \vec{x}; \vec{y}))
  \end{array}
  \]
  where $\pred(x;)$ can be defined from $\pred(;x)$ and projections.  By induction hypothesis on $\nu_i$:
  \[
  \arraycolsep=2pt
      \begin{array}{rcl}
       \model{{\nu_i} }(\pred(x;), \vec{x}; \vec{y})
       &=&
       h_{\nu_i}\left(( \lambda u,\vec{u}\subset \pred(x;),\vec{x}, \lambda \vec{v}. \model{{\mu}}(u,\vec{u};\vec{v}))_{\mu \in \close{\nu_i}}, ( \lambda u,\vec{u}\subseteq \pred(x;),\vec{x}, \lambda \vec{v}. \amodel{{\mu}}(u,\vec{u};\vec{v}))_{\mu \in\open{\nu_i}} \right)(\pred(x;),\vec{x}; \vec{y})\\
       &=&
       h_{\nu_i}\left(( \lambda u,\vec{u}\subset x,\vec{x}, \lambda \vec{v}. \model{{\mu}}(u,\vec{u};\vec{v}))_{\mu \in \close{\nu_i}}, ( \lambda u,\vec{u}\subset x,\vec{x}, \lambda \vec{v}. \amodel{{\mu}}(u,\vec{u};\vec{v}))_{\mu \in\open{\nu_i}} \right)(\pred(x;),\vec{x}; \vec{y})\\
      \end{array}
\]
  hence 
  \[
  \arraycolsep=2pt
  \begin{array}{c}
        \model{ \hnu}(x, \vec{x}; \vec{y})
        =
        h_{\hnu}\left(( \lambda u,\vec{u}\subset x,\vec{x}, \lambda \vec{v}. \model{{\mu}}(u,\vec{u};\vec{v}))_{\mu \in \close{\hnu}}, ( \lambda u,\vec{u}\subset x,\vec{x}, \lambda \vec{v}. \amodel{{\mu}}(u,\vec{u};\vec{v}))_{\mu \in\open{\hnu}} \right)(x,\vec{x}; \vec{y})\\
  \end{array}
  \]
  which shows in particular point~\ref{enum:b}. Points~\ref{enum:a} and~\ref{enum:d} are straightforward.

Let us now consider the case where $\hnu$ is an instance of $\cnd_\n$. By Proposition~\ref{prop:structureofcycles}.\ref{enum:budcompanion3} we have $\open{\nu'}=\emptyset$, so that $\model{{\nu'}}\in \nbcpp$ (resp., $\model{{\nu'}}\in \bcpp$) by induction hypothesis.  By definition,  $\open \hnu=\open{\nu_1}\cup \open{\nu_2}$ and  $\close \hnu=\close{\nu_1}\cup \close{\nu_2}$. Then, we set:
\[
\arraycolsep=2pt
\begin{array}{c}
 h_{ \hnu}(\vec{x};y,  \vec{y})
 = 
 \cnd(;y, \model{{\nu'}}(\vec{x}; \vec{y}), \model{{\nu_1}}(\vec{x};\pred(;y), \vec{y}) , \model{{\nu_2}}( \vec{x};\pred(;y), \vec{y}))    
\end{array}
\]
By induction hypothesis on $\nu_i$:
 \[
\arraycolsep=2pt
\begin{array}{rcl}
\model{{\nu_i} }(\vec{x}; \pred(;y), \vec{y}) 
&=&
h_{\nu_i}\left(( \lambda \vec{u}\subset x,\vec{x}, \lambda v,\vec{v}. \model{{\mu}}(\vec{u};v,\vec{v}))_{\mu \in \close{\nu_i}},( \lambda \vec{u}\subseteq x,\vec{x}, \lambda v,\vec{v}. \amodel{{\mu}}(\vec{u};v,\vec{v}))_{\mu \in  \open{\nu_i}}\right)(\vec{x};\pred(;y), \vec{y})
\end{array}
\]
and   hence 
   \[
\arraycolsep=2pt
\begin{array}{c}
\model{ \hnu}( \vec{x};y, \vec{y})
=
h_{\hnu}\left((\lambda \vec{u}\subset \vec{x}, \lambda v,\vec v. \model{{\mu}}(\vec{u};v,\vec v))_{\mu \in \close{\hnu}},(\lambda \vec{u}\subseteq \vec{x}, \lambda v,\vec v. \amodel{{\mu}}(\vec{u};v,\vec v))_{\mu \in \open{\hnu}}\right)(\vec{x};y, \vec{y})
\end{array}
\]
Points~\ref{enum:a} and~\ref{enum:b} are straightforward. Concerning point~\ref{enum:d} notice that by induction hypothesis on $\nu_i$ we have:
 \[
\arraycolsep=2pt
\begin{array}{rcl}
\model{{\nu_i} }(\vec{x}; \pred(;y), \vec{y}) 
&=&
h_{\nu_i}\left(( \lambda \vec{u}\subset x,\vec{x}, \lambda v,\vec{v}\subseteq \pred(;y), \vec{y}. \model{{\mu}}(\vec{u};v,\vec{v}))_{\mu \in \close{\nu_i}},( \lambda \vec{u}\subseteq x,\vec{x}, \lambda v,\vec{v}\subseteq \pred(;y), \vec{y}. \amodel{{\mu}}(\vec{u};v,\vec{v}))_{\mu \in  \open{\nu_i}}\right)(\vec{x};\pred(;y), \vec{y})
\end{array}
\]
and so we have:
 \[
\arraycolsep=2pt
\begin{array}{c}
\model{ \hnu}( \vec{x};y, \vec{y})
=
h_{\hnu}\left((\lambda \vec{u}\subset \vec{x}, \lambda v,\vec v\subset y, \vec{y}. \model{{\mu}}(\vec{u};v,\vec v))_{\mu \in \close{\hnu}},(\lambda \vec{u}\subseteq \vec{x}, \lambda v,\vec v\subset y, \vec{y}. \amodel{{\mu}}(\vec{u};v,\vec v))_{\mu \in \open{\hnu}}\right)(\vec{x};y, \vec{y})
\end{array}
\]

Let us finally consider the case where $\hnu$ is the conclusion of an instance of  $\com$  with premise $\nu'$, where  $X$ is the set of buds labelling the rule. We have  $\open{\hnu}= \open{\nu'} \setminus X$ and $\close{\hnu}=\close{\nu'}\cup\{\hnu \}$. We want to find $(h_\nu)_{\nu \in \close \nu \cup \{ \hnu\}}$ defining the equations for $(\model{ \nu})_{\nu \in \close \nu \cup \{ \hnu\}}$   in such a way that points~\ref{enum:a}-\ref{enum:d} hold.  First, note that, by definition of cycle nf, $\model{\pder{\hnu}}(\vec{x}; \vec{y}) =\model{\pder{\nu'}}(\vec{x}; \vec{y})=\model{\pder \mu}(\vec{x}; \vec{y})$ for all $\mu \in X$. Since $ \open{\nu'}=\open \hnu \cup X$ and  the path from $\nu'$ to any  $\mu\in X$ must cross an instance of $\cnd_{\sq}$ by  Proposition~\ref{prop:structureofcycles}.\ref{enum:budcompanion0}, by  induction hypothesis on $\nu'$ there exists a family  $(h'_{\nu})_{\nu \in \close{\nu'}\cup \{ \nu'\}}$ such that:
\begin{equation}\label{eqn:1}
\arraycolsep=2pt
\begin{array}{rcll}
     \model{{\nu'}}(\vec{x}; \vec{y})
     &=& 
      h'_{\nu'}\left(( \lambda \vec{u}\subset \vec{x}, \lambda \vec v. \model{{\mu}}(\vec{u};\vec{v}))_{\mu \in \close{\nu'}},( \lambda \vec{u}\subseteq \vec{x}, \lambda \vec v. \amodel{{\mu}}(\vec{u};\vec{v}))_{\mu \in  \open{\hnu}},( \lambda \vec{u}\subset \vec{x}, \lambda \vec v. \amodel{{\mu}}(\vec{u};\vec{v}))_{\mu \in  X}\right)(\vec{x}; \vec{y})
\end{array}
\end{equation}
and for all $  \nu \in \close{\nu'}$:
\begin{equation}\label{eqn:2}
\arraycolsep=2pt
\begin{array}{rcll}
    \model{{\nu}}(\vec{x}; \vec{y})
     &=& 
      h'_{\nu}\left(( \lambda \vec{u}\subset \vec{x}, \lambda \vec v. \model{{\mu}}(\vec{u};\vec{v}))_{\mu \in \close{\nu'}},( \lambda \vec{u}\subseteq \vec{x}, \lambda \vec v. \amodel{{\mu}}(\vec{u};\vec{v}))_{\mu \in  \open{\hnu}},( \lambda \vec{u}\subseteq \vec{x}, \lambda \vec v. \amodel{{\mu}}(\vec{u};\vec{v}))_{\mu \in  X}\right)(\vec{x}; \vec{y})
   \end{array}
\end{equation}
{we fix the oracles as $\amodel{\mu}\dfn \model{\hnu}$ for all $\mu \in X$} and set $h_{\hnu}\dfn h'_{\nu'}$ so that:
\begin{equation}\label{eqn:3}
\arraycolsep=2pt
\begin{array}{rcll}
     \model{{\hnu}}(\vec{x}; \vec{y})
     &=& 
      h_{\hnu}\left(( \lambda \vec{u}\subset \vec{x}, \lambda \vec v. \model{{\mu}}(\vec{u};\vec{v}))_{\mu \in \close{\nu'}},( \lambda \vec{u}\subseteq \vec{x}, \lambda \vec v. \amodel{{\mu}}(\vec{u};\vec{v}))_{\mu \in  \open{\hnu}}, \lambda \vec{u}\subset \vec{x}, \lambda \vec v. \model{{\hnu}}(\vec{u};\vec{v})\right)(\vec{x}; \vec{y})\\
      &=& 
      h_{\hnu}\left(( \lambda \vec{u}\subset \vec{x}, \lambda \vec v. \model{{\mu}}(\vec{u};\vec{v}))_{\mu \in \close{\nu'}\cup \{\hnu\}},( \lambda \vec{u}\subseteq \vec{x}, \lambda \vec v. \amodel{{\mu}}(\vec{u};\vec{v}))_{\mu \in  \open{\hnu}}\right)(\vec{x}; \vec{y})\\
       &=& 
      h_{\hnu}\left(( \lambda \vec{u}\subset \vec{x}, \lambda \vec v. \model{{\mu}}(\vec{u};\vec{v}))_{\mu \in \close{\hnu}},( \lambda \vec{u}\subseteq \vec{x}, \lambda \vec v. \amodel{{\mu}}(\vec{u};\vec{v}))_{\mu \in  \open{\hnu}}\right)(\vec{x}; \vec{y})
\end{array}
\end{equation}
and, for all $  \nu \in \close{\nu'}$:
\begin{equation}\label{eqn:4}
\arraycolsep=2pt
\begin{array}{rcll}
    \model{{\nu}}(\vec{x}; \vec{y})
     &=& 
      h'_{\nu}\left(( \lambda \vec{u}\subset \vec{x}, \lambda \vec v. \model{{\mu}}(\vec{u};\vec{v}))_{\mu \in \close{\nu'}},( \lambda \vec{u}\subseteq \vec{x}, \lambda \vec v. \amodel{{\mu}}(\vec{u};\vec{v}))_{\mu \in  \open{\hnu}}, \lambda \vec{u}\subseteq \vec{x}, \lambda \vec v. \model{{\hnu}}(\vec{u};\vec{v})\right)(\vec{x}; \vec{y})
   \end{array}
\end{equation}
Notice that $\vec{u}\subseteq \vec{x}$ might not be strict in $ \lambda \vec{u}\subseteq \vec{x}, \lambda \vec v. \model{{\hnu}}(\vec{u};\vec{v})$. However, by~\Cref{eqn:3} $f_{\hnu}$ can be rewritten to $  h_{\hnu}\left(( \lambda \vec{u}\subset \vec{x}, \lambda \vec v. \model{{\mu}}(\vec{u};\vec{v}))_{\mu \in \close{\hnu}},( \lambda \vec{u}\subseteq \vec{x}, \lambda \vec v. \amodel{{\mu}}(\vec{u};\vec{v}))_{\mu \in  \open{\hnu}}\right)$, and so $h'_\nu$ can be rewritten to a $h^*_\nu$ such that:
\begin{equation}\label{eqn:5}
\arraycolsep=2pt
\begin{array}{rcll}
    \model{{\nu}}(\vec{x}; \vec{y})
     &=& 
      h^*_{\nu}\left(( \lambda \vec{u}\subset \vec{x}, \lambda \vec v. \model{{\mu}}(\vec{u};\vec{v}))_{\mu \in \close{\hnu}},( \lambda \vec{u}\subseteq \vec{x}, \lambda \vec v. \amodel{{\mu}}(\vec{u};\vec{v}))_{\mu \in  \open{\hnu}}\right)(\vec{x}; \vec{y})
   \end{array}
\end{equation}
we set $h_\nu\dfn h^*_\nu$. Points~\ref{enum:a} and~\ref{enum:b} are straightforward.  Point~\ref{enum:d} follows  by noticing that the construction does not affect the safe arguments.

\end{proof}

Finally, we can establish the main result of this paper:
\begin{cor}\label{cor:main-theorem} 
We have the following:
\begin{itemize}
    \item $f(\vec x; )\in \cbc$ if and only if $f(\vec x)\in \fptime$;
    \item $f(\vec x; )\in \ncbc$ if and only if $f(\vec x)\in \felementary$.
\end{itemize}
\end{cor}
\begin{proof}
[Proof sketch]
Soundness ($\Rightarrow$) follows from \Cref{lem:short-translation-lemma}, by showing that $\bcpp$ and $\nbcpp$ are closed under \emph{simultaneous} versions of their recursion schemes.
Completeness ($\Leftarrow$) follows  from~\Cref{thm:bc-in-cbc} and \Cref{thm:nbc-in-ncbc}. 
\end{proof}

\section{Conclusions and further remarks}\label{sec:conclusion}

In this work we presented two-tiered circular type systems $\cbc$ and $\ncbc$ and showed that they capture polynomial-time and elementary computation, respectively.
This is the first time that methods of circular proof theory have been applied in implicit computational complexity (ICC).
Along the way we gave novel relativised algebras for these classes based on safe (nested) recursion on well-founded relations.

Since the conference version \cite{CurziDas:lics22} of this work was published, other works in Cyclic Implicit Complexity have appeared, building on the present work, in particular in the setting of \emph{non-uniform computation} \cite{CurziDas:csl23,AcclavioCurziGuerrieri:csl24}.

\subsection{Unary notation and linear space}
It is well-known that  $\flinspace$, i.e.~the class of functions computable in linear space, can be captured  by reformulating $\bc$ in  \emph{unary} notation (see~\cite{bellantoniTESIS}). A similar result can be obtained for  $\cbc$ by just defining a unary version of the conditional in $\bcpp$ (similarly to the ones in \cite{Das2021,Kuperberg-Pous21}) and by adapting the proofs  of~\Cref{lem:boundinglemma},~\Cref{thm:bc-in-cbc}  and~\Cref{lem:short-translation-lemma}.  On the other hand, $\ncbc$ is (unsurprisingly) not sensitive to such choice of notation.

 \subsection{On unnested recursion with compositions}
  We believe that \Cref{lem:boundinglemma} can be adapted to establish a polynomial bound on the growth rate  for the function algebra extending $\bcnorec$ by `Safe Composition During Safe Recursion', cf.~\Cref{rem:errata}. 
  We conjecture that this  function algebra and its extension to recursion on permutation of prefixes (cf.~\Cref{defn:saferec-pp}) capture precisely the class $\fpspace$. 
  Indeed, several function algebras for $\fpspace$ have been proposed in the literature, many of which involve variants of~\eqref{eq:safe-rec-w-comp-dur-rec-no-oracles} (see~\cite{Leivant-pspace,Oitavem-pspace}).
%
  These recursion schemes reflect the parallel nature of polynomial space functions, which can be defined in terms of alternating polynomial time computation. 
  We suspect that a circular proof theoretic characterisation of this class can thus be achieved by extending $\cbc$ with a `parallel' version of the  cut rule and by adapting the left-leaning criterion appropriately.
Parallel cuts might also play a fundamental role for potential circular proof theoretic characterisations of circuit complexity classes, like $\alogtime$ or $\nc$.

\subsection{Towards higher-order cyclic implicit complexity}
It would be pertinent to pursue higher-order versions of both $\ncbc$ and $\cbc$, in light of precursory works \cite{Das2021,Kuperberg-Pous21} in circular proof theory as well as ICC {\cite{Hofmann97,Leivant99,BellantoniNS00}}.
 In the case of polynomial-time, for instance,
 a soundness result for some higher-order version of $\cbc$ might follow by translation to (a sequent-style formulation of) Hofmann's $\slr$~\cite{Hofmann97}. 
 Analogous translations might be defined for a higher-order version of $\ncbc$ once the linearity restrictions on the recursion operator of $\slr$ are dropped. 
 Finally, as $\slr$ is essentially a subsystem of G\"{o}del's system $\tgodel$, such translations could refine the results on the  abstraction complexity (i.e.~type level) of the circular version of system $\tgodel$  in~\cite{Das2021,Das2021-preprint}. 
 
\paragraph{Acknowledgements}

We would like to thank {the anonymous reviewers for their valuable comments and suggestions. We are also indebted to} Dominik Wehr for pointing out a minor error in the conference version of this paper~\cite{CurziDas:lics22} (see~\Cref{rem:errata}).
This work was supported by a UKRI Future Leaders Fellowship, ‘Structure vs Invariants in Proofs’ (project reference MR/S035540/1),   by the Wallenberg Academy Fellowship Prolongation 
project `Taming Jörmungandr: The Logical Foundations of Circularity' (project reference 251080003), and by the
VR starting grant “Proofs with Cycles in Computation” (project reference 251088801).

\bibliographystyle{ACM-Reference-Format}
\bibliography{main}


\clearpage
\appendix

\section{Computational expressivity of (progressing) coderivations}
\label{sec:computational-strength}
\subsection{Examples of coderivations for~\Cref{prop:computational-props-coderivations}}

In light of~\Cref{prop:modal-nonmodal-distinction}, we shall simply omit modalities in (regular) (progressing) coderivations in what follows, i.e.~we shall regard any formula in the antecedent of a sequent as modal and we shall omit applications of $\sqr$. Consequently,  we shall write e.g.~$\cut$ instead of $\cut_{\sq}$ (and so on) and avoid writing semicolons in the semantics of a coderivation.

From here, \Cref{item:type1-complete,item:turing-complete} from \Cref{prop:computational-props-coderivations} are proved by way of the following examples. 

\begin{exmp}[Extensional completeness at type 1]\label{exmp:extensional-completeness-type-1}   For any function  $f: \Nat^k \to \Nat$ there is a progressing coderivation $\completeness$ such that $\denot\completeness=f$. Proceeding by induction on $k$, if $k=0$ then we use the rules $\zero, \succ 0, \succ 1 , \cut$ to construct $\completeness$ defining the natural number $f$. Otherwise, suppose $f:\Nat \times \Nat^k \to \Nat$ and define $f_n$ as $f_n(\vec x)=f(n, \vec x)$. We construct the coderivation defining $f$ as follows:
\[
\tiny
\vlderivation{
\vliin{\cut}{}{\n, \vec \n \seqar \n}
{
   \vltr{\binarytounary}{\n,  \vec \n \seqar \red{\n}}{\vlhy{\ }}{\vlhy{\ }}{\vlhy{\ }}
}
{
 \vliin{\cnd}{}{\red{\n}, \vec \n, \seqar \n}
  {
        \vltr{f_0}{\vec \n \seqar \n}{\vlhy{\ }}{\vlhy{\ }}{\vlhy{\ }}
  }
   {
        \vliin{\cnd}{}{\red{\n}, \vec \n \seqar \n}
  {
        \vltr{f_1}{\vec \n \seqar \n}{\vlhy{\ }}{\vlhy{\ }}{\vlhy{\ }}
  }
  {
        \vliin{\cnd}{}{\red{\n}, \vec \n \seqar \n}
  {
        \vltr{f_2}{\vec \n \seqar \n}{\vlhy{\ }}{\vlhy{\ }}{\vlhy{\ }}
  }
  {
        \vlin{\cnd}{}{\red{\n}, \vec \n \seqar \n}{\vlhy{\vdots}}
  }
    }
  }
}
}
\]
where $\binarytounary$ is a coderivation converting $n$ into $\succ 1 \overset{n}{\ldots} \succ 1\zero$, and we omit the second premise of $\cnd$ as it is never selected. Notice that $\completeness$  is progressing, as red formulas $\red{\n}$ (which are modal) form a progressing thread. 

\end{exmp}

The above example illustrates the role of regularity as a \emph{uniformity} condition:  regular coderivations admit a finite description  (e.g.~a finite tree with backpointers), so that they define  computable functions. In fact, it turns out that any (partial) recursive function  is $\bcnorec$-definable by a regular coderivation. This can be easily inferred from  the next two examples using Proposition~\ref{prop:modal-nonmodal-distinction}.

\begin{exmp}[Primitive recursion] \label{exmp:primitive-recursion} Let  $f(x,\vec x)$ be defined by primitive recursion from  $g(\vec x)$ and $h(x, \vec x,y)$. Given  coderivations $\gfunction$ and $\hfunction$ defining $g$ and $h$ respectively, we construct the following coderivation $\primrec$:
\[
\tiny
\vlderivation{
\vliin{\cut}{}{\red{\n}, \vec \n \seqar \n}
{\vltr{\binarytounary}{\red{\n} \seqar \blue{\n}}{\vlhy{\ }}{\vlhy{\ }}{\vlhy{\ }}}{ 
     \vliin{\cnd}{\bullet}{\underline{\blue{\n}},\red{\n},   \vec \n  \seqar \n}
     {
    \vltr{\gfunction}{ \vec \n \seqar \n}{\vlhy{\ }}{\vlhy{\ }}{\vlhy{\ }}
     }
     {
        \vliin{\cut}{}{\blue{\n},\red{\n},  \vec \n \seqar \n}{
                  \vltr{\predecessor}{\red{\n} \seqar  \orange{\n}}{\vlhy{\ }}{\vlhy{\ }}{\vlhy{\ }}
                 }
                 {
            \vliin{\cut}{}{\blue{\n},\orange{\n},  \vec \n \seqar \n}{
                  \vlin{\cnd}{\bullet}{\blue{\n},\orange{\n},  \vec \n \seqar \green{\n}}{\vlhy{\vdots}}
                 }{
             \vltr{\hfunction}{\orange{\n},   \vec \n, \green{\n} \seqar \n}{\vlhy{\ }}{\vlhy{\ }}{\vlhy{\ }}
                 }
                 }
    }
}
}
\]
where $\binarytounary$ is the coderivation converting $n$ into $\succ 1\overset{n}{\ldots} \succ 1 \zero$, $\predecessor$ is a coderivation defining unary predecessor, and we avoid writing the second premise of $\cnd$ as it is never selected.  Notice that $\primrec$  is progressing, as blue formulas $\blue{\n}$ (which are modal) form a progressing thread contained in the infinite branch that loops on $\bullet$. From the associated equational program we obtain:
\[
\arraycolsep=2pt
\begin{array}{rcll}
\denot{\primrec_{\epsilon}}  (x, \vec x) & = &  \denot{\primrec{1}}(\denot{\binarytounary}(x),x, \vec x)\\
    \denot{\primrec_{1}}  (0,x, \vec x) & = & g(\vec x)\\
     \denot{\primrec_{1}}  (\succ 1z, x , \vec x ) & = & h(\denot{\predecessor}(x), \vec x,  \denot{\primrec_{1}}(z,\denot{\predecessor}(x), \vec x))
\end{array}
\]
so that $\denot\primrec=\denot{\primrec{\epsilon}}=f$.
\end{exmp}

\begin{exmp}
[Unbounded search] Let $g(x,\vec{x})$ be a function, and let $f(\vec x)\dfn \mu x. (g(x, \vec x)=0)$  be the unbounded search function obtained by applying the minimisation operation on $g(\vec x)$. Given  a coderivation $\gfunction$  defining $g$, we construct the following coderivation $\unbounded$:
\[
\tiny
\vlderivation{
\vliin{\cut}{}{ \vec \n \seqar \n}{
    \vlin{\zero}{}{\seqar\n }{\vlhy{\ }}
}{
\vliin{\cut}{\bullet}{\blue{\n},  \vec \n \seqar \n}
{
\vltr{\gfunction}{\blue{\n},  \vec \n \seqar \red{\n}}{\vlhy{\ }}{\vlhy{\ }}{\vlhy{\ }}
}
{
     \vliin{\cnd}{}{\red{\underline{\n}},\blue{\n},  \vec \n \seqar \n}{
    \vlin{\id}{}{\blue{\n} \seqar \n}{\vlhy{\ }}
     }{ 
         \vliin{\cut}{}{\blue{\n},  \vec \n \seqar \n}{
               \vltr{\successor}{
                \blue{\n} \seqar \orange{\n}
                 }{
                 \vlhy{\ }
                }{
                \vlhy{\ }
                }
                {\vlhy{\ }
                }
      }{
      \vlin{\cut}{\bullet}{\orange{\n},   \vec \n \seqar \n}{\vlhy{\vdots}}
      }
     }
}
{
}
}
}
\]
where the coderivation $\binarytounary$  computes the unary successor, and we identify the sub-coderivations corresponding to the second and the third premises of the conditional rule.   It is easy to check that the above coderivation is regular but not progressing, as  threads containing  principal formulas for $\cnd$ are finite.  From the associated equational program we obtain:
\[
\arraycolsep=2pt
\begin{array}{rcll}
 \denot{\unbounded_{\epsilon}}  (\vec x) & = & \denot{\unbounded_{1}}(0, \vec x) \\
  \denot{\unbounded_{1}}  (x,\vec x) & = & \denot{\unbounded_{11}}(\denot{\gfunction}(x,\vec{x}) ,x, \vec x) \\
   \denot{\unbounded_{11}}(0,x, \vec x) & = & x \\
   \denot{\unbounded_{11}}(\succ 0z,x, \vec x) & = &\denot{\unbounded_{1}}( \denot\successor(x), \vec x) & z \neq 0 \\
   \denot{\unbounded_{11}}(\succ 1z,x, \vec x) & = & \denot{\unbounded_{1}}( \denot\successor(x), \vec x)
\end{array}
\]
Which searches for the least $x \geq 0$ such that $g(x,\vec x)=0$. Hence, $\denot{\unbounded_\epsilon}(\vec x)=\denot{\unbounded}(\vec x)=f(\vec x)$.
\end{exmp}

\subsection{Completeness for type 1 primitive recursion}

\cite{Kuperberg-Pous21} shows that, in the absence of contraction rules, \emph{only} the primitive recursive functions are so definable (even when using arbitrary finite types).
It is tempting, therefore, to conjecture that the regular and progressing $\bcnorec$-coderivations define \emph{just} the primitive recursive functions.

However there is a crucial difference between our formulation of cut and that in \cite{Kuperberg-Pous21}, namely that ours is context-sharing and theirs is context-splitting.
Thus the former admits a quite controlled form of contraction that, perhaps surprisingly, is enough to simulate the type 0 fragment $\ct_0$ from \cite{Das2021} (which has explicit contraction).

\begin{proof}[Proof of \Cref{item:type1-prim-rec} from \Cref{prop:computational-props-coderivations}]
First, by~\Cref{prop:modal-nonmodal-distinction} we can neglect modalities in $\bcnorec$-coderivations (and semicolons in the corresponding semantics). The left-right implication thus follows from the natural inclusion of our system into $\ct_0$ and \cite[Corollary 80]{Das2021}. 

Concerning the right-left implication, 
we employ a formulation $\ntgodel 1 (\vec a)$ of the type 1 functions of $\ntgodel 1$ over oracles $\vec a$ as follows. 
$\ntgodel 1(\vec a)$ is defined just like the primitive recursive functions, including oracles $\vec a$ as initial functions, and by adding the following version of type $1$ recursion:
  \begin{itemize}
      \item if $g(\vec x) \in \ntgodel 1 (\vec a)$ and $h(a)(x,\vec x) \in \ntgodel 1 (a, \vec a)$, then the $f(x, \vec x)$ given by,
    \[
    \arraycolsep=2pt
    \begin{array}{rcl}
        f(0, \vec x) & = & g(\vec x) \\
        f( \succ i x , \vec x) & = & h_i (\lambda \vec u . f(x, \vec u)) (x, \vec x)
    \end{array}
    \]
    is in $\ntgodel 1 (\vec a)$.
  \end{itemize}
  It is not hard to see that the type 1 functions of $\ntgodel 1$ are precisely those of $\ntgodel 1 (\emptyset)$ (e.g.\ because of \cite[Appendix A]{Das2021-preprint}).  
  We now proceed by showing how to define the above scheme by regular and progressing $\bcnorec$-coderivations.




Given a function $f(\vec x) \in \ntgodel 1 (\vec a)$, we construct a regular progressing coderivation of $\bcnorec$,
\[
\vltreeder{\der_f (\vec a)}{ \vec \n \seqar \n}{}{ \left\{
\vlinf{a_i}{}{ \vec \n \seqar \n}{}
\right\}_i }{}
\]
computing $f(\vec x)$ over $\vec a$, by induction on the definition of $f(\vec x)$.

If $f(\vec x)$ is an initial function,  an oracle, or $f(\vec x) = h(g(\vec x), \vec x)$ then the construction is easy (see, e.g., the proof of~\Cref{thm:nbc-in-ncbc}).


Let us  consider the case of recursion (which subsumes usual primitive recursion at type 0).
Suppose $f(x, \vec x) \in \ntgodel 1 (\vec a)$ where,
\[
\arraycolsep=2pt
\begin{array}{rcl}
     f(0, \vec x) & = & g(\vec x) \\
     f(\succ 0 x , \vec x) & = & h_0(\lambda \vec u . f(x,\vec u) ) (x, \vec x) \\
     f(\succ 1 x , \vec x) & = & h_1(\lambda \vec u . f(x,\vec u) ) (x, \vec x)
\end{array}
\]
where $\der_g(\vec a)$ and $\der_h(a, \vec a)$ are already obtained by the inductive hypothesis.
We define $\der_f(\vec a)$ as follows:
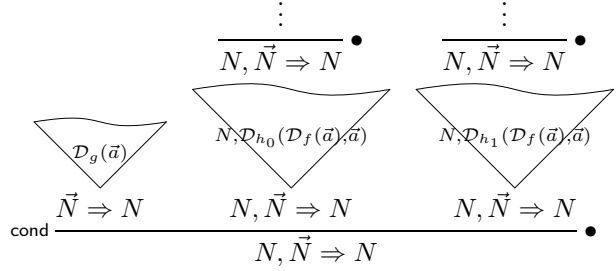
\[
\vlderivation{
    \vliiin{\cnd}{\bullet}{\n,  \vec \n \seqar \n }{
        \vltr{\der_g (\vec a)}{ \vec \n \seqar \n}{\vlhy{\quad }}{\vlhy{\quad }}{\vlhy{\quad }}
    }{
        \vltr{\n, \der_{h_0}(\der_f(\vec a),\vec a)}{\n ,  \vec \n \seqar \n}{\vlhy{}}{
            \vlin{}{\bullet}{\n,  \vec \n \seqar \n}{\vlhy{\vdots}}
        }{\vlhy{}}
    }{
        \vltr{\n, \der_{h_1}(\der_f(\vec a),\vec a)}{\n ,  \vec \n \seqar \n}{\vlhy{}}{
            \vlin{}{\bullet}{\n,  \vec \n \seqar \n}{\vlhy{\vdots}}
        }{\vlhy{}}
    }
}
\]
Note that the existence of the second and third coderivations above the conditional is given by a similar construction to that of \Cref{lem:to-prove-ncbc-complete-for-nbc} (see also \cite[Lemma 42]{Das2021}).
\end{proof}

\subsection{On the `power' of contraction}
Given the computationally equivalent system $\ct_0$ with contraction from \cite{Das2021}, we can view the above result as a sort of `contraction admissibility' for regular progressing $\bcnorec$-coderivations.
Let us take a moment to make this formal.

Call $\bcnorec+ \{\cntr_\n, \cntr_\sn \}$ the extension of $\bcnorec$ with the rules $\cntr_\n$ and $\cntr_\sn$ below:
\begin{equation} \label{eqn:explicit-contraction}
  \vlinf{\cntr_\n}{}{\orange \Gamma, \blue \n\seqar B}{\orange \Gamma, \blue \n , \blue \n \seqar B}
\qquad 
\vlinf{\cntr_\sq}{}{\orange \Gamma, \blue \sn \seqar B}{\orange \Gamma, \blue \sn , \blue \sn \seqar B}
\end{equation}
 where  the semantics for the new system extends the one for $\bcnorec$ in the obvious way, and the notion of (progressing) thread is induced by the given colouring.\footnote{Note that the totality argument of \Cref{prop:prog-total} still applies in the presence of these rules, cf.~also \cite{Das2021}.} 
We have:
\begin{cor} $f(\vec x; )$ is definable by a regular progressing  $\bcnorec+ \{\cntr_\n, \cntr_\sn \}$-coderivation  iff it is definable by a regular progressing $\bcnorec$-coderivation.
\end{cor}
\begin{proof}
[Proof idea]
The former system is equivalent to $\ct_0$ from \cite{Das2021}, whose type 1 functions are just those of $\ntgodel 1$ by \cite[Corollary 80]{Das2021}, which are all defined by regular progressing $\bcnorec$-coderivations by \Cref{prop:computational-props-coderivations} \Cref{item:type1-prim-rec}.
\end{proof}

\begin{rem}\label{rem:explicit-contraction}
The reader may at this point wonder if a  direct  ‘contraction-admissibility' argument exists for the rules in~\eqref{eqn:explicit-contraction}. First, notice that  $\cntr_\n$ can be derived in $\bcnorec$:
\[
     \vlderivation{
     \vliin{\cut_\n}{}{\Gamma, \blue \n  \seqar B}
{
   \vliq{\wk_\n}{}{\Gamma, \n \seqar \red \n}{
   \vlin{\id}{}{ \n \seqar \n}{\vlhy{}}}
}
{ 
 \vlhy{\Gamma, \blue \n, \red \n   \seqar B  }
}
}
\]
While a similar derivation exists for $\cntr_{\sq}$,
note that this crucially does not preserve the same notion of thread (cf.~colours above) and so does not, a priori, preserve progressiveness.
\end{rem}

\begin{exmp}
[Ackermann-P\'eter]
As suggested by \Cref{prop:computational-props-coderivations} \Cref{item:type1-prim-rec}, regular progressing $\bcnorec$ coderivations are able to define the Ackermann-P\'eter function $A(x,y)$,
\[
\arraycolsep=2pt
\begin{array}{rcl}
    A(0,y) & = & x+1 \\
    A(x+1,0) & = & A(x,1) \\
    A(x+1,y+1) & = & A(x,A(x+1,y))
\end{array}
\]
which is well-known to not be (type 0) primitive recursive, but is nevertheless type 1 primitive recursive.

The usual cyclic representation of $A(x,y)$ using only base types, e.g.\ as in \cite{Das2021,Kuperberg-Pous21}, mimics the equational program above, namely applying case analysis on the first input before applying case analysis on the second input.
It is important in this construction to be able to explicitly contract the first input, corresponding to $x$, due to the third line of the equational program above.

In our context-sharing setting, without explicit contraction, we may nonetheless represent $A(x,y)$ by conducting a case analysis on the second input $y$ first, and including some redundancy, as in the regular progressing coderivation in \Cref{fig:ack-pet}.
To facilitate readability, we again omit modalities and work purely in unary notation, representing a natural number $n$ by $\succ1 \overset{n}\ldots \succ 10$.

\begin{figure*}
\centering
\begin{equation*}
\begin{array}{c}
\hspace{8em}
\vlderivation{
\vliin{\cnd}{\bullet}{\red{\n},\underline{\blue{\n}}\seqar \n}{
    \vliin{\cnd}{}{\underline{\red{\n}}\seqar \n}{
        \vliq{1}{}{\seqar \n}{\vlhy{}}
    }{
        \vlin{1}{}{\red{\n} \seqar \n}{
        \vlin{\cnd}{\bullet}{\red{\n},\underline{\n}\seqar\n}{\vlhy{\vdots (1)}}
        }
    }
}{
    \vliin{\cut}{}{\red{\n},\blue{\n} \seqar \n}{
        \vlin{\cnd}{\bullet}{\red{\n},\underline{\blue{\n}}\seqar \n}{\vlhy{\vdots(2)}}
    }{
        \vliin{\cnd}{}{\underline{\red{\n}},\blue{\n},\n \seqar \n}{
            \vliq{+2}{}{\blue{\n} \seqar \n}{\vlhy{}}
        }{
            \vlin{\wk}{}{\red{\n},\underline{\blue{\n}},\n \seqar \n}{
            \vlin{\cnd}{\bullet}{\red{\n},\underline{\n} \seqar \n}{\vlhy{\vdots (3)}}
            }
        }
    }
}
}
\end{array}
\end{equation*}
\caption{Ackermann-P\'eter function computed by a regular progressing $\bcnorec$-coderivation (without explicit contractions).}
\label{fig:ack-pet}
\end{figure*}

The verification of the semantics of this coderivation is routine. 
To see that it is progressing, we may conduct a case analysis on the set $L$ of infinitely often visited loops, among $(1),(2),(3)$, in an arbitrary branch:
\begin{itemize}
    \item if $L = \emptyset$ then the branch is finite;
    \item if $L=\{(1)\}$ then there is a progressing thread on the red $\red\n$;
    \item if $L=\{(2)\}$ then there is a progressing thread on the blue $\blue \n$ (note that the red $\red \n $ thread does not progress);
    \item if $L=\{(3)\}$ then there is a progressing thread on the red $\red \n$;
    \item if $L=\{(1),(2)\}$ then there is a progressing thread on the red $\red\n$: on each iteration of $(1)$ the thread progresses, and remains intact (but does not progress) on each iteration of $(2)$;
    \item if $L=\{(1),(3)\}$ then there is a progressing thread on the red $\red \n$: the thread progresses on each iteration of either $(1)$ or $(3)$;
    \item if $L=\{(2),(3)\}$ then there is a progressing thread on the red $\red\n$: on each iteration of $(2)$ the thread remains intact (but does not progress), and on each iteration of $(3)$ the thread progresses;
    \item if $L=\{(1),(2),(3)\}$ then there is a progressing thread on the red $\red\n$: the thread progresses on each iteration of $(1)$ or $(3)$, and remains intact (but does not progress) on each iteration of $(3)$.
\end{itemize}
\end{exmp}

\end{document}